\documentclass[12pt]{article}

\usepackage{multirow}
\usepackage{xypic,geometry,latexsym,amssymb,amsthm,amsmath,amsfonts,verbatim,cite}
\usepackage[numbers,sectionbib,sort&compress,comma,square]{natbib}
\usepackage[pdftex]{hyperref}
\usepackage[caption=false]{subfig}
\usepackage{graphicx}
\usepackage{color}
\usepackage{xspace}
\usepackage{authblk}
\usepackage[ruled,vlined,linesnumbered]{algorithm2e}

\usepackage[utf8]{inputenc}

\newtheorem{theorem}{Theorem}

\newtheorem{lemma}[theorem]{Lemma}
\newtheorem{corollary}[theorem]{Corollary}

\usepackage{lineno,setspace}

\newtheorem{probaux}{Problem}
\newenvironment{prob}[3]{\bigskip\noindent\framebox{\parbox{16.6cm}{\begin{probaux}{\sc
#1}\\{\hspace*{1cm} \bf \sf Instance:} #2\\{\hspace*{1cm} \bf \sf Question:}
#3\end{probaux}}}\bigskip}{}

\newcommand{\ka}{\mathcal{K}_1\xspace}
\newcommand{\kb}{\mathcal{K}_2\xspace}
\newcommand{\kc}{\mathcal{K}_3\xspace}
\newcommand{\kd}{\mathcal{K}_4\xspace}

\newif\ifconference
\conferencefalse

\begin{document}


\title{Hierarchical complexity of 2-clique-colouring weakly chordal graphs and
perfect graphs having cliques of size at least~3\thanks{Partially supported
by CNPq and FAPERJ.}}

\newcommand*\samethanks[1][\value{footnote}]{\footnotemark[#1]}

\author[1]{H\'elio B. Mac\^edo Filho\thanks{$\{$\removelastskip\href{mailto:helio@cos.ufrj.br}{helio}, \href{mailto:celina@cos.ufrj.br}{celina}\removelastskip$\}$@cos.ufrj.br}}
\author[2]{Raphael C. S. Machado\thanks{\href{mailto:rcmachado@inmetro.gov.br}{rcmachado@inmetro.gov.br}}}
\renewcommand\Authands{,\authorcr}
\author[1]{Celina M. H. de Figueiredo\samethanks[1]}
\affil[1]{COPPE, Universidade Federal do Rio de Janeiro}
\affil[2]{Inmetro --- Instituto Nacional de Metrologia, Qualidade e Tecnologia}

\date{}

\maketitle

\let\thefootnote\relax\footnotetext{
\itshape An extended abstract of this work was accepted for presentation at Latin 2014, the
11th Latin American Symposium on Theoretical Informatics. 
\hfill\today
}

\begin{abstract}
A clique of a graph is a maximal set of vertices of size at least 2
that induces a complete graph. A $k$-clique-colouring of a graph is a
colouring of the vertices with at most $k$ colours such that no clique is
monochromatic. D\'efossez proved that the 2-clique-colouring of perfect graphs
is a $\Sigma_2^P$-complete problem~[J. Graph Theory 62 (2009) 139--156].
We strengthen this result by showing that it is still $\Sigma_2^P$-complete for
weakly chordal graphs.
We then determine a hierarchy of nested subclasses of
weakly chordal graphs whereby each graph class is in a distinct complexity class,
namely $\Sigma_2^P$-complete, $\mathcal{NP}$-complete, and $\mathcal{P}$. We
solve an open problem posed by Kratochv\'il and Tuza to determine the
complexity of 2-clique-colouring of perfect graphs with all cliques having size
at least~3~[J. Algorithms 45 (2002), 40--54], proving that it is a
$\Sigma_2^P$-complete problem.
We then determine a hierarchy of nested subclasses of
perfect graphs with all cliques having size at least~3 whereby each graph class
is in a distinct complexity class,
namely $\Sigma_2^P$-complete, $\mathcal{NP}$-complete, and $\mathcal{P}$.

\end{abstract}

\section{Introduction}
\label{sec:introduction}
Let $G=(V,E)$ be a simple graph with $n=|V|$ vertices and
$m=|E|$ edges. A \emph{clique} of $G$ is a maximal set of vertices of size
at least~2 that induces a complete graph. A \emph{$k$-clique-colouring}
of a graph is a colouring of the vertices with at most $k$ colours such that no
clique is monochromatic. Any undefined notation concerning complexity classes
follows that of Marx~\cite{Marx2011}.

A \emph{cycle} is sequence of vertices starting and ending at the same vertex,
with each two consecutive vertices in the sequence adjacent to each other in the
graph. A \emph{chord} of a cycle is an edge joining two nodes that are not
consecutive in the cycle.

The \emph{clique-number} $\omega(G)$ of a graph $G$ is the number of vertices of
a clique with the largest possible size in $G$. A \emph{perfect graph} is a
graph in which every induced subgraph $H$ needs exactly $\omega(H)$ colours in
its vertices such that no $K_2$ (not necessarily clique) is monochromatic. The
celebrated \emph{Strong Perfect Graph Theorem} of Chudnovsky et
al.~\cite{MR2233847} says that a graph is perfect if neither it nor its
complement contains a chordless cycle with an odd number of vertices greater
than 4. A graph is \emph{chordal} if it does not contain a chordless cycle with
a number of vertices greater than 3, and a graph is \emph{weakly chordal} if
neither it nor its complement contains a chordless cycle with a number of
vertices greater than 4.

Both clique-colouring and perfect graphs have attracted much
attention due to a conjecture posed by Duffus et al.~\cite{MR1101751} that
\emph{perfect graphs are $k$-clique-colourable for some constant $k$}. 
This conjecture has not yet been proved. Following the chronological order,
Kratochv\'il and Tuza gave a framework to argue that 2-clique-colouring is
$\mathcal{NP}$-hard 
and proved that 2-clique-colouring is $\mathcal{NP}$-complete
for $K_4$-free perfect graphs~\cite{Kratochvil}. Notice that $K_3$-free perfect
graphs are bipartite graphs, which are clearly 2-clique-colourable.
Moreover, 2-clique-colouring is in $\Sigma_2^P$, since it is co$\mathcal{NP}$ to
check that a colouring of the vertices is a clique-colouring. A few years later,
the $2$-clique-colouring problem was proved to be a $\Sigma_2^P$-complete problem
by Marx~\cite{Marx2011}, a major breakthrough in the clique-colouring area.
D\'efossez~\cite{DefossezOddHoleFreeCliqueColouringComplexity} proved later
that 2-clique-colouring of perfect graphs remained a $\Sigma_2^P$-complete
problem.

When restricted to chordal graphs, 2-clique-colouring is in
$\mathcal{P}$, since all chordal graphs are
2-clique-colourable~\cite{poon}. Notice that chordal graphs are a subclass of
weakly chordal graphs, while perfect graphs are a superclass of weakly chordal graphs.
In constrast to chordal graphs, not all weakly chordal graphs are
2-clique-colourable (see Fig.~\ref{fig:not2ccweaklychordal}).

We show that 2-clique-colouring of weakly chordal graphs is a
$\Sigma_2^P$-complete problem, improving the proof
of D\'efossez~\cite{DefossezOddHoleFreeCliqueColouringComplexity} that
2-clique-colouring is a $\Sigma_2^P$-complete problem for perfect graphs.
As a remark,
D\'efossez~\cite{DefossezOddHoleFreeCliqueColouringComplexity} constructed
a graph which is not a weakly chordal graph as long as it has chordless cycles
with even number of vertices greater than 5 as induced subgraphs.
We determine a hierarchy of nested subclasses of weakly
chordal graphs whereby each graph class is in a distinct complexity class,
namely $\Sigma_2^P$-complete, $\mathcal{NP}$-complete, and $\mathcal{P}$.

A graph is \emph{($\alpha, \beta$)-polar} if there
exists a partition of its vertex set into two sets $A$ and $B$ such that all
connected components of the subgraph induced by $A$ and of the complementary
subgraph induced by $B$ are complete graphs. Moreover, the order of each
connected component of the subgraph induced by $A$ (resp. of the complementary
subgraph induced by $B$) is upper bounded by $\alpha$ (resp. upper bounded by
$\beta$)~\cite{MR863038}. A \emph{satellite} of an ($\alpha,
\beta$)-polar graph is a connected component of the
subgraph induced by $A$ (see Fig.~\ref{fig:abpolar}). 
In this work, we restrict ourselves
to the \emph{($\alpha, \beta$)-polar} graphs with $\beta = 1$, so the subgraph
induced by $B$ is complete and the order of each satellite is upper bounded by
$\alpha$ (see Fig.~\ref{fig:21polar}). 
Clearly, ($\alpha, 1$)-polar graphs are perfect, since they do not contain
chordless cycles with an odd number of vertices greater than 4 nor their
complements.

A \emph{generalized split graph} is a graph $G$ such that $G$ or its
complement is an ($\infty, 1$)-polar graph~\cite{MR1167295}.
See Fig.~\ref{fig:21polar} for an example of a generalized split graph, which is
a (2, 1)-polar graph. 
The class of generalized split graphs plays an important role in the areas of
perfect graphs and clique-colouring.
This class was introduced by Pr\"omel and Steger~\cite{MR1167295} 
to show that 
the strong perfect graph conjecture is at least asymptotically
true by proving that almost all $C_5$-free graphs are generalized
split graphs. Approximately 14 years later the strong perfect graph conjecture
became the \emph{Strong Perfect Graph Theorem} by Chudnovsky et
al.~\cite{MR2233847}. Regarding clique-colouring, Bacs\'o et
al.~\cite{MR20506790} proved that generalized split graphs are 3-clique-colourable and concluded that almost all perfect
graphs are 3-clique-colourable~\cite{MR20506790}. This conclusion
supports the conjecture due to Duffus et al.~\cite{MR1101751}. In fact,
there is no example of a perfect graph where more than three colors would be
necessary to clique-colour.
Surprisingly, after more than 20 years, relatively little progress has
been made on the conjecture.

\begin{figure}[t!]
\centering
	\subfloat
		[A weakly chordal graph with an optimal 3-clique-colouring] {
			\includegraphics[scale=0.3]{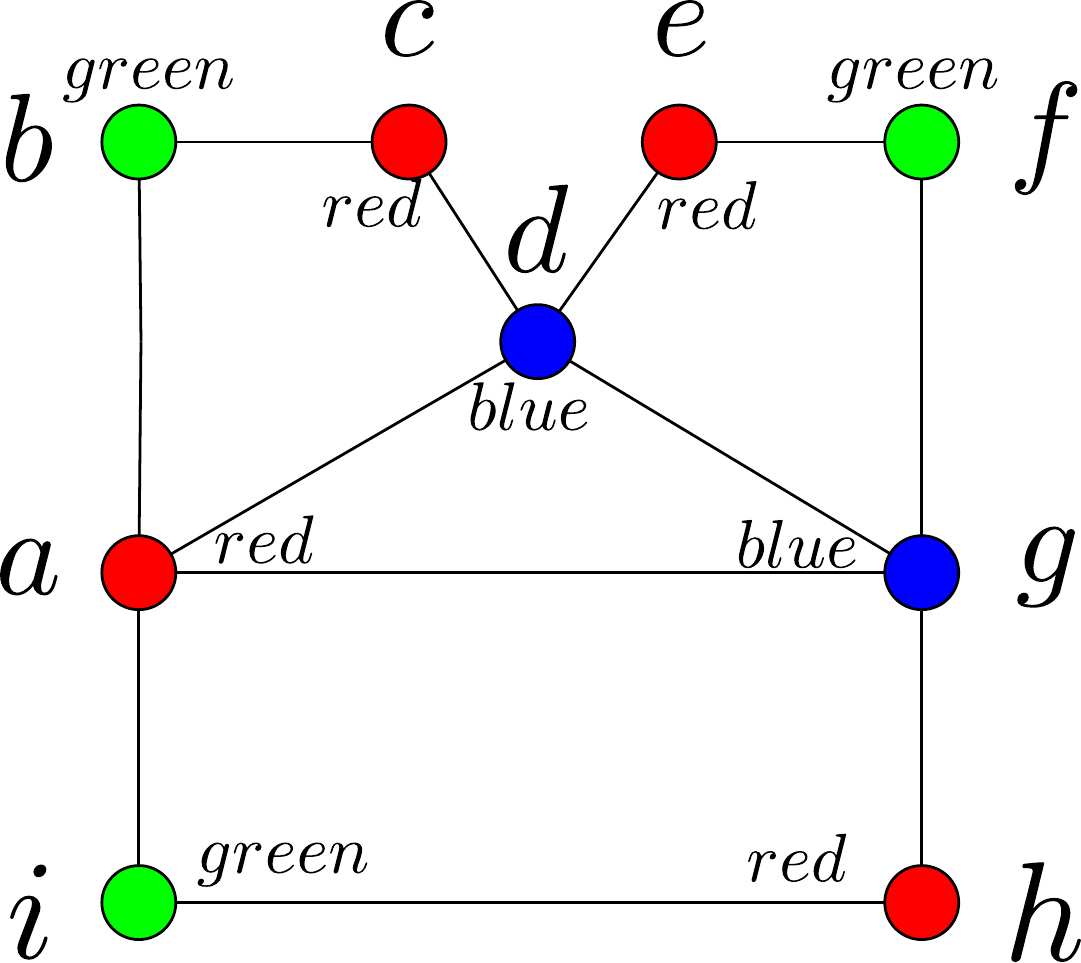}
			\label{fig:not2ccweaklychordal}
		}
	\qquad 
	\subfloat
		[A ($4, 3$)-polar graph] {
			\includegraphics[scale=0.18]{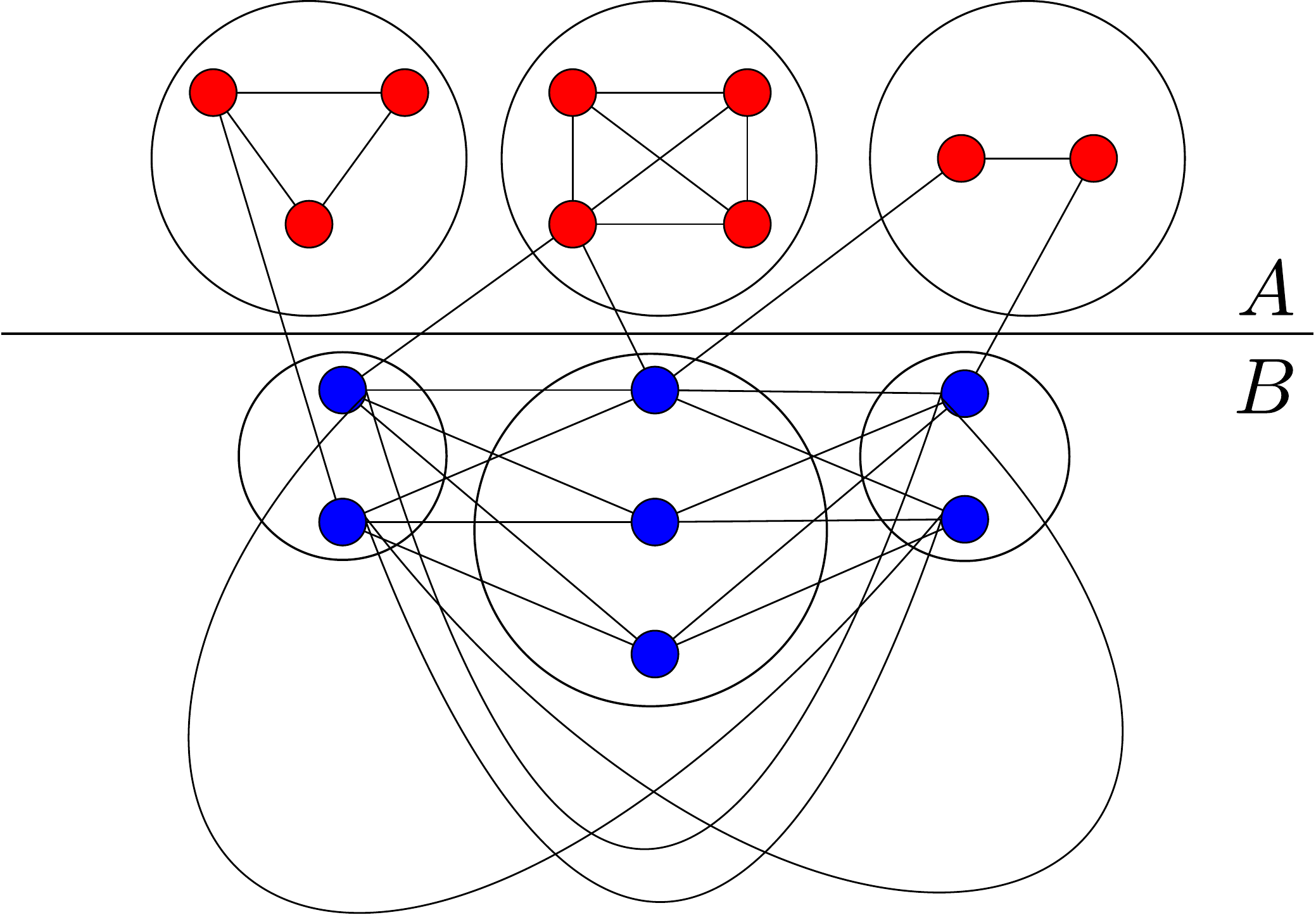}
			\label{fig:abpolar}
		}
	\qquad 
	\subfloat
		[A generalized split graph, which is a ($k, 1$)-polar graph, for
	fixed $k \geq 2$] {
			\includegraphics[scale=0.18]{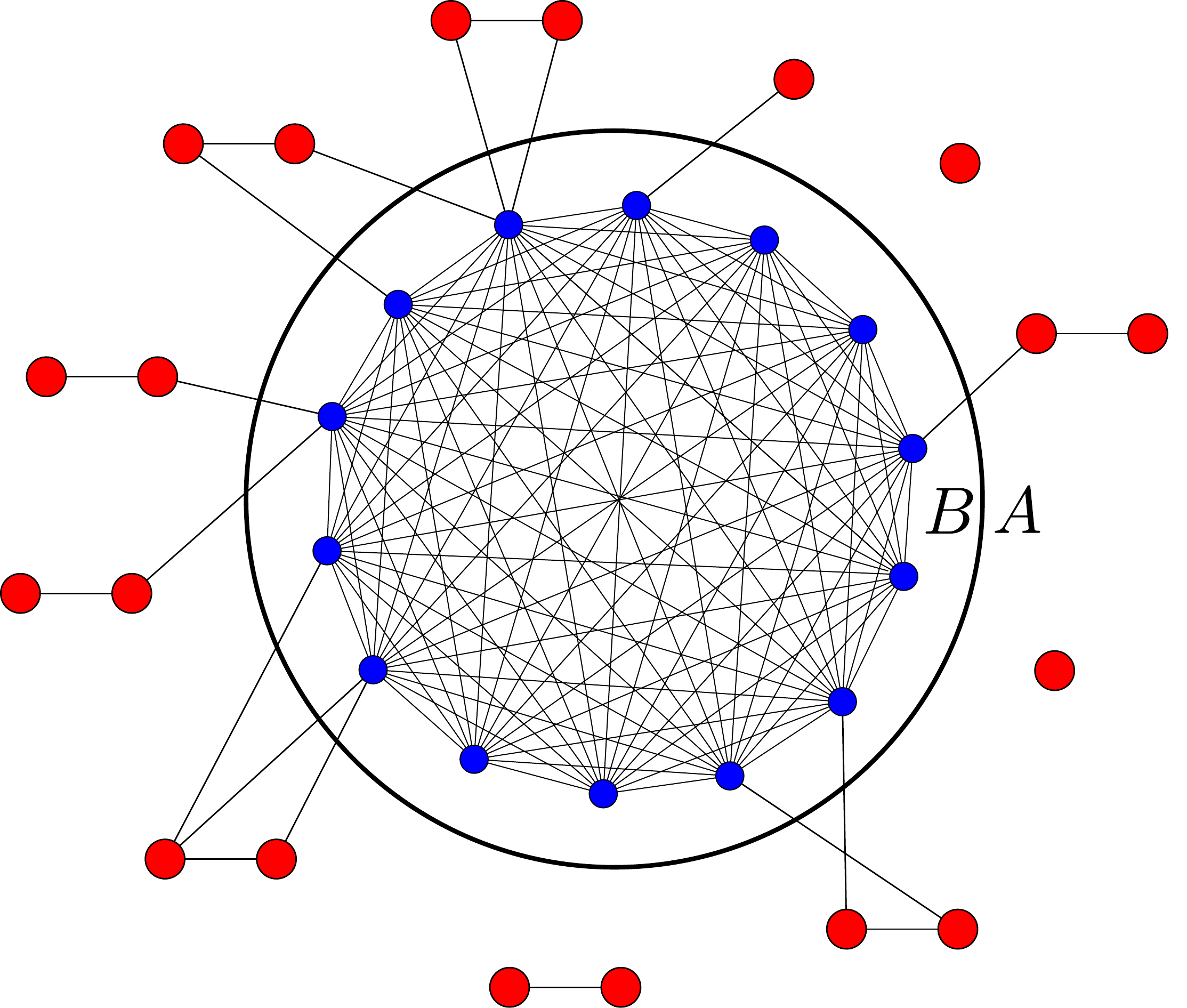}
			\label{fig:21polar}
		}
	\caption{Examples of ($\alpha, \beta$)-polar graphs}
\end{figure}

The class of ($k, 1$)-polar graphs, for fixed $k \geq 3$, is incomparable to the
class of weakly chordal graphs. Indeed, a chordless path with seven
vertices $P_7$ and a complement of a chordless cycle with six vertices
$\overline{C_6}$ are witnesses. Nevertheless, (2, 1)-polar graphs are a
subclass of weakly chordal graphs, since they do not contain a chordless cycle
with an even number of vertices greater than 5. We show that 2-clique-colouring
of (2, 1)-polar graphs is a $\mathcal{NP}$-complete problem. Finally, the class
of (1, 1)-polar graphs is precisely the class of split graphs. It is interesting
to recall that 2-clique-colouring of ($1, 1$)-polar graphs is in $\mathcal{P}$,
since ($1, 1$)-polar are a subclass of chordal graphs, which are
2-clique-colourable.


Giving continuity to our results, we investigate an open problem left by
Kratochv\'il and Tuza~\cite{Kratochvil} to determine the complexity of
2-clique-colouring of perfect graphs with all cliques having size at least 3.
Restricting the size of the cliques to be at least 3, we first show that
2-clique-colouring is still $\mathcal{NP}$-complete for (3, 1)-polar graphs,
even if it is restricted to weakly chordal graphs with all cliques
having size at least 3. Subsequently, we prove that the 2-clique-colouring of
(2, 1)-polar graphs becomes polynomial when all cliques have size at least 3.
Recall that the 2-clique-colouring of (2, 1)-polar graphs is
$\mathcal{NP}$-complete when there are no restrictions on the size of the
cliques.

We finish the paper answering the open problem of determining the complexity of
2-clique-colouring of perfect graphs with all cliques having size at least
3~\cite{Kratochvil}, by improving our proof that 2-clique-colouring is a
$\Sigma_2^P$-complete problem for weakly chordal graphs. We replace
each $K_2$ clique by a gadget with no clique of size 2, which forces
distinct colours into two given vertices.
 
The paper is organized as follows. In Section~\ref{sec:weaklychordal}, we show
that 2-clique-colouring is still $\Sigma_2^P$-complete for weakly chordal graphs. 
We then determine a hierarchy
of nested subclasses of weakly chordal graphs whereby each graph class is in a
distinct complexity class, namely $\Sigma_2^P$-complete,
$\mathcal{NP}$-complete, and $\mathcal{P}$. In
Section~\ref{sec:restrictingthesize}, we determine the complexity of
2-clique-colouring of perfect graphs with all cliques having size at
least~3, answering a question of Kratochv\'il and Tuza~\cite{Kratochvil}. We
then determine a hierarchy of nested subclasses of perfect graphs with all
cliques having size at least 3 whereby each graph class is in a distinct
complexity class. 
We refer the reader to Table~\ref{t:tabela} for our results and related work
about 2-clique-colouring complexity of perfect graphs.
\begin{table}[b!]
\caption{2-clique-colouring complexity of perfect graphs and subclasses.}
\centering
\begin{tabular}{|c|c|c||c|}
\hline
\multicolumn{3}{|c||}{\bf Class} & {\bf 2-clique-colouring complexity}\\
\hline\hline
\multirow{9}{*}{-} & \multirow{4}{*}{Perfect} & - & $\Sigma_2^P$-complete
\cite{DefossezOddHoleFreeCliqueColouringComplexity}\\
\cline{3-4}
 & & $K_4$-free & $\mathcal{NP}$-complete~\cite{Kratochvil} \\
\cline{3-4}
 & & $K_3$-free & \multirow{2}{*}{$\mathcal{P}$} \\
 & & (Bipartite) & \\
\cline{2-4}
 & Weakly chordal & - & $\Sigma_2^P$-complete \\
\cline{2-4}
 & (3, 1)-polar & - & \multirow{2}{*}{$\mathcal{NP}$-complete} \\
\cline{2-3}
 & (2, 1)-polar & - & \\
\cline{2-4}
 & Chordal & \multirow{2}{*}{-} & \multirow{2}{*}{$\mathcal{P}$~\cite{poon}} \\
 & (includes Split) & & \\
\hline
\multirow{4}{*}{\parbox{2cm}{\centering All cliques \\ having size \\ at least 3}}
 & Perfect & - & \multirow{2}{*}{$\Sigma_2^P$-complete}\\
\cline{2-3}
 & \multirow{2}{*}{Weakly chordal} & - & \\
\cline{3-4}
 & & (3, 1)-polar & $\mathcal{NP}$-complete \\
\cline{2-4}
 & (2, 1)-polar & - & $\mathcal{P}$ \\
\hline
\end{tabular}
\label{t:tabela}
\end{table}

\section{Hierarchical complexity of 2-clique-colouring of weakly chordal graphs}
\label{sec:weaklychordal}

D\'efossez proved that 2-clique-colouring of perfect graphs
is a $\Sigma_2^P$-complete
problem~\cite{DefossezOddHoleFreeCliqueColouringComplexity}.
In this section, we strengthen this result by showing that it is still
$\Sigma_2^P$-complete for weakly chordal graphs. We show a subclass of perfect
graphs (resp. of weakly chordal graphs) in which 2-clique-colouring is neither
a $\Sigma_2^P$-complete problem nor in $\mathcal{P}$, namely ($3, 1$)-polar
graphs (resp. ($2, 1$)-polar graphs). Recall that 2-clique-colouring of ($1,
1$)-polar graphs is in $\mathcal{P}$, since ($1, 1$)-polar are a subclass of
chordal graphs, thereby 2-clique-colourable. Notice that weakly chordal,
($2, 1$)-polar, and ($1, 1$)-polar (resp. perfect, ($3, 1$)-polar, and ($1,
1$)-polar) are nested classes of graphs.

Given a graph $G=(V, E)$ and adjacent vertices $a, g \in V$, we say that we add
to $G$ a copy of an auxiliary graph $AK(a, g)$ of order $7$  -- depicted in 
Fig.~\ref{fig:adjacentkeeper} -- if we change the definition of $G$ by
doing the following: we first change the definition of $V$ by adding to it
copies of the five vertices $b$, \ldots, $f$ of the auxiliary graph $AK(a,
g)$; then we change the definition of $E$, adding to it copies of the eight
edges $(u, v)$ of $AK(a, g)$.  Similarly, given a graph $G=(V, E)$ and
non-adjacent vertices $a, j \in V$, we say that we add to $G$ a copy of an
auxiliary graph $NAS(a, j)$ of order $10$ -- depicted in
Fig.~\ref{fig:nonadjacentswitcher} -- if we change the definition of $G$ by
doing the following: we first change the definition of $V$ by adding to it
eight copies of the vertices $b$, \ldots, $i$ of the auxiliary graph $NAS(a,
j)$; then we change the definition of $E$, adding to it copies of the thirteen
edges $(u, v)$ of $NAS(a, j)$.

\begin{figure}[t!]
\centering
	\subfloat
		[$AK(a, g)$]
		{
			\includegraphics[scale=0.4]{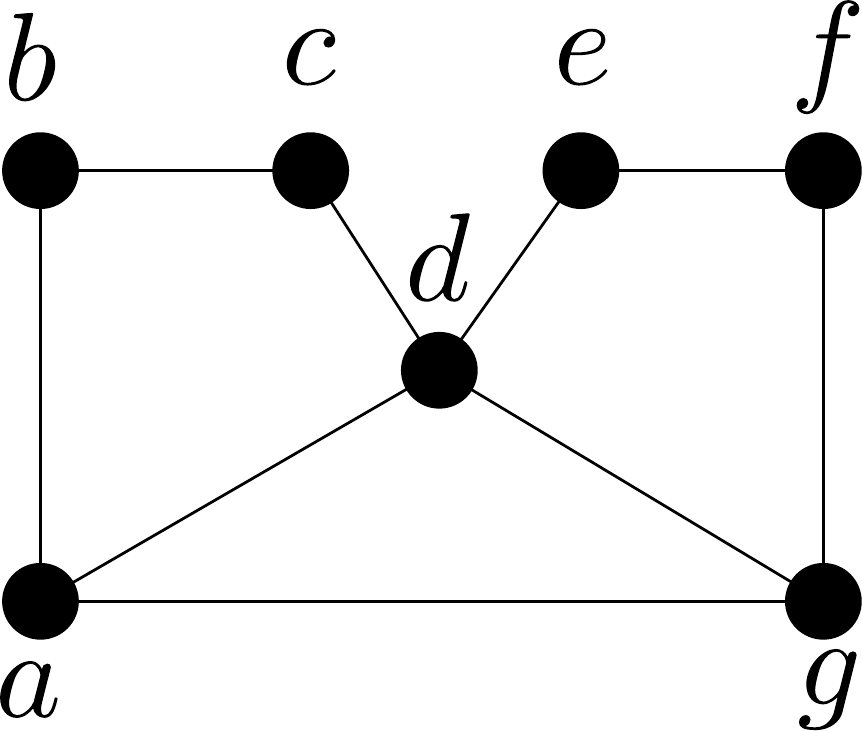}
			\label{fig:adjacentkeeper}
		}
	\qquad
	\subfloat
		[$NAS(a, j)$] 
		{
			\includegraphics[scale=0.4]{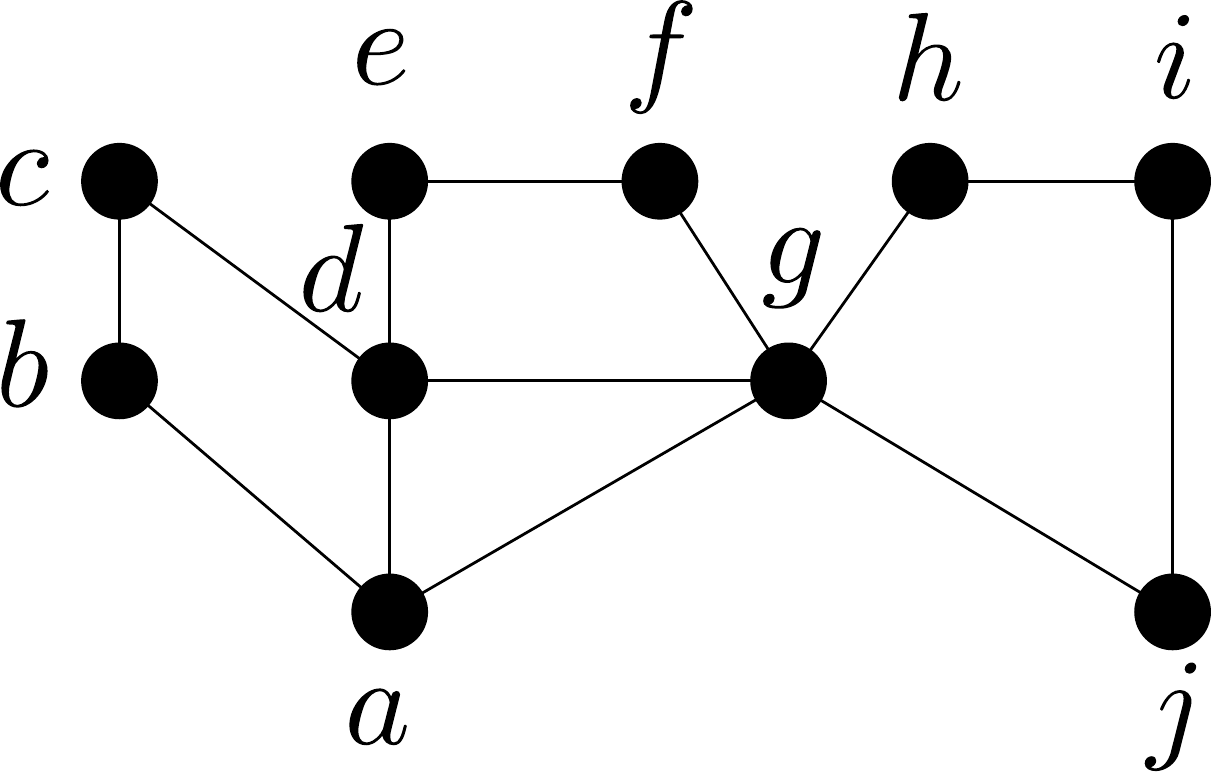}
			\label{fig:nonadjacentswitcher}
		}
	\caption{Auxiliary graphs $AK(a, g)$ and $NAS(a, j)$}
	\label{fig:auxiliarygraphs}
\end{figure}

The auxiliary graph $AK(a, g)$ is constructed to force the same
colour (in a 2-clique-colouring) to adjacent vertices $a$ and $g$, while the
auxiliary graph $NAS(a, j)$ is constructed to force distinct colours (in a
2-clique-colouring) to non-adjacent vertices $a$ and $j$ (see
Lemmas~\ref{lem:auxiliarygraphak}~and~\ref{lem:auxiliarygraphnas}). 

\begin{lemma}
\label{lem:auxiliarygraphak}
Let $G$ be a graph and $a, g$ be adjacent vertices in $G$. If we add to $G$ a
copy of an auxiliary graph $AK(a, g)$, then in any 2-clique-colouring of the
resulting graph, adjacent vertices $a$ and $g$ have the same colour.
\end{lemma}
\begin{proof}
Follows from the fact that in $AK(a, g)$ there exists a path $a b c \ldots
g$ such that no edge lies in a triangle of $G$.
\end{proof}

\begin{lemma}
\label{lem:auxiliarygraphnas}
Let $G$ be a graph and $a, j$ be non-adjacent vertices in $G$. If we add to $G$
a copy of an auxiliary graph $NAS(a, j)$, then in any 2-clique-colouring of the resulting
graph, non-adjacent vertices $a$ and $j$ have distinct colours.
\end{lemma}
\begin{proof}
Follows from the fact that in $NAS(a, j)$ there exists a path $a b c \ldots
j$ such that no edge lies in a triangle of $G$.
\end{proof}

We improve the proof of
D\'efossez~\cite{DefossezOddHoleFreeCliqueColouringComplexity}, in order to
determine the complexity of 2-clique-colouring for weakly chordal graphs.
Consider the {\sc QSAT2} problem, which is the $\Sigma_2^P$-complete canonical
problem~\cite{Marx2011}, as follows.

\begin{prob}
	{
		Quantified 2-Satisfiability (QSAT2)
	}
	{
		A formula $\Psi=(X,Y,D)$ composed of a disjunction $D$ of implicants (that 
	    are conjunctions of literals) over two sets $X$ and $Y$ of variables.
	}
	{
		Is there a truth assignment for $X$ such that for every truth assignment
		for $Y$ the formula is true? 
	}
\end{prob}

We prove that 2-clique-colouring weakly chordal graphs is $\Sigma_2^P$-complete
by reducing the $\Sigma_2^P$-complete canonical problem {\sc QSAT2} to it.
For a {\sc QSAT2} formula $\Psi=(X, Y, D)$, a weakly chordal graph $G$ is
constructed such that graph~$G$ is 2-clique-colourable if, and only if, there
is a truth assignment of $X$, such that $\Psi$ is true for every truth
assignment of $Y$. 

\begin{theorem}
\label{thm:weaklychordal}
The problem of 2-clique-colouring is $\Sigma_2^P$-complete for weakly chordal
graphs.
\end{theorem}
\begin{proof}
A 2-partition of the graph is a certificate to decide whether a graph has a
2-clique-colouring. Moreover, a monochromatic clique is a certificate to check
whether a 2-partition is not a 2-clique-colouring. Finally, it is easy to
describe a polynomial-time algorithm to check whether a complete set is
monochromatic and maximal. Hence, 2-clique-colouring is a $\Sigma_2^P$ problem.

We prove that 2-clique-colouring weakly chordal graphs is $\Sigma_2^P$-hard
by reducing {\sc QSAT2} to it. Let $n$, $m$, and $p$ be the number of variables
$X$, $Y$, and implicants, respectively, in formula~$\Psi$. We define graph
$G$, as follows.

\begin{itemize}
	\item for each variable $x_i$, we create vertices $x_i$ and $\overline{x}_i$;
	\item for each variable $y_j$, we create vertices $y_j$, $y^\prime_j$,
	and $\overline{y}_j$ and edges $y_j y^\prime_j$, and $y^\prime_j
	\overline{y}_j$;
	\item we create a vertex $v$ and edges so that the set
	$\{x_1, \overline{x}_1, \ldots , x_n, \overline{x}_n, y_1, \overline{y}_1,
	\ldots, y_m, \overline{y}_m, v\}$ induces a complete subgraph of $G$ minus the
	matching $\{\{x_1, \overline{x}_1\}, \ldots , \{x_n, \overline{x}_n\}, \{y_1,
	\overline{y}_1\}, \ldots, \{y_m, \overline{y}_m\}\}$;
	\item add copies of the auxiliary graph $NAS(x_i, \overline{x}_i)$, for $i =
	1, \ldots, n$;
	\item add copies of the auxiliary graph $AK( \overline{y}_j, y_{j+1})$, for $j
	= 1, \ldots, m-1$;
	\item add a copy of $AK(\overline{y}_m, v)$; and
	\item for each implicant $d_k$, we create vertices $d_k, d^\prime_k$,
	$d^{\prime\prime}_k$, and we add the edges $d_k d^\prime_k$, $d^\prime_k
	d^{\prime \prime}_k$, $d^{\prime\prime}_k v$, and $d_k v$. Moreover, each
	vertex $d_k$ is adjacent to a vertex $l$ in $\{x_1, \overline{x}_1, \ldots ,
	x_n, \overline{x}_n, y_1, \overline{y}_1, \ldots, y_m, \overline{y}_m, v\}$
	if, and only if, the literal correspondent to $\overline{l}$ is not in the
	implicant correspondent to vertex~$d_k$.
\end{itemize}

Refer to Fig.~\ref{fig:reduction} for an example of such construction, given a
formula $\Psi = (x_1 \wedge \overline{x}_2 \wedge y_2) \vee (x_1 \wedge x_3 \wedge
\overline{y}_2) \vee ( \overline{x}_1 \wedge \overline{x}_2 \wedge y_1)$.

\begin{figure}[t!]
\centering
	\subfloat
		[Graph constructed for a QSAT2 instance $\Psi = (x_1 \wedge \overline{x}_2 \wedge y_2) \vee (x_1
	\wedge x_3 \wedge \overline{y}_2) \vee ( \overline{x}_1 \wedge \overline{x}_2 \wedge
	y_1)$] {
			\includegraphics[scale=0.5]{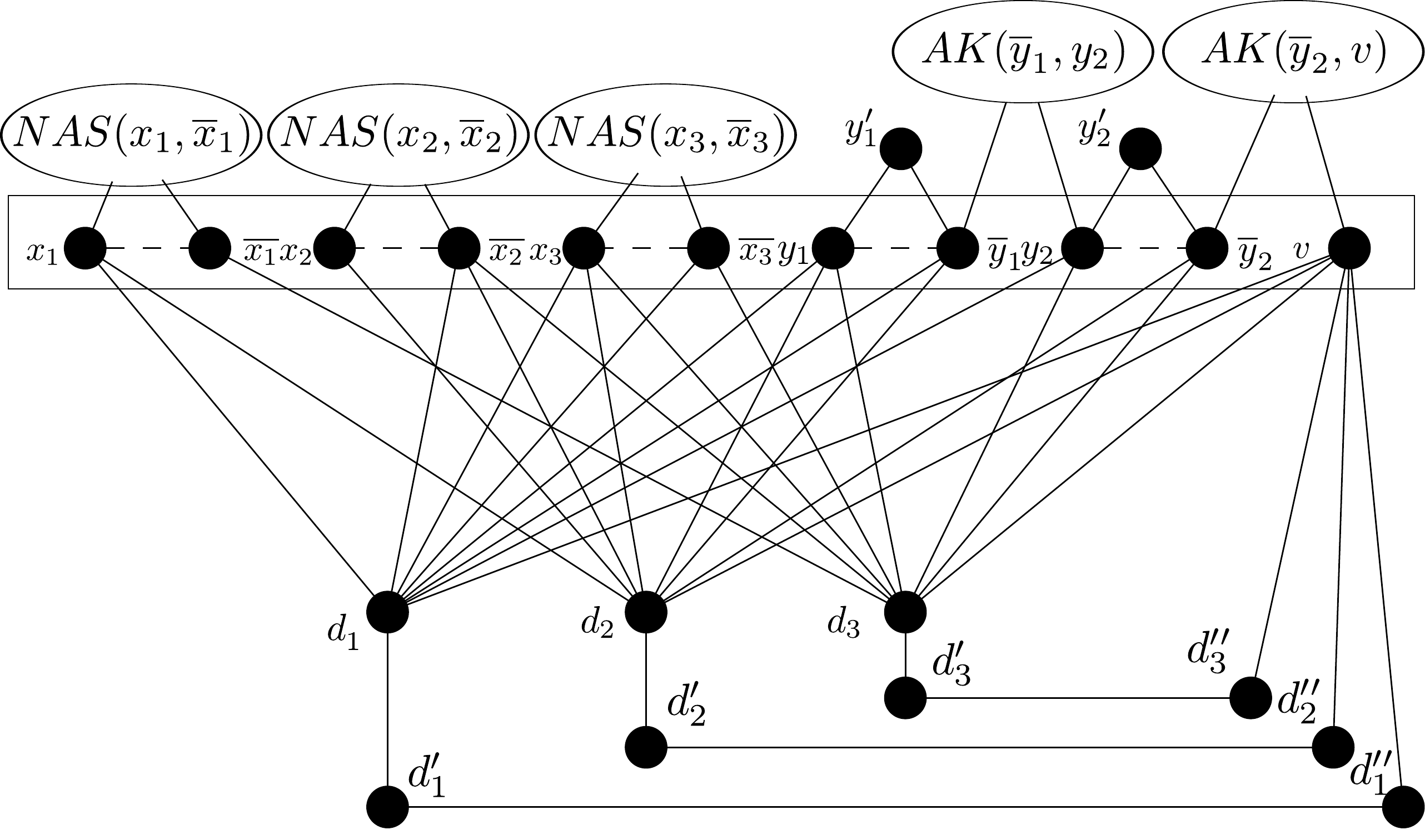}
			\label{fig:reduction-uncoloured}
		}
	\qquad
	\subfloat
		[A satisfying truth assignment of $x_{1} = \overline{x_{2}} = x_{3} = T$] {
			\includegraphics[scale=0.5]{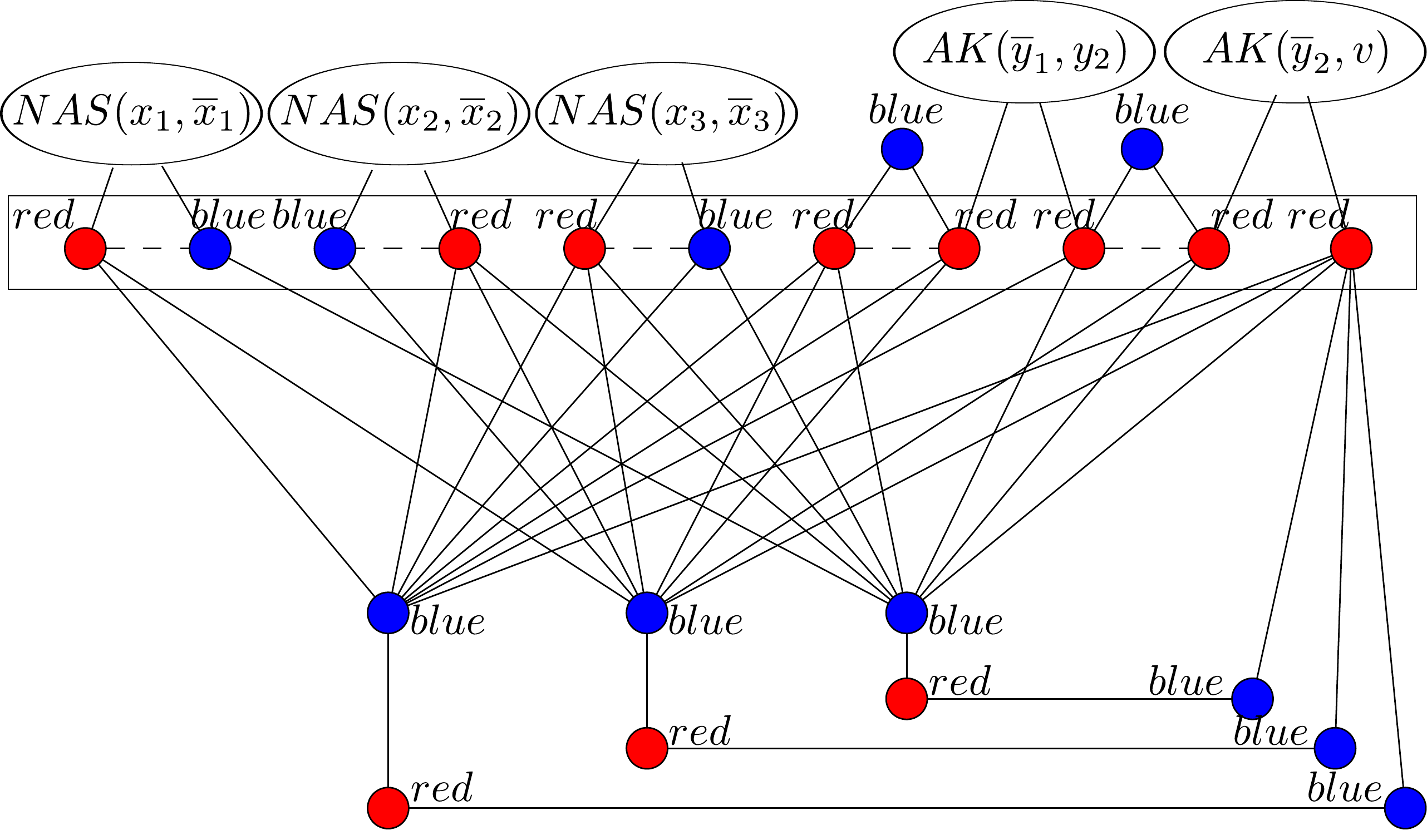}
			\label{fig:reduction-coloured}
		}
	\caption{Example of a graph constructed  for a QSAT2 instance, where $NAS$ and $AK$ denote the respectively auxiliary graphs}
	\label{fig:reduction}
\end{figure}

We claim that graph $G$ is 2-clique-colourable if, and only if, $\Psi$ has a
solution. For every $i$, the vertices $x_i$ and $\overline{x}_i$ have
opposite colours in any 2-clique-colouring of $G$
(see Lemma~\ref{lem:auxiliarygraphnas}).
The set $\{y_1, y_2, \ldots, y_{m}, \overline{y}_1, \overline{y}_2, \ldots,
\overline{y}_m, v\}$ is monochromatic. Indeed, $y_j$, $\overline{y}_j$,
$y_{j+1}$ have the same colour, since $y_j y^\prime_j$, $y^\prime_j
\overline{y}_j$ are cliques and, by Lemma~\ref{lem:auxiliarygraphak},
$\overline{y}_j y_{j+1}$ as well as $\overline{y}_{m} v$ have the same colour.
Finally, $d_1$, $d_2$, \ldots, $d_{p}$ all have the same colour, which is the
opposite to the colour of $v$.

Assume there exists a valuation $v_X$ of $X$ such that $\Psi$ is satisfied for
any valuation of $Y$. We give a colouring to the graph $G$, as follows.
 
\begin{itemize}
	\item assign colour 1 to $y_j$, $\overline{y}_j$, $d^\prime_k$, and $v$,
	\item assign colour 2 to $y^\prime_j$, $d_k$ and $d^{\prime \prime}_k$,
	\item extend the unique 2-clique-colouring to the $m-1$ copies of the
	auxiliary graph $AK(\overline{y}_j, y_{j+1})$ and $AK(\overline{y}_{m}, v)$,
	\item assign colour 1 to $x_i$ if the corresponding variable is $true$ in
	$v_X$, otherwise we assign colour 2 to it,
	\item assign colour 2 to $\overline{x}_i$ if the corresponding variable is
	$true$ in $v_X$, otherwise we assign colour 1 to it,
	\item extend the unique 2-clique-colouring to the $n$ copies of the auxiliary
	graph $NAS(x_i, \overline{x}_i)$.
\end{itemize}

It still remains to be proved that this is indeed a 2-clique-colouring. Let us assume that
it is not the case and that there exists a maximal clique $K$ of $G$ that is monochromatic.
Clearly,~$K$ is not contained in a copy of any auxiliary graph, and that it does
not contain any vertex of type $y^\prime_j$, $d^\prime_k$, or $d^{\prime \prime}_k$. 
As $v$ is adjacent to all other vertices (which are the $x_i$, $\overline{x}_i$,
$y_j$, $\overline{y}_j$, and $d_k$), we deduce that $v \in K$ and, subsequently,
that all vertices of $K$ have colour 1. Moreover, $K$ contains exactly one
vertex among $x_i$ and $\overline{x}_i$, i.e. the one corresponding to the
literal which is $true$ in $v_X$, and similarly exactly one vertex among $y_j$
and $\overline{y}_j$. We remark that $K$ does not contain any $d_k$ since they
have colour 2. Then we define a valuation $v_Y$ in the following way. If $y_j
\in K$, then $v_Y$ assigns value $true$ to the corresponding variable,
otherwise $v_Y$ assigns the value $false$. Thus, the literals corresponding to
the vertices of $K \setminus \{v\}$ are exactly those that are $true$ in the
total valuation $(v_X,v_Y)$. Let us consider now any $d_k$. Since $K$ is
maximal, each $d_k$ is not adjacent to at least one vertex of $K$. By
construction of $G$, this means that all implicants are $false$, which
contradicts the definition of $v_X$. Hence, there is no monochromatic clique
and we have a 2-clique-colouring.

For the converse, we now assume that $G$ is 2-clique-colourable and we consider
any 2-clique-colouring with colours 1 and 2. Without loss of generality, we can
assume that $v$ has colour 1. Then, $y_j$ and $\overline{y}_j$ have colour 1
and $d_k$ has colour 2. Vertices $x_i$ and $\overline{x}_i$ have opposite
colours and we define $v_X$ in the following way. The literal $x_i$ is assigned
$true$ in $v_X$ if the corresponding vertex has colour 1 in the
clique-colouring, otherwise it is assigned $false$ in $v_X$.
Let $v_Y$ be any valuation of $Y$. Consider the clique $K$ that contains $v$ and the vertices
corresponding to literals which are $true$ in the total valuation $(v_X, v_Y)$. Since all
those vertices have colour 1 and we have a 2-clique-colouring, it follows that
$K$ cannot be maximal. As a consequence, there exists some $d_k$ which is
adjacent to all vertices of $K$. Thus, the corresponding implicant is $true$ in
that valuation and this proves that $\Psi$ is satisfied for any valuation $v_Y$
and that $v_X$ has the right property.

It now remains to be proved that $G$ is a weakly chordal graph. Fixing edge
$v d_i$ as an edge of a cycle, one can check that $G$ has no chordless cycle
of size greater than 5 as an induced subgraph. Now, we prove that $G$ has no
complement of a chordless cycle of size greater than 4 as an induced subgraph.

Let $\overline{H}$ be the complement of a chordless cycle of size greater than
5. Clearly, any vertex of $H$ has degree at least 3. Hence, we analyse the
vertices of $G$ with degree at least 3. Let $S = \{x_i, \overline{x}_i, y_j,
\overline{y}_j, v \mid 1 \leq i \leq n, 1 \leq j \leq m\}$ and let $R = \{d_k |
1 \leq k \leq p\}$. All vertices of $G$ with degree at least 3 are precisely the
vertices of the auxiliary graphs, $S$, and $R$. We invite the reader to check
that any vertex of an auxiliary graph that is not in $S$ does not belong to
$H$. Hence, every vertex of $H$ belongs to $S$ or $R$. First, we claim that
$R$ has at most 2 vertices of $H$. Indeed, 3 vertices of $R$ induce a
$\overline{K_3}$. Second, we claim that $S$ has at most 2 vertices. Notice that
$|R \cap H| > 0$, since a vertex in $S$ has at most one non-neighbor in $S$ and
every neighbor of $H$ has at least two non-neighbors in $H$. If $|R \cap H|$ =
1, then $|S \cap H| \leq 2$, since a vertex in $S$ has at most one non-neighbor
in $S$ and the unique vertex of $R \cap H$ has at most two non-neighbors in $S
\cap H$. If $|R \cap H|$ = 2, then $|S \cap H| \leq 2$, since a vertex in $S$
has at most one non-neighbor in $S$ and each vertex of $R \cap H$ has at most one
non-neighbor in $S$. Hence, at most two vertices of $H$ are in $R$ and at most
two vertices of $H$ are in $S$, i.e. $|H| \leq 4$, which is a contradiction.
\end{proof}

Now, our focus is on showing a subclass of weakly chordal graphs in which
2-clique-colouring is $\mathcal{NP}$-complete, namely (3, 1)-polar and (2,
1)-polar graphs.

Complements of bipartite graphs are a subclass of ($\infty, 1$)-polar graphs.
Indeed, let $G = (V,E)$ be a complement of a bipartite graph, where $(A, B)$ is
a partition of $V$ into two disjoint complete sets. Clearly, $G$ is a ($\infty,
1$)-polar graph.
D\'efossez~\cite{DefossezOddHoleFreeCliqueColouringComplexity} showed that it
is co$\mathcal{NP}$-complete to check whether a 2-colouring of a complement of
a bipartite graph is a
2-clique-colouring~\cite{DefossezOddHoleFreeCliqueColouringComplexity}. Hence,
it is co$\mathcal{NP}$-hard to check if a colouring of the vertices of a
($\infty, 1$)-polar graph is a 2-clique-colouring. On the other hand, we show
next that, if $k$ is fixed, listing all cliques of a ($k, 1$)-polar graph and
checking if each clique is polychromatic can be done in polynomial-time,
although the constant behind the big $O$ notation is impraticable.
The outline of the algorithm follows. We create a subroutine in which, given a
satellite~$K$ of $G$, we check whether every clique of $G$ containing a subset
of $K$ is polychromatic. Lemma~\ref{lem:nptocheck} determines the complexity of
the subroutine and proves its correctness. The algorithm runs the subroutine for
each satellite of $G$ and, as a final step, check whether partition $B$
is polychromatic if, and only if, partition $B$ is a clique of $G$.
Theorem~\ref{thm:nptocheck} determines the complexity of the algorithm and prove
its correctness.

\begin{lemma}
\label{lem:nptocheck}
There exists an $O(n)$-time algorithm to check whether every clique that
contains a subset of a satellite $S$ of a ($k, 1$)-polar graph, for a fixed $k
\geq 1$,is polychromatic.
\end{lemma}
\begin{proof}
We prove the correctness of Algorithm~\ref{alg:A_i} by induction. Let $A_1 = S$
and $B_1 = \displaystyle\bigcap_{v \in A_1} \left( N(v) \cap B \right)$. Notice
that $A_1 \cup B_1$ is the unique clique of graph $G$ that contains $S$. If $A_1 \cup
B_1$ is monochromatic, then $\pi$ is not a 2-clique-colouring of $G$.
Otherwise, i.e. $A_1 \cup B_1$ is polychromatic, we are done. Now, we need to
check whether every clique of $G$ containing a proper subset of $S$ is
polychromatic.

Let $A_{i+1} = A_1 \setminus \{x_1, \ldots, x_{i}\}$ and $B_{i+1} =
\displaystyle\bigcap_{v \in A_{i+1}} (N(v) \cap B)$, for some $\{x_1,
\ldots, x_{i}\} \subset A_1$. As an induction hyphotesis, suppose that every
clique of $G$ containing $A_{j}$, for every $1 \leq j \leq i$, is
polychromatic.

For the induction step, consider a clique $K$ of graph $G$ containing
$A_{i+1}$. By induction hyphotesis, if $K$ contains $A_{j}$, for some $1
\leq j \leq i$, then $K$ is polychromatic.
Now, consider that $K$ does not contain $A_{j}$, for any $1 \leq j \leq i$.
Then, $K = A_{i+1} \cup B_{i+1}$. If $K$ is monochromatic, then $\pi$ is not a
2-clique-colouring of $G$. Otherwise, i.e. $K$ is polychromatic, then every
clique containing $A_{i+1}$ is polychromatic and the proof of the correctness
of Algorithm~\ref{alg:A_i} is done.
 
Now, we give the time-complexity of Algorithm~\ref{alg:A_i}. First, there are at
most $k!$ recursive calls. Second, the number of steps in an iteration of the
algorithm is upper bounded by the complexity of calculating $B_i$.
One can design an $O(|B| \log k)$-time algorithm to calculate $B_i$. Then,
the algorithm is executed in $O(n)$ steps, since $k$ is a constant and $|B|$ is
upper bounded by $n$.
\end{proof}

	\begin{algorithm}[h!]
	\SetKwInOut{Input}{input}
	\SetKwInOut{Output}{output}
	\Input{
			$G = (A, B)$, a ($k, 1$)-polar graph \\
			$\pi$, a 2-colouring of $G$ \\
			$A_i$, a satellite of $G$
	}
	\Output{$yes$, if every clique of $G$ containing a subset of a satellite
	$A_i$ is polychromatic}
	\caption{$O(n)$-time algorithm to output $yes$, if every
clique of a ($k, 1$)-polar graph containing a subset of a satellite of $G$ is
polychromatic, for a fixed $k \geq 1$.}
	\BlankLine
	\Begin
	{	
		\eIf{$|\pi(A_i)| \geq 2$}
		{
			\For{
				$i=1$ \emph{\KwTo} $|A_i|$
			}
			{
        		$answer \longleftarrow recursive(A_i \setminus \{x_i\})$\;
    	    	\If{$answer = no$}
    	    	{
					\Return{$no$\;}
    	    	}
	        }
			\Return{$yes$\;}
		}
		{
			$B_i \longleftarrow \displaystyle\bigcap_{v \in A_{i}} (N(v) \cap B)$\;
			\eIf{$|\pi(A_i \cup B_i)| \geq 2$}
			{
				\Return{$yes$\;}
	        }
	        {
				\Return{$no$\;}	        	
	        }
		}
	}
	\label{alg:A_i}
	\end{algorithm}

\begin{theorem}
\label{thm:nptocheck}
There exists an $O(n^2)$-time algorithm to check whether a colouring of the
vertices of a ($k, 1$)-polar graph, for a fixed $k \geq 1$, is a clique-colouring.
\end{theorem}
\begin{proof}
The correctness of Algorithm~\ref{alg:test2-clique-colouring} follows.
A clique of $G$ contains at least one vertex of a satellite of $G$ or it is $B$.
The first loop of Algorithm~\ref{alg:test2-clique-colouring} checks whether all
cliques in the former case are polychromatic. The second loop of
Algorithm~\ref{alg:test2-clique-colouring} checks whether $B$ is a clique. If
$B$ is a clique, then we check whether $B$ is polychromatic.

Now, we give the time-complexity of Algorithm~\ref{alg:test2-clique-colouring}.
The first loop of Algorithm~\ref{alg:test2-clique-colouring} runs at most $n$
times the Algorithm~\ref{alg:A_i}, which runs in $O(n)$-time. 
The second loop of Algorithm~\ref{alg:test2-clique-colouring} runs at most $n$
times one comparison, which runs in $O(n)$-time.
Then, Algorithm~\ref{alg:test2-clique-colouring} is executed in at most $O(n^2)$
steps.
\end{proof}
\def\proofname{Proof}

	\begin{algorithm}[h!]
	\SetKwInOut{Input}{input}
	\SetKwInOut{Output}{output}
	\Input{
			$G = (A, B)$, ($k, 1$)-polar graph \\
			$\pi$, a 2-colouring of $G$ \\
	}
	\Output{$yes$, if $\pi$ is a 2-clique-colouring o $G$.}
	\caption{$O(n^2)$-time algorithm to output $yes$, if $\pi$ is a
	clique-colouring of a ($k, 1$)-polar graph, for a fixed $k \geq 1$.}
	\BlankLine
	\Begin
	{	
		\ForEach
		{
			maximal complete set $A^\prime \in A$
		}
		{
        	$answer \longleftarrow Algorithm~\ref{alg:A_i}(G, \pi, A^\prime)$\;
    	   	\If{$answer = no$}
    	   	{
				\Return{$no$\;}
    	   	}
		}
		\ForEach
		{
			maximal complete set $A^\prime \in A$
		}
		{
			\ForEach
			{
				$v \in A^\prime$
			}
			{
	    	   	\If{$|N_B(v)| = |B|$}
	    	   	{
					\Return{$yes$\;}
	    	  	}
    	  	}
	    }
		\eIf{$|\pi(B)| \geq 2$}
		{
			\Return{$yes$\;}
	    }
	    {
			\Return{$no$\;}	        	
		}
	}
	\label{alg:test2-clique-colouring}
	\end{algorithm}

Consider the {\sc NAE-SAT} problem, known to be
$\mathcal{NP}$-complete~\cite{MR521057}.

\begin{prob}
	{
		Not-all-equal satisfiability (NAE-SAT)
	}
	{
		A set $X$ of boolean variables and a collection $C$ of clauses (set of
		literals over $U$), each clause containing at most three different literals. 
	}
	{
		Is there a truth assignment for $X$ such that every clause contains
at least one $true$ and at least one $false$ literal? 
	}
\end{prob}

We first illustrate the framework of Kratochv\'il and Tuza~\cite{Kratochvil} to
argue that 2-clique-colouring is $\mathcal{NP}$-hard with a reduction from {\sc
NAE-SAT}, as follows. Consider an instance $\phi$ of NAE-SAT. We construct
a graph $G$, as follows. For every variable $x$, add an edge between vertices $x$
and $\overline{x}$. For every clause~$c$, add a triangle on vertices $\ell_c$
for all literals $\ell$ occurring in $c$. To finish the construction of $G$,
for every literal $\ell$ and for every clause $c$ containing $\ell$, add an
edge between $\ell$ and $\ell_c$. The (maximal) cliques of $G$ are the edges
$x\overline{x}$, $\ell \ell_c$, and triangles $\{\ell_c \mid \ell \in c\}$.
Hence, $G$ is 2-clique-colourable if, and only if, $\phi$ is not-all-equal
not-all-equal satisfiable. Refer to Fig.~\ref{fig:roadmap} for an example of
such construction, given a formula $\phi = (x_1 \vee \overline{x}_2 \vee y_2) \wedge
(x_1 \vee x_3 \vee \overline{y}_2) \wedge ( \overline{x}_1 \vee \overline{x}_2
\vee y_1)$.

\begin{figure}[t!]
\centering
	\subfloat
		[Graph constructed for a NAE-SAT instance $\phi = ( x_{1} \vee
	\overline{x_{2}} \vee x_{4}) \wedge ( x_{2} \vee \overline{x_{3}} \vee
	\overline{x_{5}} )\wedge ( x_{1} \vee x_{3} \vee x_{5})$] {
			\includegraphics[scale=0.3]{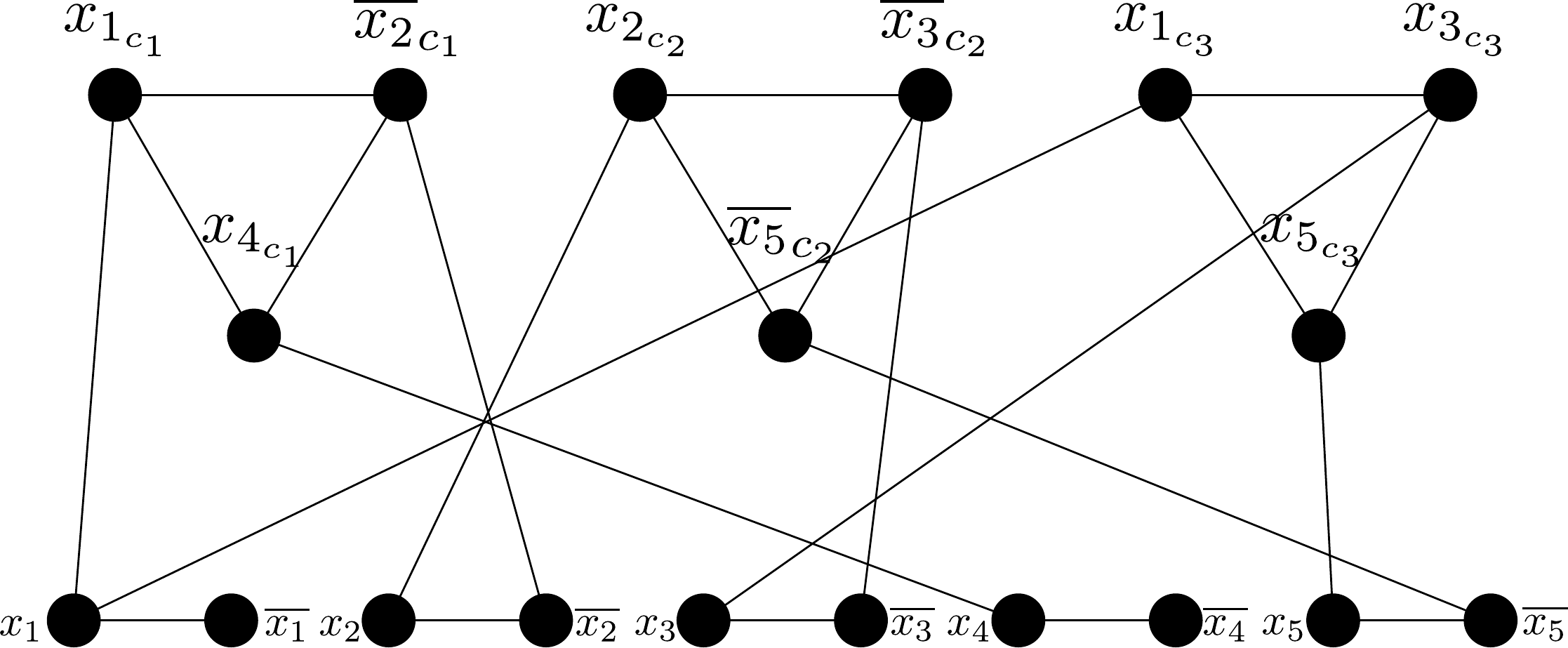}
			\label{fig:roadmap-uncoloured}
		}
	\qquad
	\subfloat
		[A satisfying truth assignment of $x_{1} = x_{2} = x_{3} = x_{4} = \overline{x_{5}} = T$
		implies a 2-clique-colouring of $G$, where $x_{1}, x_{2}, x_{3}, x_{4},
		\overline{x_{5}}$ receive blue colour] {
			\includegraphics[scale=0.3]{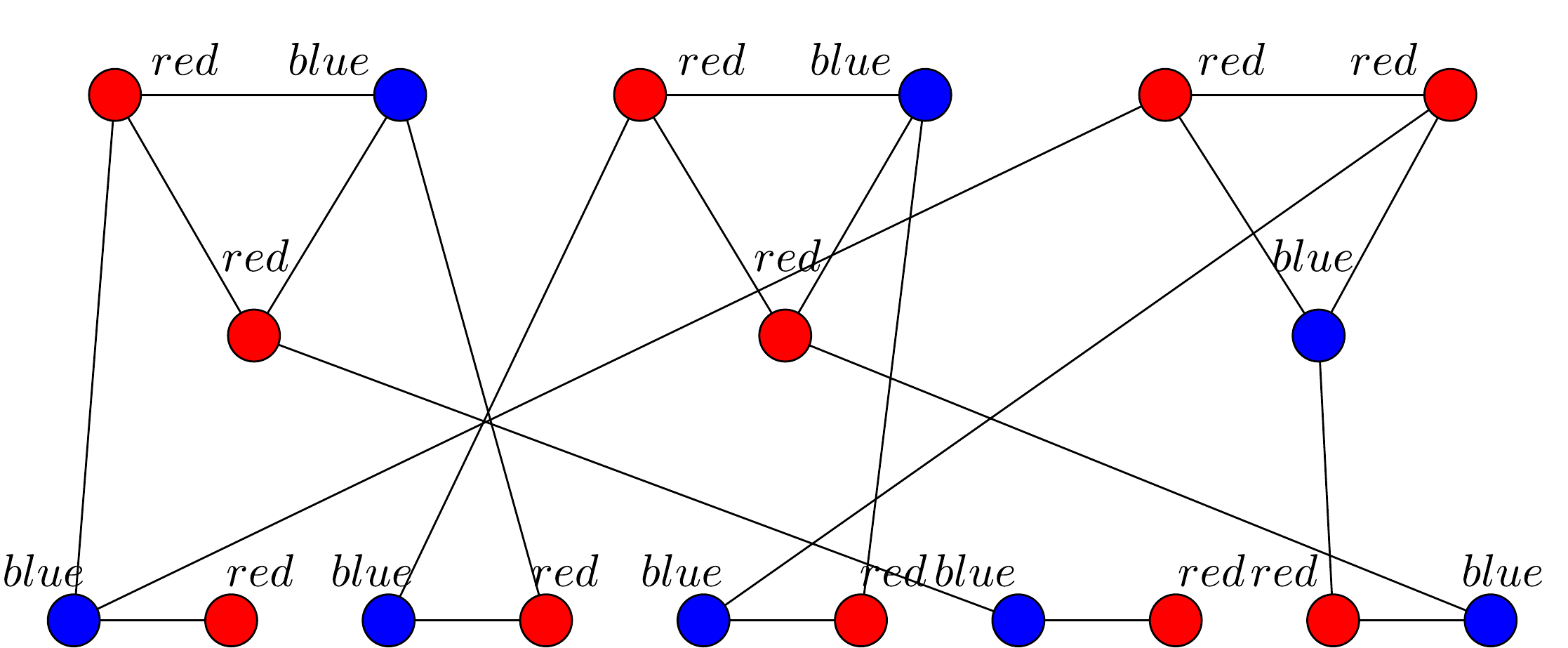}
			\label{fig:roadmap-coloured}
		}
	\caption{Example of a graph constructed following the framework given by
	Kratochv\'il and Tuza~\cite{Kratochvil} for a {\sc NAE-SAT} instance $\phi$}
	\label{fig:roadmap}
\end{figure}

We apply the ideas of the framework of Kratochv\'il and Tuza~\cite{Kratochvil}
to determine the complexity of 2-clique-colouring of (3, 1)-polar graphs.
We prove that 2-clique-colouring (3, 1)-polar graphs is $\mathcal{NP}$-complete
by reducing the {\sc NAE-SAT} problem to it. For a {\sc NAE-SAT} formula
$\phi$, a (3, 1)-polar graph $G$ is constructed such that graph~$G$ is
2-clique-colourable if, and only if, $\phi$ is not-all-equal satisfiable. 
This is an intermediary step to achieve the complexity of
2-clique-colouring of (2, 1)-polar graphs, which are a subclass of weakly
chordal graphs.

%


\begin{theorem}
\label{thm:2cc31polar}
The problem of 2-clique-colouring is $\mathcal{NP}$-complete for
(3,~1)-polar~graphs.
\end{theorem} 
\begin{proof}
The problem of 2-clique-colouring a (3, 1)-polar graph is in $\mathcal{NP}$:
Theorem~\ref{thm:nptocheck} confirms that to check whether a colouring of a ($3,
1$)-polar graph is a 2-clique-colouring is in $\mathcal{P}$.

We prove that 2-clique-colouring (3, 1)-polar graphs is $\mathcal{NP}$-hard by
reducing {\sc NAE-SAT} to it.
The outline of the proof follows. For every formula $\phi$, a graph $G$
is constructed such that $\phi$ is not-all-equal satisfiable if, and only if,
graph $G$ is 2-clique-colourable. We finish the proof showing that~$G$ is (3,
1)-polar. Let $n$ (resp. $m$) be the number of variables (resp. clauses) in
formula $\phi$. We define graph $G$, as follows.

\begin{itemize}
	\item for each variable $x_i$, $1 \leq i \leq n$, we create four vertices
	$x_i$, $x^\prime_i$, $x^{\prime\prime}_i$, and $\overline{x}_i$ with edges $x_i
	x^\prime_i$, $x^\prime_i x^{\prime\prime}_i$, and $x^{\prime\prime}_i
	\overline{x}_i$. Notice that vertices $x_i$ and $\overline{x}_i$ correspond to the
	literals of variable~$x_i$. Moreover, we create edges so that the set
	$\{x_1, \overline{x}_1, \ldots, x_n, \overline{x}_n\}$ induces a complete subgraph
	of~$G$;
	\item for each clause $c_j = (l_a, l_b, l_c)$, $1 \leq j \leq m$, we create a
	triangle $c_j$ with three vertices $l_{a_{c_j}}$, $l_{b_{c_j}}$, and
	$l_{c_{c_j}}$. Notice that vertices $l_{a_{c_j}}$, $l_{b_{c_j}}$, and
	$l_{c_{c_j}}$ correspond to the literals of clause $c_j$. Moreover, each
	vertex $l \in \{l_{a_{c_j}}, l_{b_{c_j}}, l_{c_{c_j}}\}$ is adjacent to a
	vertex $\overline{l}$ in $\{x_1, \overline{x}_1, \ldots , x_n,
	\overline{x}_n\}$ if, and only if, the literals correspondent to $l$ and
	$\overline{l}$ are distinct literals of the same variable.
\end{itemize}

Refer to Fig.~\ref{fig:reduction-np} for an example of such construction, given
a formula $\phi = (x_{1} \vee \overline{x_{2}} \vee x_{4}) \wedge ( x_{2} \vee
\overline{x_{3}} \vee \overline{x_{5}} )\wedge ( x_{1} \vee x_{3} \vee x_{5})$.

\begin{figure}[t!]
\centering
	\subfloat
		[Graph constructed for a NAE-SAT instance $\phi = ( x_{1} \vee
	\overline{x_{2}} \vee x_{4}) \wedge ( x_{2} \vee \overline{x_{3}} \vee
	\overline{x_{5}} )\wedge ( x_{1} \vee x_{3} \vee x_{5})$] {
			\includegraphics[scale=0.28]{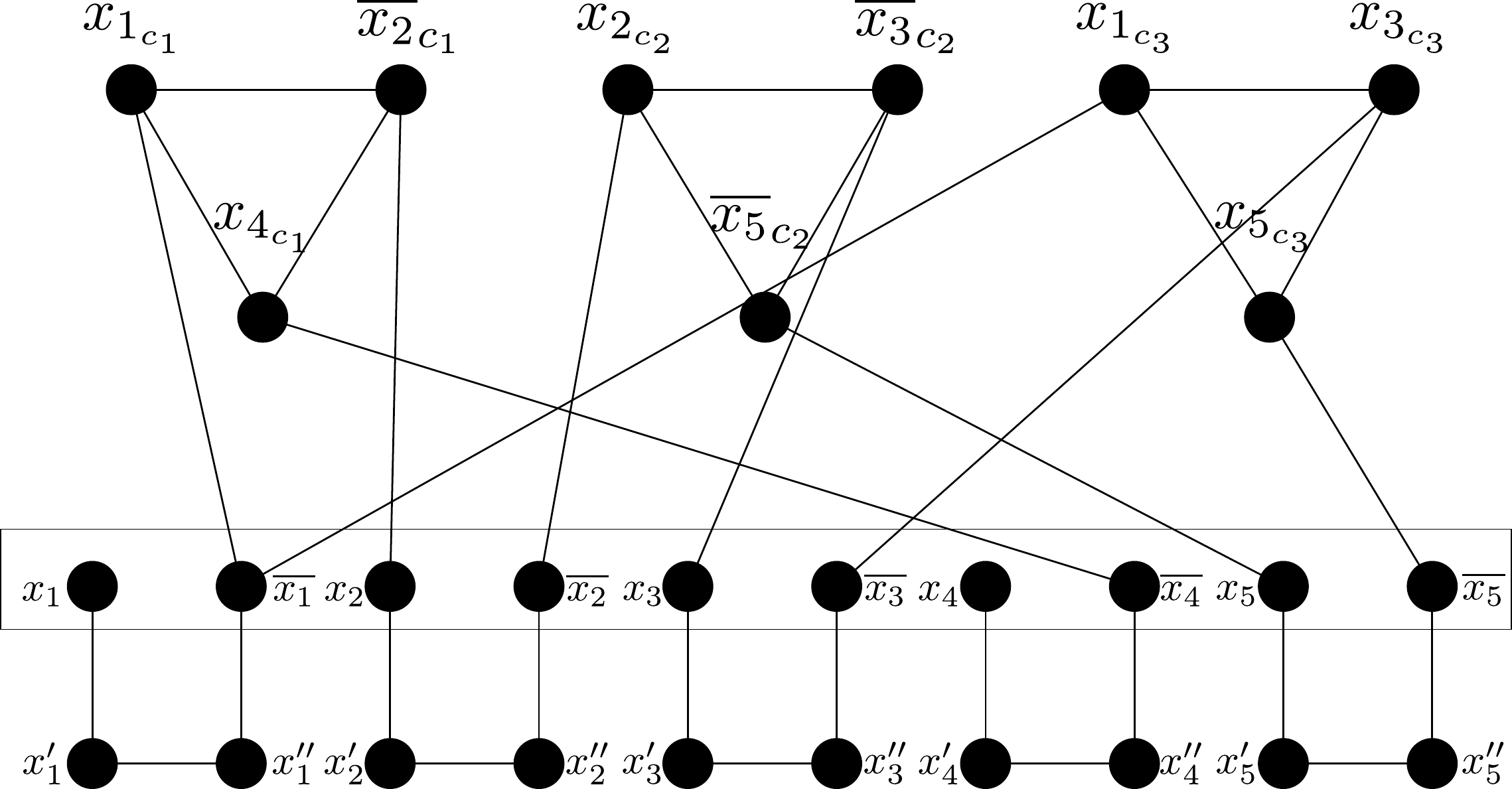}
			\label{fig:reduction-np-uncoloured}
		}
	\qquad 
	\subfloat
		[A satisfying truth assignment of $x_{1} = x_{2} = x_{3} = x_{4} = \overline{x_{5}} = T$
		implies a 2-clique-colouring of $G$, where $x_{1}, x_{2}, x_{3}, x_{4},
		\overline{x_{5}}$ receive blue colour] {
			\includegraphics[scale=0.28]{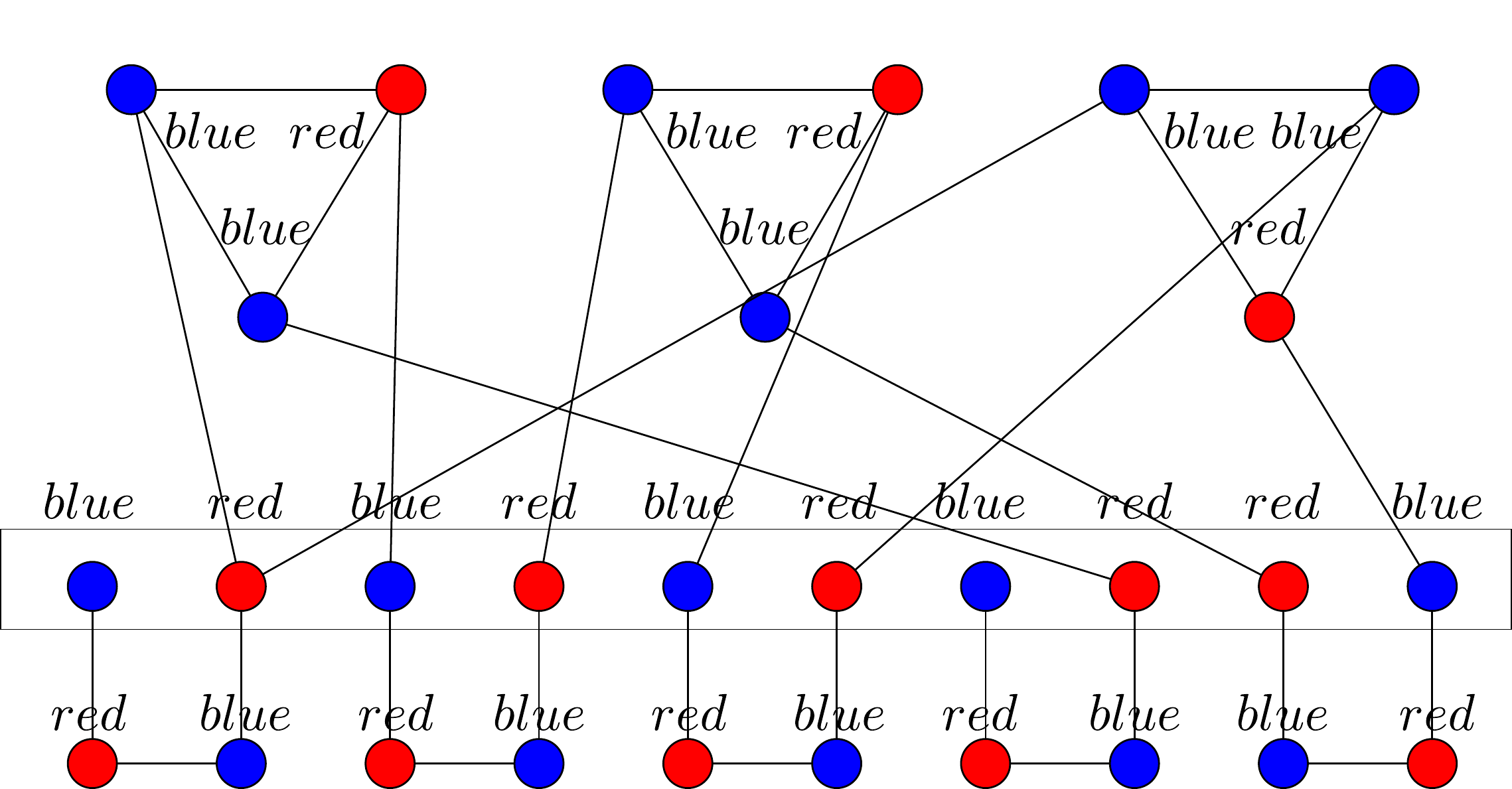}
			\label{fig:reduction-np-coloured}
		}
	\caption{Example of a (3, 1)-polar graph constructed for a NAE-SAT instance}
	\label{fig:reduction-np}
\end{figure}

We claim that there exists a 2-clique-colouring in~$G$ if, and only if, formula
$\phi$ is satisfiable. Assume that there exists a valuation $v_\phi$ such that $\phi$ is
not-all-equal satisfied. We give a colouring to graph~$G$, as follows. 
\begin{itemize}
	\item assign colour 1 to $l \in \{x_1, \overline{x}_1, \ldots, x_n,
	\overline{x}_n\}$ if it corresponds to the literal which receives the $true$
	value in $v_\phi$, otherwise we assign colour 2 to it,
	\item extend the unique 2-clique-colouring to vertices $x^\prime_i$ and
	$x^{\prime\prime}_i$, for each $1 \leq i \leq n$, i.e. assign
	colour 2 to $x^\prime_i$ and colour 1 to $x^{\prime\prime}_i$ if the
	corresponding literal of $x_i$ is $true$ in $v_\phi$. Otherwise, we assign
	colour 2 to $x^\prime_i$ and colour 1 to $x^{\prime\prime}_i$,
	\item extend the unique 2-clique-colouring to vertices $l_{a_{c_j}}$,
	$l_{b_{c_j}}$, and $l_{c_{c_j}}$, for each triangle $c_j = \{l_{a_{c_j}},
	l_{b_{c_j}}, l_{c_{c_j}}\}$, $1 \leq j \leq m$, i.e. assign colour 1 to $l \in
	\{l_{a_{c_j}}, l_{b_{c_j}}, l_{c_{c_j}}\}$ if the corresponding opposite
	literal of the same variable is $false$ in $v_\phi$, or else we assign colour 2 to $l$.
\end{itemize}

It still remains to be proved that this is indeed a 2-clique-colouring.

The cliques of size 2 are $x_i x^\prime_i$, $x^\prime_i x^{\prime\prime}_i$, and
$x^{\prime\prime}_i \overline{x}_i$ and $l\overline{l}$, where $l \in c_j$ and
$\overline{l}$ correspond to distinct literals of the same variable. The above
colouring gives distinct colours to each vertex of a clique of size 2.

The cliques of size at least 3 are $\{x_1, \overline{x}_1, \ldots, x_n,
\overline{x}_n\}$ and triangle $c_j$. The former is
polychromatic, since two vertices which represent distinct literals of the same
variable receive distinct colours. The latter is also polychromatic, since
$(i)$ every vertex of triangle $c_j$ represents a literal of some variable in
$\phi$, $(ii)$ each clause $c_j$ has at least one literal that receives the
$true$ value and at least one literal that receives the $false$ value (recall
we are reducing from {\sc NAE-SAT} problem), and $(iii)$ a vertex in
triangle $c_j$ receives colour 1 if it corresponds to the literal in which the
opposite literal of the same variable receives the $false$ value in $v_\phi$.
Otherwise, we assign colour 2 to it.

For the converse, we now assume that $G$ is 2-clique-colourable and we consider
any 2-clique-colouring. Recall the vertices $x_i$ and $\overline{x}_i$
have distinct colours, since $x_i x^\prime_i$, $x^\prime_i x^{\prime\prime}_i$,
and $x^{\prime\prime}_i \overline{x}_i$ are cliques. Hence, we define $v_\phi$
as follows. The literal $x_i$ is assigned $true$ in $v_\phi$ if the
corresponding vertex has colour 1 in the clique-colouring, otherwise it is
assigned $false$. Since we are considering a 2-clique-colouring, every triangle
$c_j$ is polychromatic. As a consequence, there exists at
least one literal with $true$ value in $c_j$ and at least one literal with
$false$ value in every clause $c_j$. This proves that $\phi$ is satisfied for
valuation $v_\phi$.

It now remains to be proved that $G$ is a (3, 1)-polar graph.
Let $A = \displaystyle \left(\bigcup_{i = 1}^n \{x^\prime_i,
x^{\prime\prime}_i\}\right) \cup \left(\bigcup_{j = 1}^m V(c_j)\right)$ and $B
= \{x_1, \overline{x}_1, \ldots, x_n, \overline{x}_n\}$ be the partition of
$V(G)$ into two sets. Notice that each satellite is either a triangle or an
edge. Hence, $G$ is a (3, 1)-polar graph.
\end{proof}

An additional requirement to the {\sc NAE-SAT} problem is that all variables
must be positive (no negated variables). This defines the known variant {\sc
Positive NAE-SAT}~\cite[Chapter 7]{Moret}. The variant {\sc Positive NAE-SAT}
is $\mathcal{NP}$-complete and the proof is by reduction from {\sc NAE-SAT}. 
Replace every negated variable $\overline{l_i}$ by a fresh
variable $l_j$, and add a new clause $(l_i, l_j)$, to enforce the complement
relationship. Notice that the new clause has only two literals. Hence,
all that is needed here is (i) duplicate a clause with a negated variable, say
$\overline{l_i}$, (ii) replace negated variable $\overline{l_i}$ by fresh
variables $l_j$ and $l_k$, respectively, and (iii) add a new clause $(l_i, l_j,
l_k)$, to enforce the complement relationship. An alternative proof
that 2-clique-colouring is $\mathcal{NP}$-complete for (3, 1)-polar graphs
can be obtained by a reduction from {\sc Positive NAE-SAT}. The proof is
analogous, but the constructed graph is as simple as the input (in constrast to
{\sc NAE-SAT} input).

We use a reduction from 2-clique-colouring of (3, 1)-polar graphs to determine
the complexity of 2-clique-colouring of (2, 1)-polar graphs.
In what follows, we provide some notation to classify the structure of
2-clique-colouring of (2, 1)-polar graphs and of (3, 1)-polar graphs. We
capture their similarities and make it feasible a reduction from
2-clique-colouring ($3, 1$)-polar graphs to 2-clique-colouring (2, 1)-polar graphs.

Let $G=(V, E)$ be a (3, 1)-polar graph. Let $K$ be a
satellite of $G$. Consider the following four cases: ($\ka$) there
exists a vertex of $K$ such that none of its neighbors is in partition $B$;
($\kb$) the complementary case of $\ka$, where there exists a pair of vertices of $K$,
such that the closed neighborhood of one vertex of the pair is contained in the
closed neighborhood of the other vertex of the pair; ($\kc$) the complementary
case of $\kb$, where the intersection of the closed neighborhood of the
vertices of $K$ is precisely~$K$; and ($\kd$) the complementary case of $\kc$.
Clearly, any satellite $K$ is either in case $\ka$, $\kb$,
$\kc$, or $\kd$. Refer to Fig.~\ref{fig:3k} (resp. Fig.~\ref{fig:2k}) for an
example of each case of a triangle (resp. edge) satellite.

\begin{figure}[t!]
\centering
	\subfloat
		[Case $\ka$] {
			\includegraphics[scale=0.23]{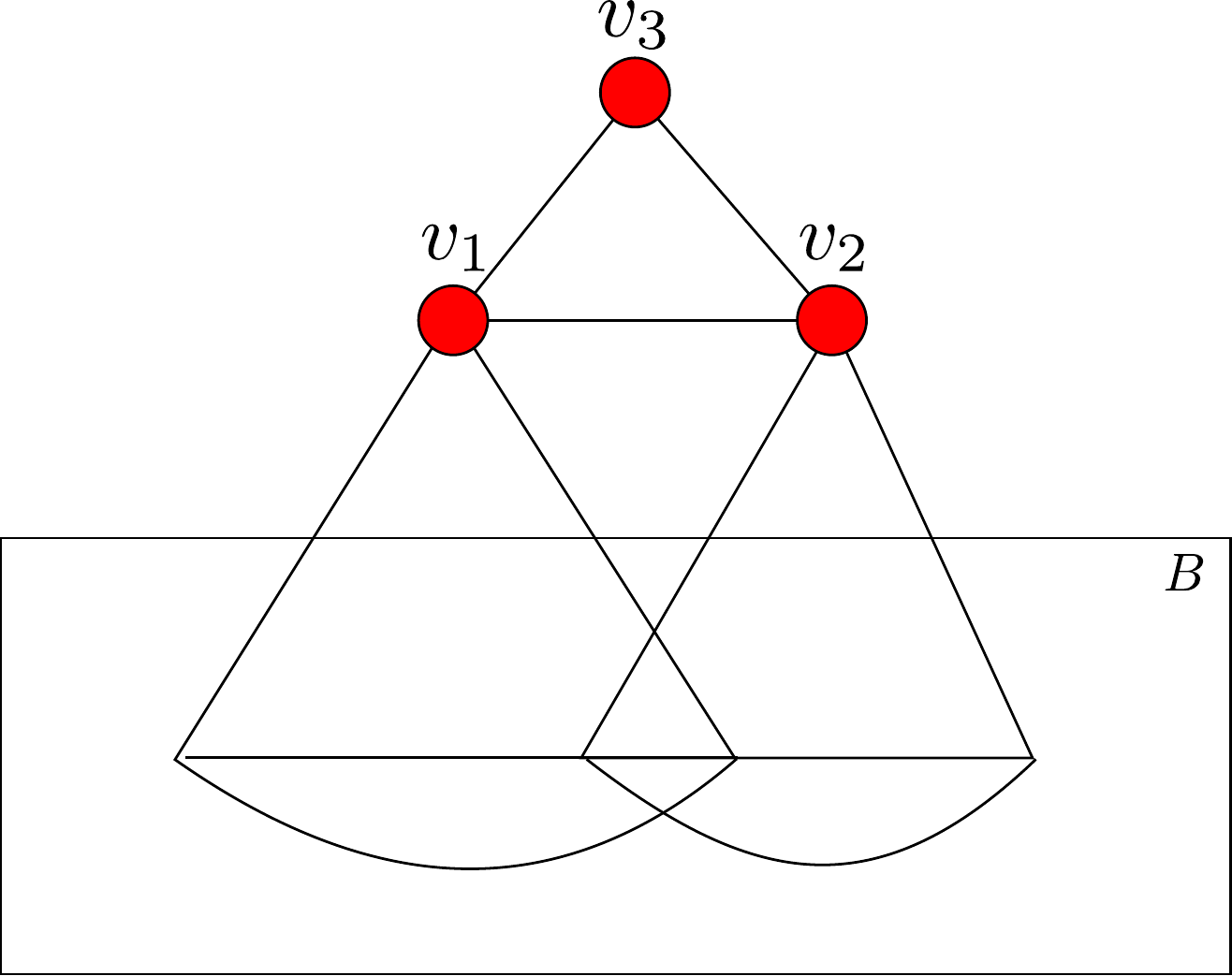}
			\label{fig:3ka}
		}
	\subfloat
		[Case $\kb$] {
			\includegraphics[scale=0.23]{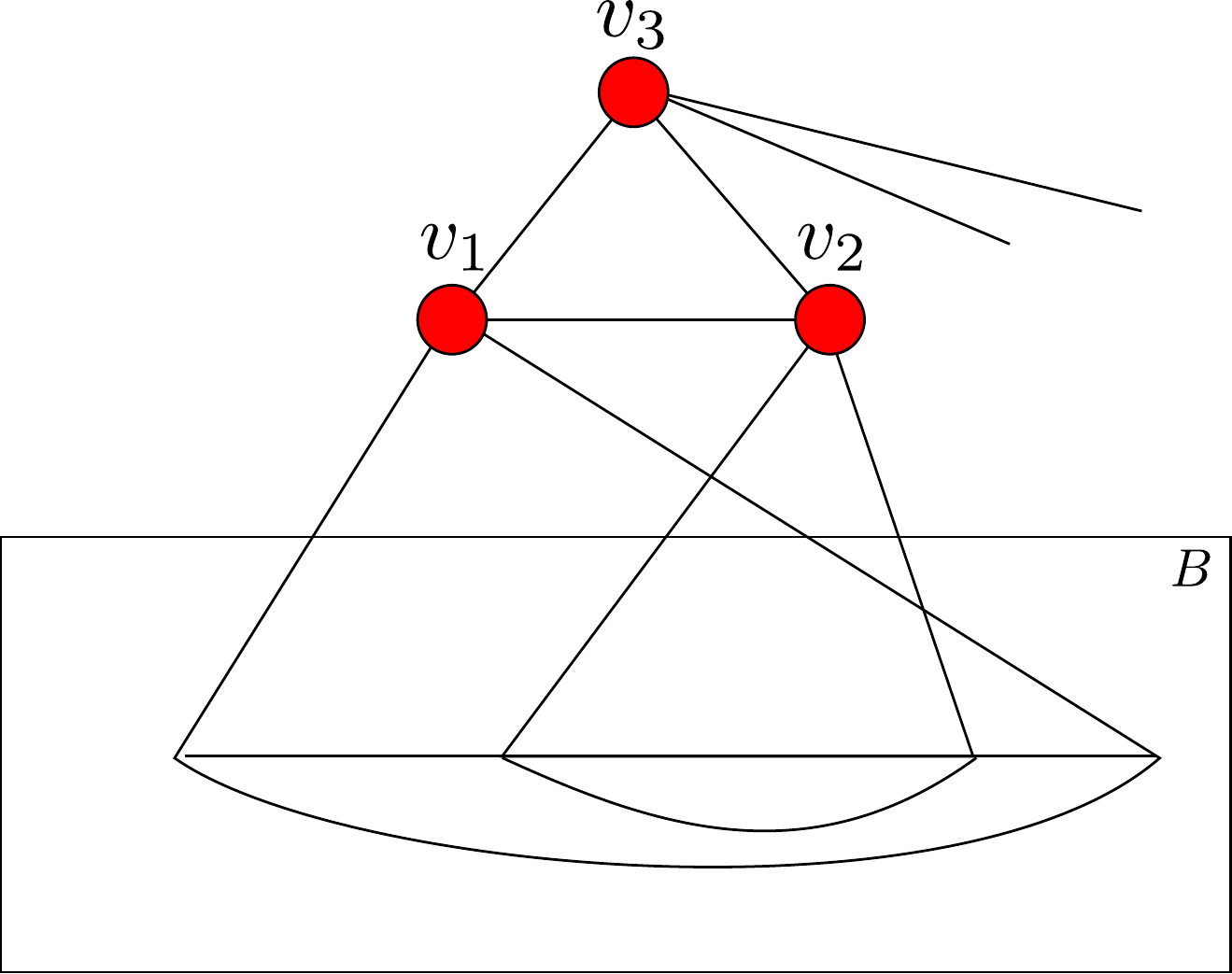}
			\label{fig:3kb}
		}
	\subfloat
		[Case $\kc$] {
			\includegraphics[scale=0.23]{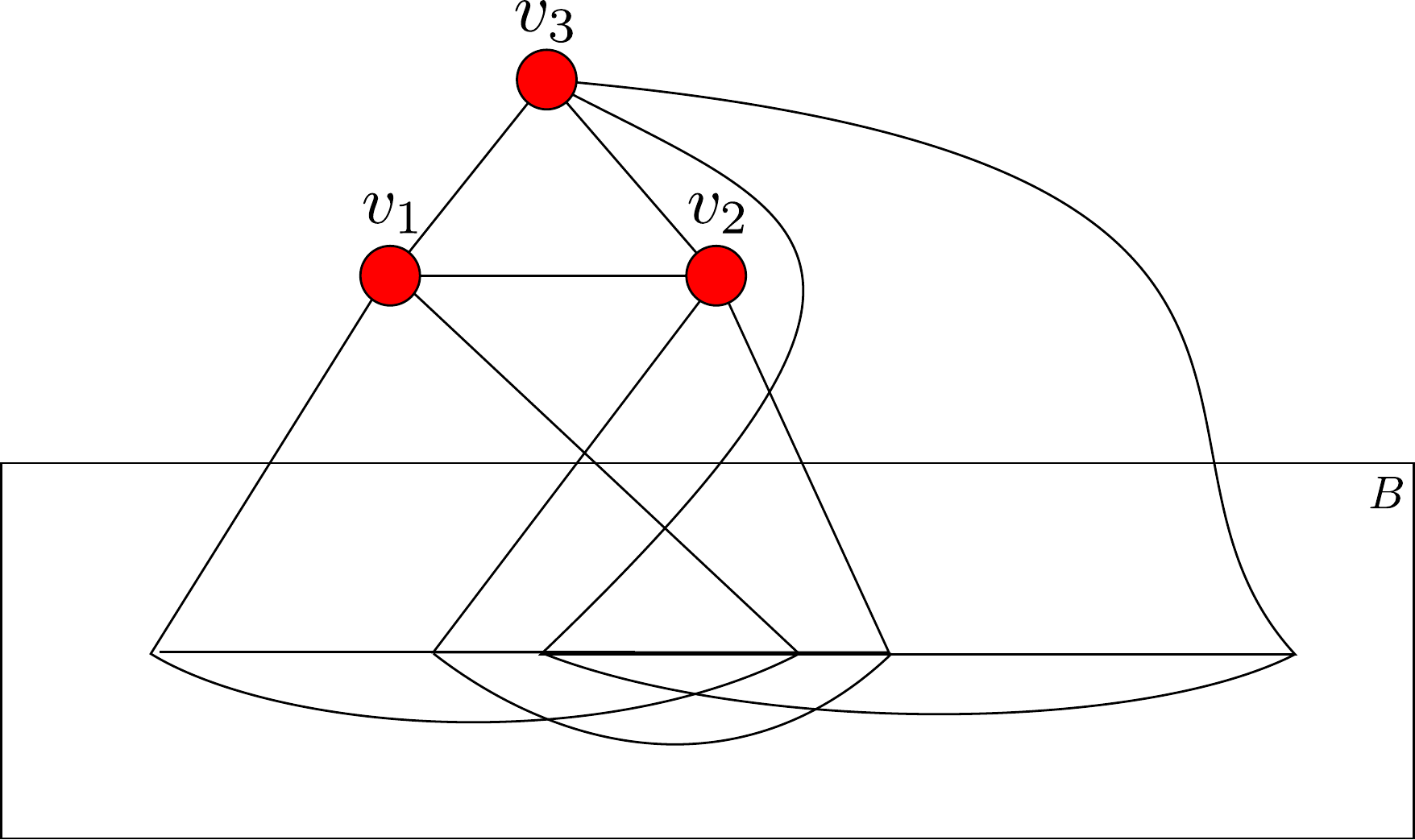}
			\label{fig:3kc}
		}
	\subfloat
		[Case $\kd$] {
			\includegraphics[scale=0.23]{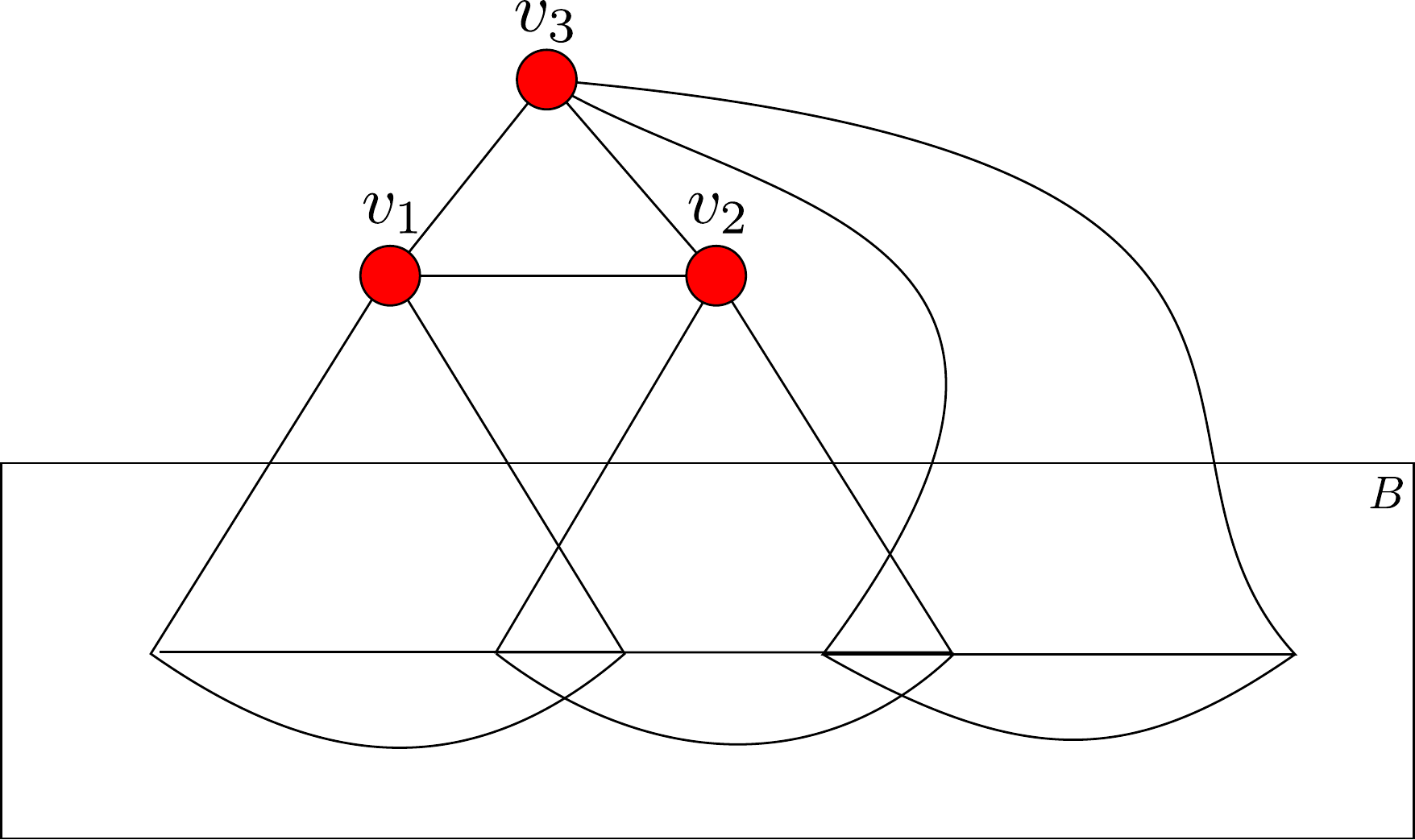}
			\label{fig:3kd}
		}
	\caption{A triangle satellite of an ($\alpha, \beta$)-polar graph}
	\label{fig:3k}
\end{figure}

\begin{figure}[t!]
\centering
	\subfloat
		[Case $\ka$] {
			\includegraphics[scale=0.23]{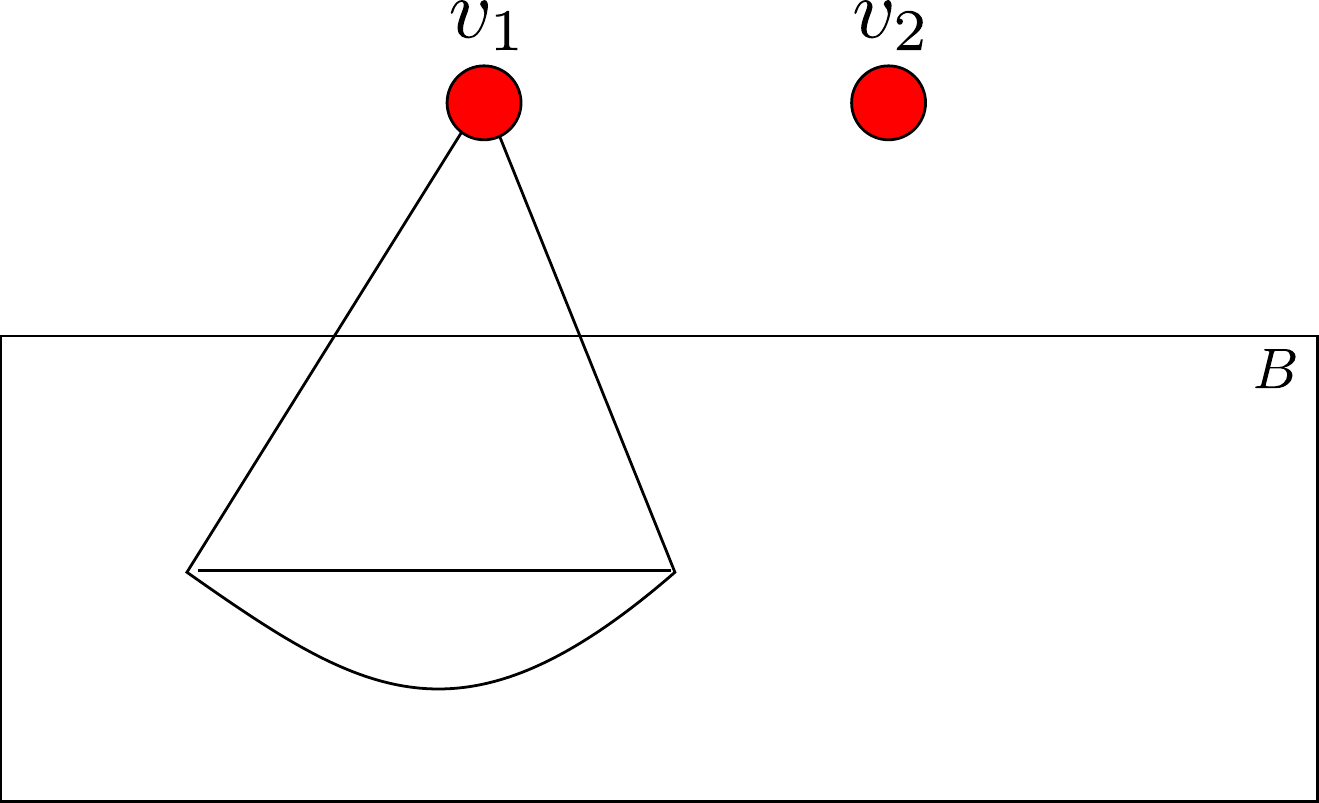}
			\label{fig:2ka}
		}
	\qquad 
	\subfloat
		[Case $\kb$] {
			\includegraphics[scale=0.23]{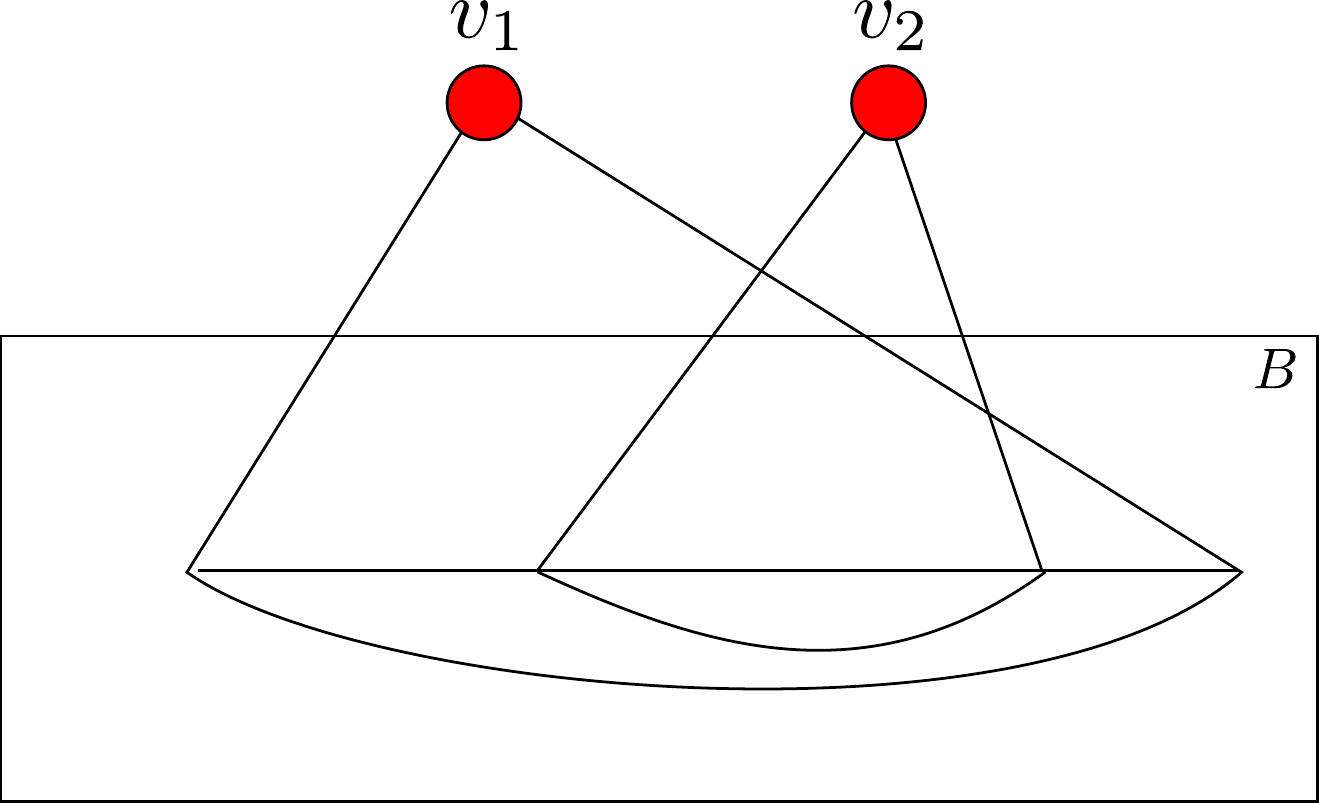}
			\label{fig:2kb}
		}
	\qquad 
	\subfloat
		[Case $\kc$] {
			\includegraphics[scale=0.23]{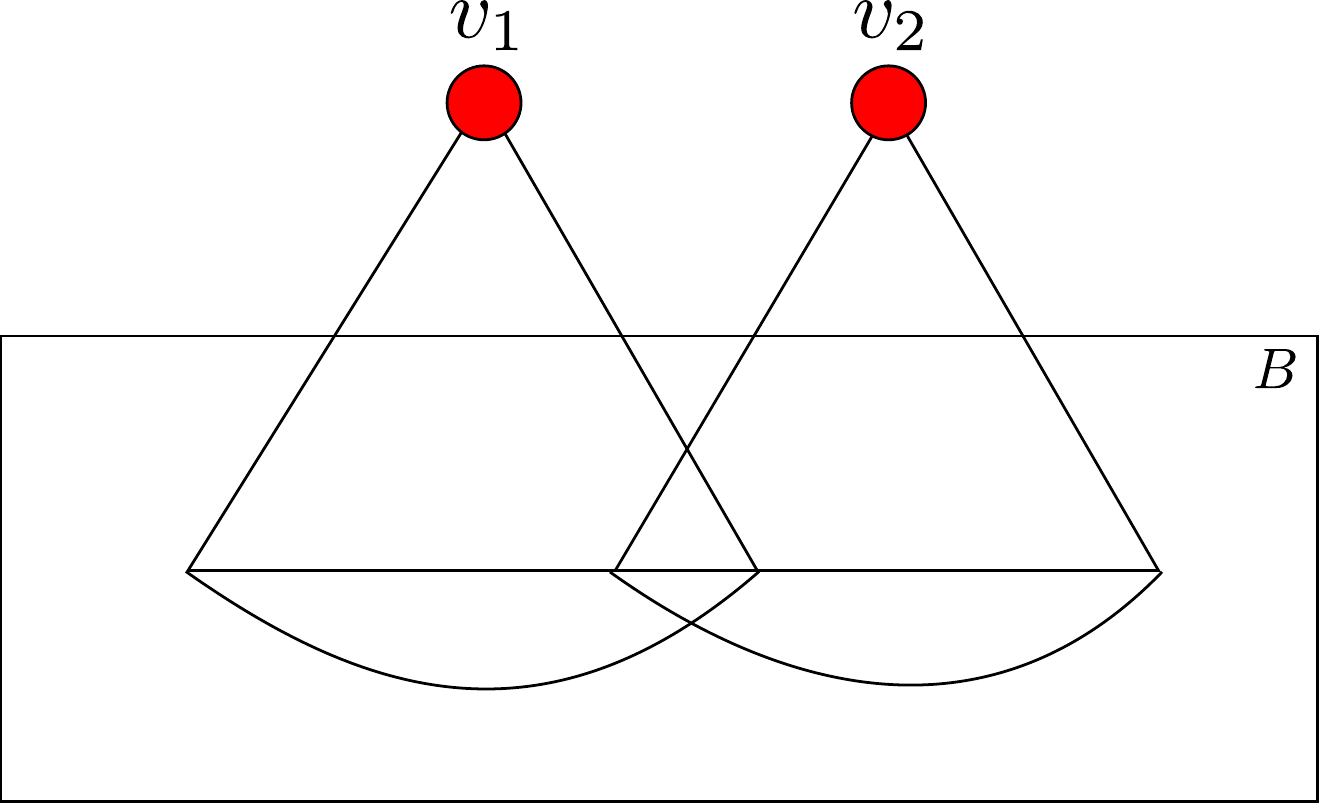}
			\label{fig:2kc}
		}
	\qquad 
	\subfloat
		[Case $\kd$] {
			\includegraphics[scale=0.23]{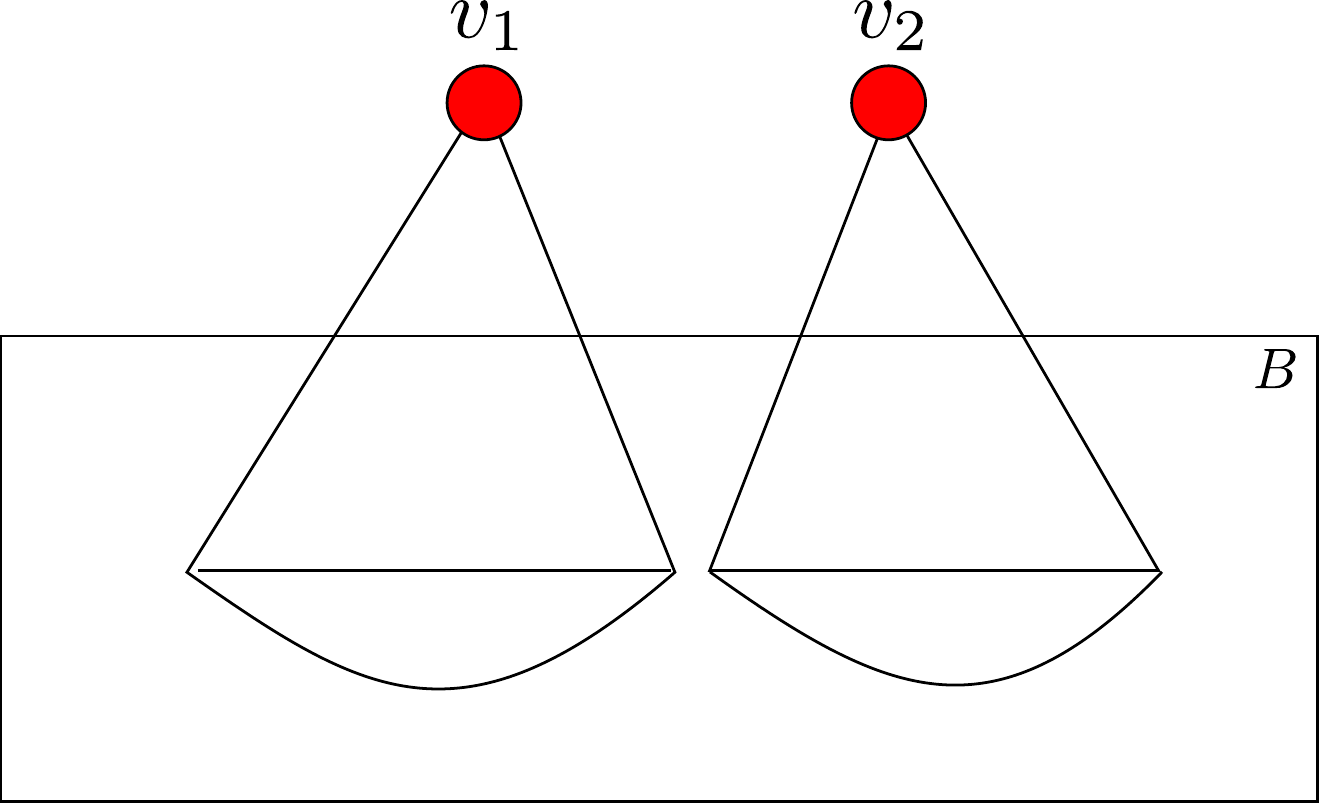}
			\label{fig:2kd}
		}
	\caption{An edge satellite of an ($\alpha, \beta$)-polar graph}
	\label{fig:2k}
\end{figure}

The following lemma is an important step to understand the role of triangles and
edges that are either in case $\ka$, $\kb$, $\kc$, or $\kd$ in a
2-clique-colouring of (3, 1)-polar and of (2, 1)-polar graphs. The following
lemma is also important towards the modification of a (3, 1)-polar graph to
obtain a (2, 1)-polar graph, which is closely related to
Theorem~\ref{thm:2cc21polar}.

\begin{lemma}
\label{lem:31polargraphs}
Let $G=(V, E)$ be a (3, 1)-polar graph, $\mathcal{K}$ be the set of satellites
of $G$ in case $\kd$, and $K \in \mathcal{K}$.
\begin{itemize}
  \item If $G$ has a 2-clique-colouring, then $\displaystyle\bigcup_{v \in K}
  N_B(v)$ is polychromatic.
  \item If $B$ has a 2-colouring that, for every $K^\prime \in
  \mathcal{K}$, $\displaystyle\bigcup_{v \in K^\prime} N_B(v)$ is polychromatic,
  then $G$ is 2-clique-colourable.
\end{itemize}
\end{lemma}
\begin{proof}
We begin by proving the former assertion. For the sake of
contradiction, suppose $\displaystyle \bigcup_{v \in K} N_B(v)$ is
monochromatic for some 2-clique-colouring of $G$. Without loss of generality,
suppose that all vertices in $\displaystyle \bigcup_{v \in K} N_B(v)$ have
colour 1. Hence, every vertex $v \in K$ has colour 2, otherwise $\{v\} \cup
N_B(v)$ is a monochromatic clique. On the other hand, if every vertex $v \in K$
has the same colour, then $K$ is a monochromatic clique, which is a
contradiction.

Now, we prove the latter assertion. For each singleton satellite $K$, we extend
the 2-colouring of $B$, as follows. Assign colour 1 to the vertex $v$ of $K$ if
there exists a vertex of $N_B(v)$ with colour 2. Otherwise, assign colour 2
to $v$. For each edge satellite $K$, we extend the 2-colouring of $G$, as
follows.

\begin{enumerate}
  \item $K$ is in case $\ka$. Let $u$ be a vertex of $K$, such that $N_B(u) =
  \emptyset$.  Let $v$ be a vertex of $K$, $u \neq v$.
  If $N_B(v) = \emptyset$, then assign colour 1 to $v$ and colour 2 to $u$.  
  Otherwise, i.e. $N_B(v) \neq \emptyset$, assign colour 1 to $v$ if there
  exists a vertex in $N_B(v)$ with colour 2, otherwise assign colour 2 to $v$.
  Moreover, give colour 1 to $u$, if vertex $v$ has colour 2, or else we assign
  colour 1 to $u$.
  \item $K$ is in case $\kb$. Let $v, w$ be distinct vertices of $K$, such that
  $N_B(v) \subseteq N_B(w)$. Assign colour 1 to $v$ if there exists a vertex of
  $N_B(v)$ with colour 1, otherwise assign colour 2 to $v$. Assign colour 2 to
  $w$ if vertex $v$ received colour 1, otherwise assign colour 1 to $w$.
  \item $K$ is in case $\kc$. Assign colour 1 to every vertex of $K$, if there
  exists a vertex in $\displaystyle \bigcap_{v \in K} N_B(v)$ with colour 2,
  otherwise assign colour 2 to every vertex of $K$.
  \item $K$ is in case $\kd$. Let $u, v$ be distinct vertices of $K$. Suppose
  $N_B(u)$ and $N_B(v)$ are monochromatic. Assign colour 1 to $u$ if the
  vertices of $N_B(u)$ have colour 2. Otherwise, assign colour 2 to $u$. Assign
  colour 1 to $v$ if the vertices of $N_B(v)$ have colour 2. Otherwise, assign
  colour 2 to $v$. Now, suppose, without loss of generality, that $N_B(u)$ is
  polychromatic. Then, assign colour 1 to $v$ if there exists a vertex of
  $N_B(v)$ with colour 2, otherwise assign colour 2 to $v$. Moreover, assign
  colour 2 to $u$ if vertex $v$ received colour 1, otherwise assign colour 2 to
  $u$.
\end{enumerate}
   
For each triangle satellite $K$, we extend the 2-colouring of $B$, as follows.

\begin{enumerate}
  \item $K$ is in case $\ka$. Let $u, v, w$ be the distinct vertices of $K$,
  such that $N_B(u) = \emptyset$. If $N_B(v) \cup N_B(w) = \emptyset$, then
  assign colours 1, 1, and 2 to $u$, $v$, and $w$ respectively. Now, if $N_B(v)
  \cup N_B(w)$ is monochomatic, then assign colour 1 to $v$ and $w$ if the
  colour assigned to every vertex of $N_B(v) \cup N_B(w)$ is 2, otherwise we
  assign colour 1 to $v$ and $w$. Moreover, assign colour 1 to $u$ if $v$ and
  $w$ have colour 2, otherwise assign colour 2 to $u$. Now, if $N_B(v) \cup
  N_B(w)$ is polychromatic, we have two cases. If $N_B(v) \cap N_B(w) \neq
  \emptyset$, then assign colour 1 to $v$ and $w$ if there exists a vertex in
  $N_B(v) \cap N_B(w)$ with colour 2, otherwise assign colour 1 to $v$ and
  $w$. Moreover, assign colour 1 to $u$ if $v$ and $w$ have colour 2, otherwise
  assign colour 2 to $u$. If $N_B(v) \cap N_B(w) = \emptyset$, assign colour
  1 to $v$ (resp. to $w$) if there exists a vertex of $N_B(v)$ (resp. $N_B(w)$)
  with colour 2, otherwise assign colour 2 to $v$ (resp. $w$). Moreover, assign
  colour 1 to $u$ if $v$ or $w$ have colour 2, otherwise assign colour 2 to
  $u$.
  \item $K$ is in case $\kb$. Let $v, w$ be distinct vertices of $K$, such that
  $N_B(v) \subseteq N_B(w)$. Assign colour 1 to $v$ if there exists a vertex in
  $N_B(v)$ with colour 1, otherwise assign colour 2 to $v$. Assign colour 2 to
  $w$ if vertex $v$ received colour 1, otherwise assign colour 1 to $w$. Let $u
  \in K \setminus \{v, w\}$. By hyphotesis, $N_B(u) \neq \emptyset$. Assign
  colour 1 to $u$ if there exists a vertex of $N_B(u)$ with colour 2, otherwise
  assign colour 2 to $u$.
  \item $K$ is in case $\kc$. Assign colour 1 to every vertex of $K$, if there
  exists a vertex of $\displaystyle \bigcap_{v \in K} N_B(v)$ with colour 2,
  otherwise assign colour 2 to every vertex of $K$.
  \item $K$ is in case $\kd$. Let $u, v, w$ be distinct vertices of $K$. Suppose
  $N_B(u)$, $N_B(v)$, and $N_B(w)$ are monochromatic. Assign colour 1 to $u$
  if the vertices of $N_B(u)$ have colour 2. Otherwise, assign colour 2 to $u$.
  Assign colour 1 to $v$ if the vertices of $N_B(v)$ have colour 2. Otherwise,
  assign colour 2 to $v$. Assign colour 1 to $w$ if the vertices of $N_B(w)$
  have colour 2. Otherwise, assign colour 2 to $w$. Now, suppose, without loss
  of generality, that $N_B(u)$ is polychromatic. Assign colour 1 to $v$ if there
  exists a vertex of $N_B(v)$ with colour 2, otherwise assign colour 2 to $v$.
  Assign colour 1 to $w$ if there exists a vertex of $N_B(w)$ with colour 2,
  otherwise assign colour 2 to $w$. Finally, assign colour 2 to $u$ if vertex
  $v$ and $w$ received colour 1, otherwise assign colour 2 to $u$.
\end{enumerate}

We invite the reader to check that the given colouring is a
2-clique-colouring to graph~$G$.
\end{proof}

For a given a (3, 1)-polar graph $G$, we proceed to obtain a (2, 1)-polar
graph $G^\prime$ that is 2-clique-colourable if, and only if, $G$ is
2-clique-colourable, as follows. For each triangle satellite, if it is in case
$\kd$, we replace it by an edge in which (i) both complete sets have the same
neighboorhod contained in $B$ and (ii) the edge is also in case $\kd$,
otherwise we just delete triangle $K$. See Fig.~\ref{fig:k1234} for examples. 
Such construction is done in polynomial-time and we depict it as
Algorithm~\ref{alg:reducion-np-2}. See Fig.~\ref{fig:31polar3cc} for an
application of Algorithm~\ref{alg:reducion-np-2} given as input a (3,1)-polar
graph, which is not (2,1)-polar, with clique-chromatic number 3.
Algorithm~\ref{alg:reducion-np-2} and Theorem~\ref{thm:nptocheck} imply the
following theorem.

\begin{figure}[t!]
\centering
	\subfloat
		{
			\includegraphics[scale=0.22]{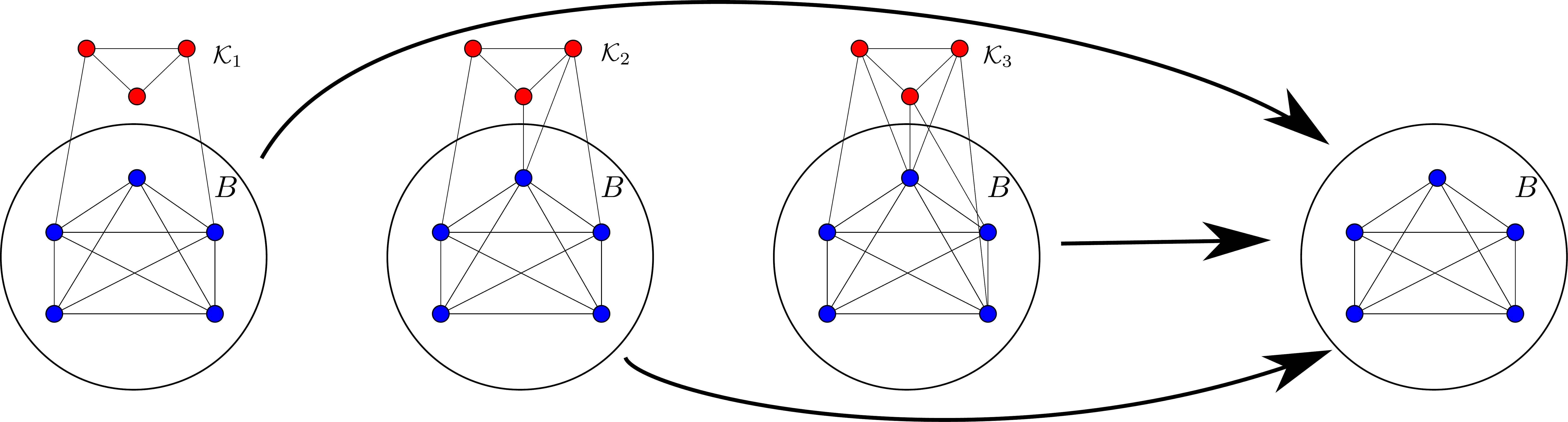}
			\label{fig:k123}
		}
	\qquad
	\subfloat
		{
			\includegraphics[scale=0.22]{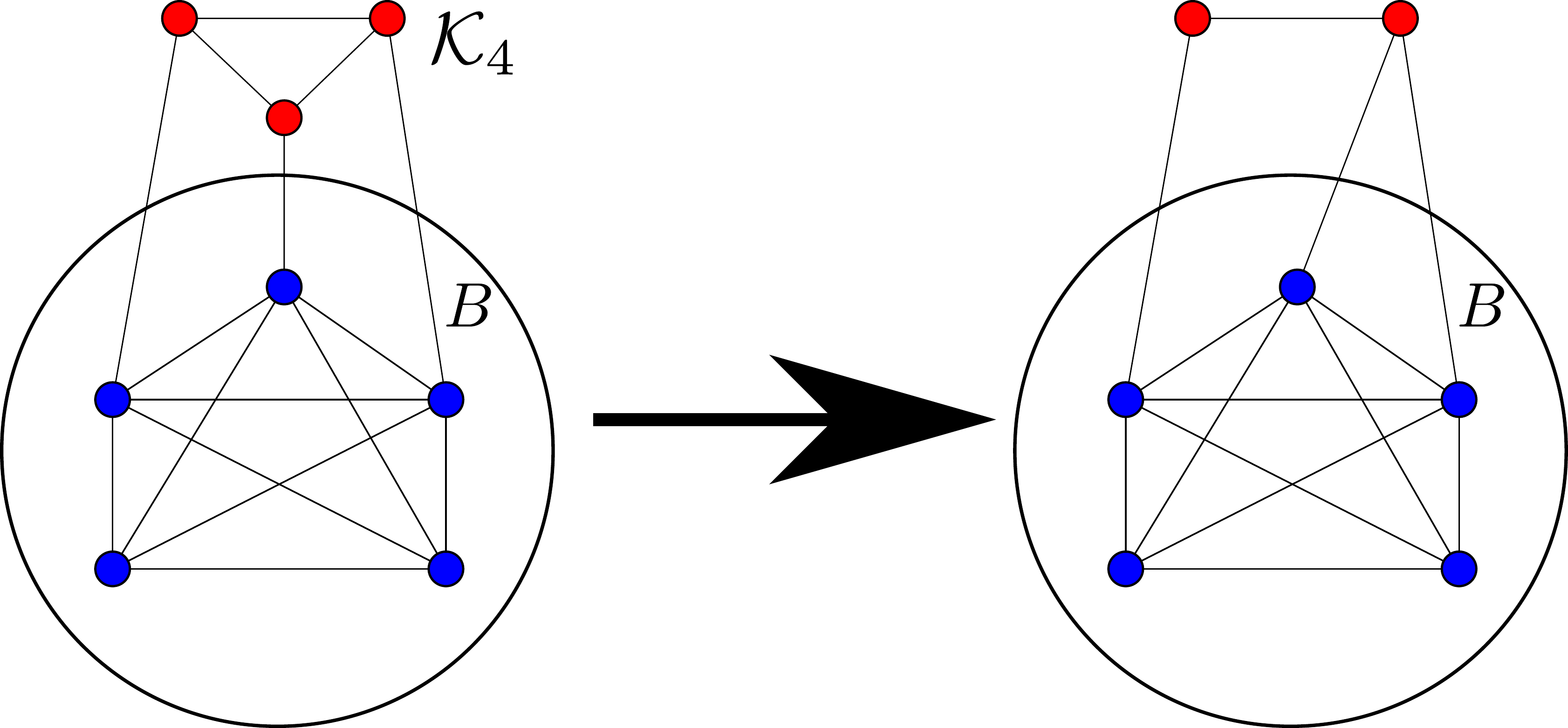}
			\label{fig:k4}
		}
	\caption{An iteration to obtain a (2, 1)-polar graph $G^\prime$, given a
	(3, 1)-polar graph $G$, such that $G$ is 2-clique-colourable if and
	only if $G^\prime$ is 2-clique-colourable}
	\label{fig:k1234}
\end{figure}

	\begin{algorithm}[h]
	\SetKwInOut{Input}{input}
	\SetKwInOut{Output}{output}
	\Input{$G = (A, B)$, a (3, 1)-polar graph.}
	\Output{$G^\prime$, a (2, 1)-polar graph that is 2-clique-colourable if, and
	only if, $G$ is 2-clique-colourable.} 
	\caption{An $O(n^2)$-time algorithm to output a (2, 1)-polar graph $G^\prime$,
	such that graph $G$ is 2-clique-colourable if, and only if, graph $G^\prime$ is
	2-clique-colourable.}
	\BlankLine
	\Begin
	{	
		\ForEach
		{
			satellite $K = \{v_1, v_2, v_3\}$
		}
		{
			$V^\prime \longleftarrow \emptyset$\;
			$E^\prime \longleftarrow \emptyset$\;
			\If{$K$ is in case $\kd$}
			{
					$V^\prime \longleftarrow \{u_1, u_2\}$\;
					$E^\prime \longleftarrow \{(u_1, u_2)\}$\;
					$E^\prime \longleftarrow E^\prime \cup \{(u_1, x) \mid x \in N_B(v_1)\}$\; 
					$E^\prime \longleftarrow E^\prime \cup \{(u_2, x) \mid x \in ((N_B(v_2)
					\cup N_B(v_3)) \setminus N_B(v_1))\}$\; 
			}
			$G \longleftarrow G[V(G) \setminus K]$\;
			$V(G) \longleftarrow V(G) \cup V^\prime$\;
			$E(G) \longleftarrow E(G) \cup E^\prime$\; 
			
	    }
	    \Return{$G$\;}
	}
	\label{alg:reducion-np-2}
	\end{algorithm}

\begin{figure}[t!]
\centering
	\subfloat
		[A (3,1)-polar graph, which is not (2,1)-polar, with clique-chromatic
		number 3.] {
			\hspace{1.2cm}
			\includegraphics[scale=0.2]{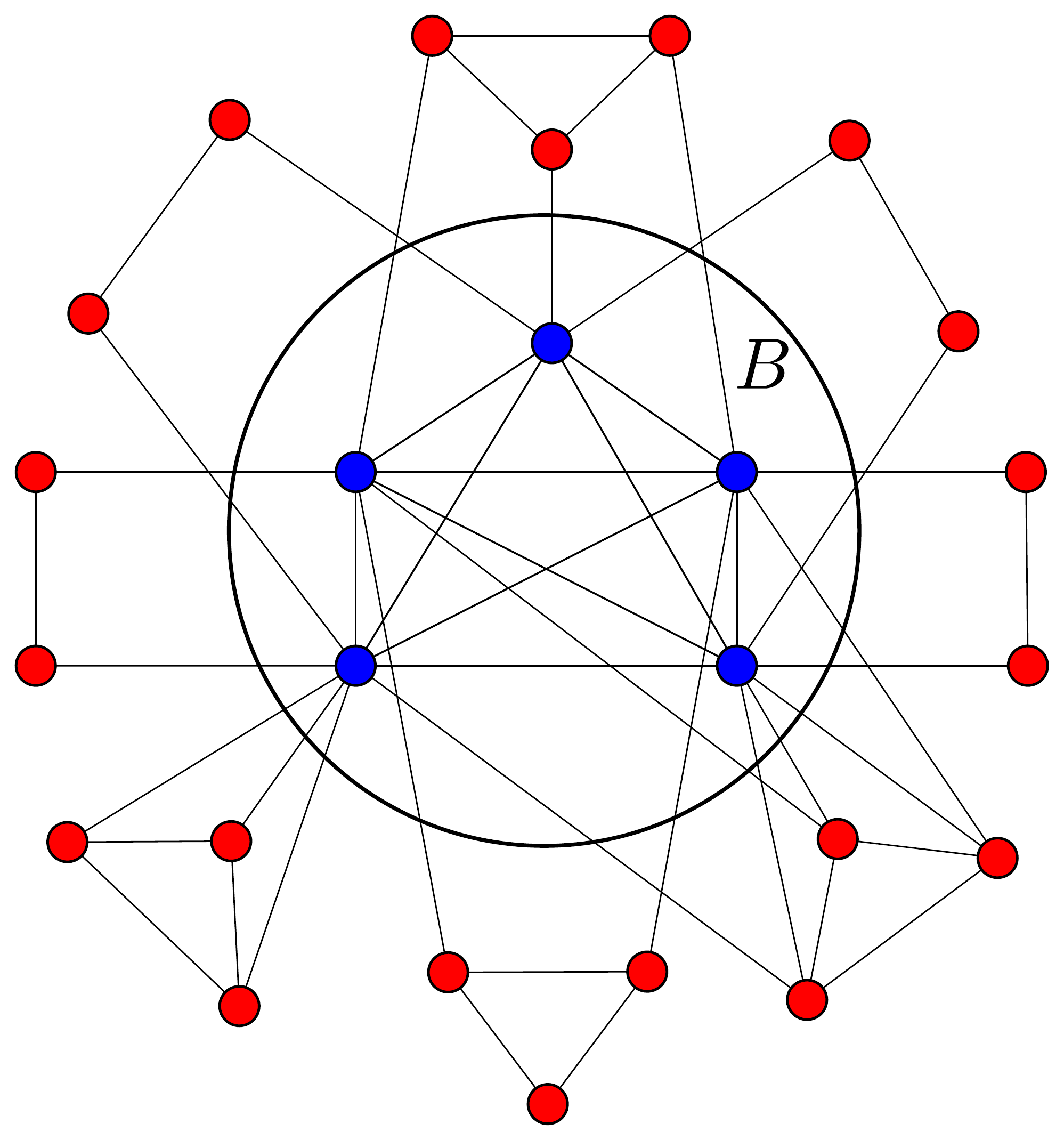}
			\hspace{1.2cm}
			\label{fig:31polar3cc-0}
		}
	\qquad 
	\subfloat
		[A satellite $K_1$ that belongs to $\ka$.] {
			\hspace{1.2cm}
			\includegraphics[scale=0.2]{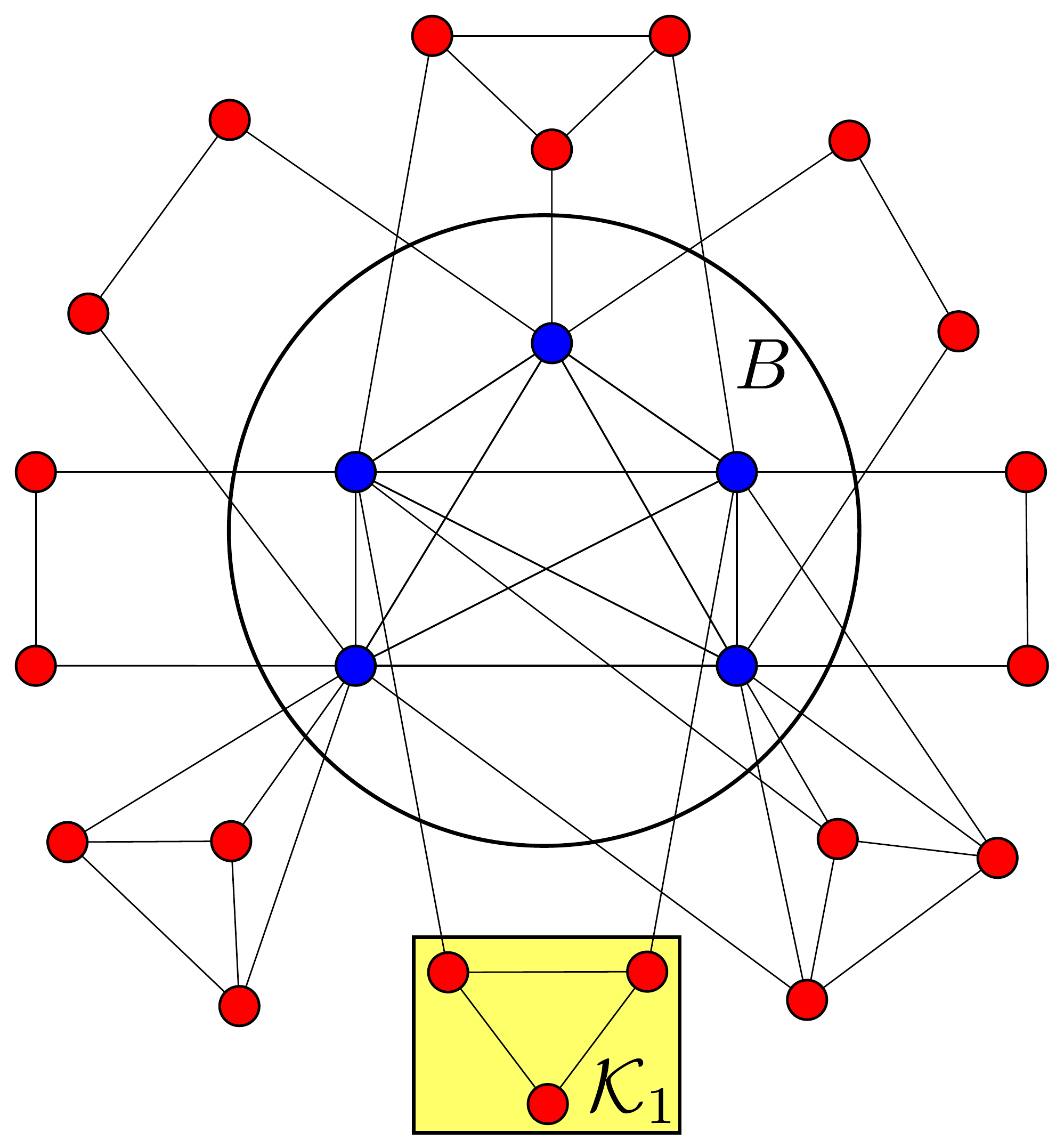}
			\hspace{1.2cm}
			\label{fig:31polar3cc-1}
		}
	\qquad 
	\subfloat
		[A satellite $K_2$ that belongs to $\kb$.] {
			\hspace{1.2cm}
			\includegraphics[scale=0.2]{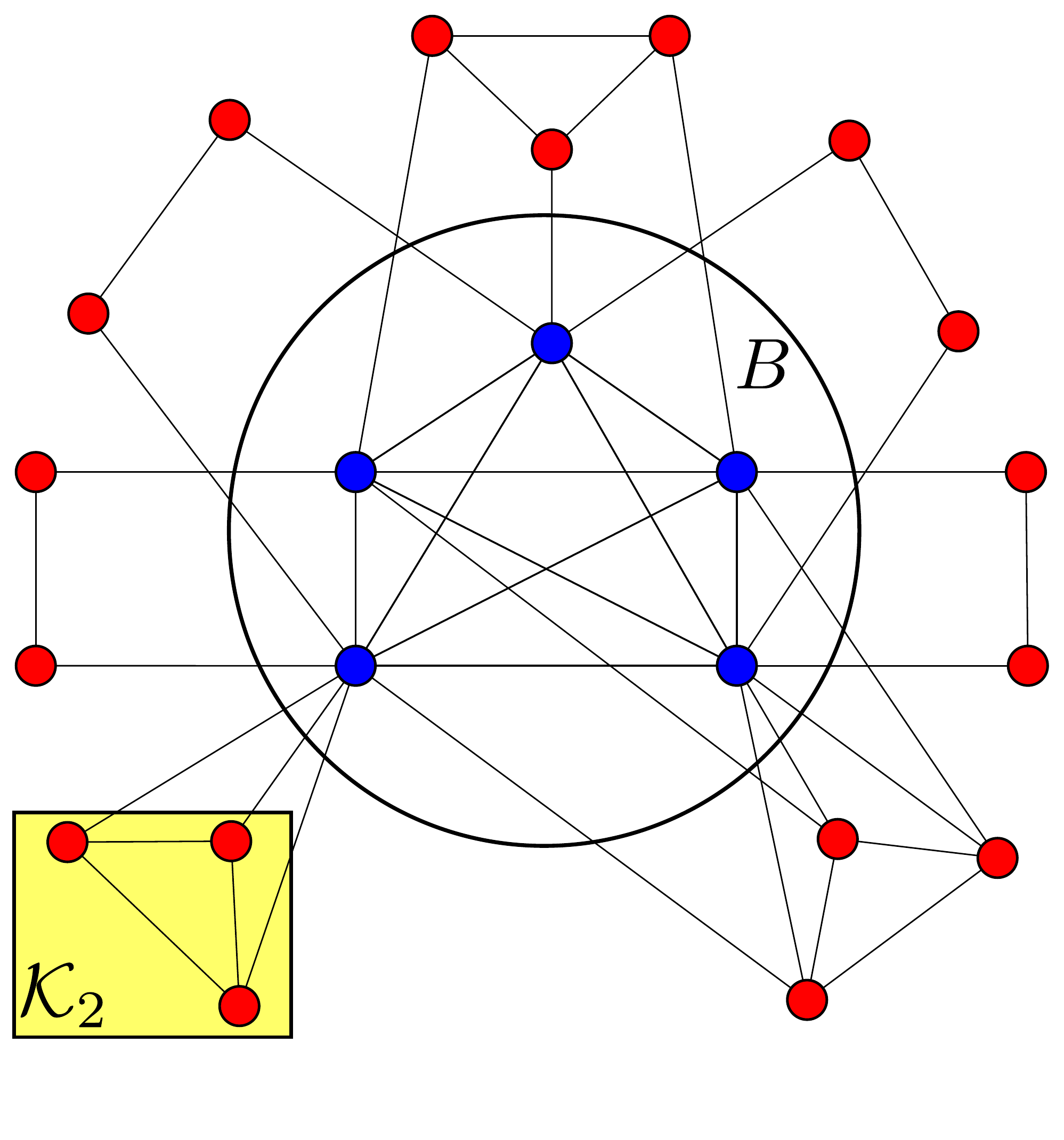}
			\hspace{1.2cm}
			\label{fig:31polar3cc-2}
		}
	\qquad 
	\subfloat
		[A satellite $K_3$ that belongs to $\kc$.] {
			\hspace{1.2cm}
			\includegraphics[scale=0.2]{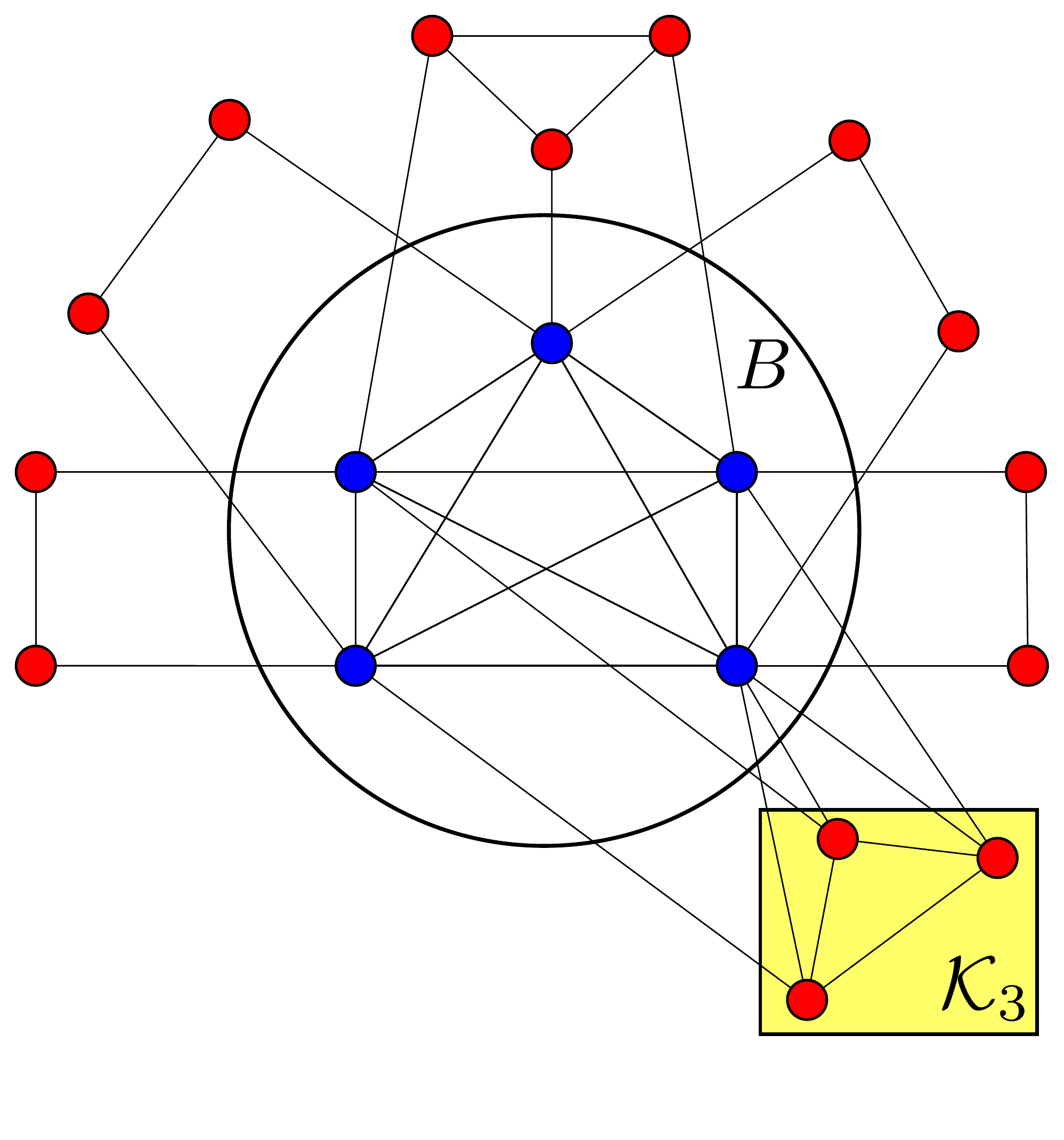}
			\hspace{1.2cm}
			\label{fig:31polar3cc-3}
		}
	\qquad 
	\subfloat
		[A satellite $K_4$ that belongs to $\kd$.] {
			\hspace{1.2cm}
			\includegraphics[scale=0.2]{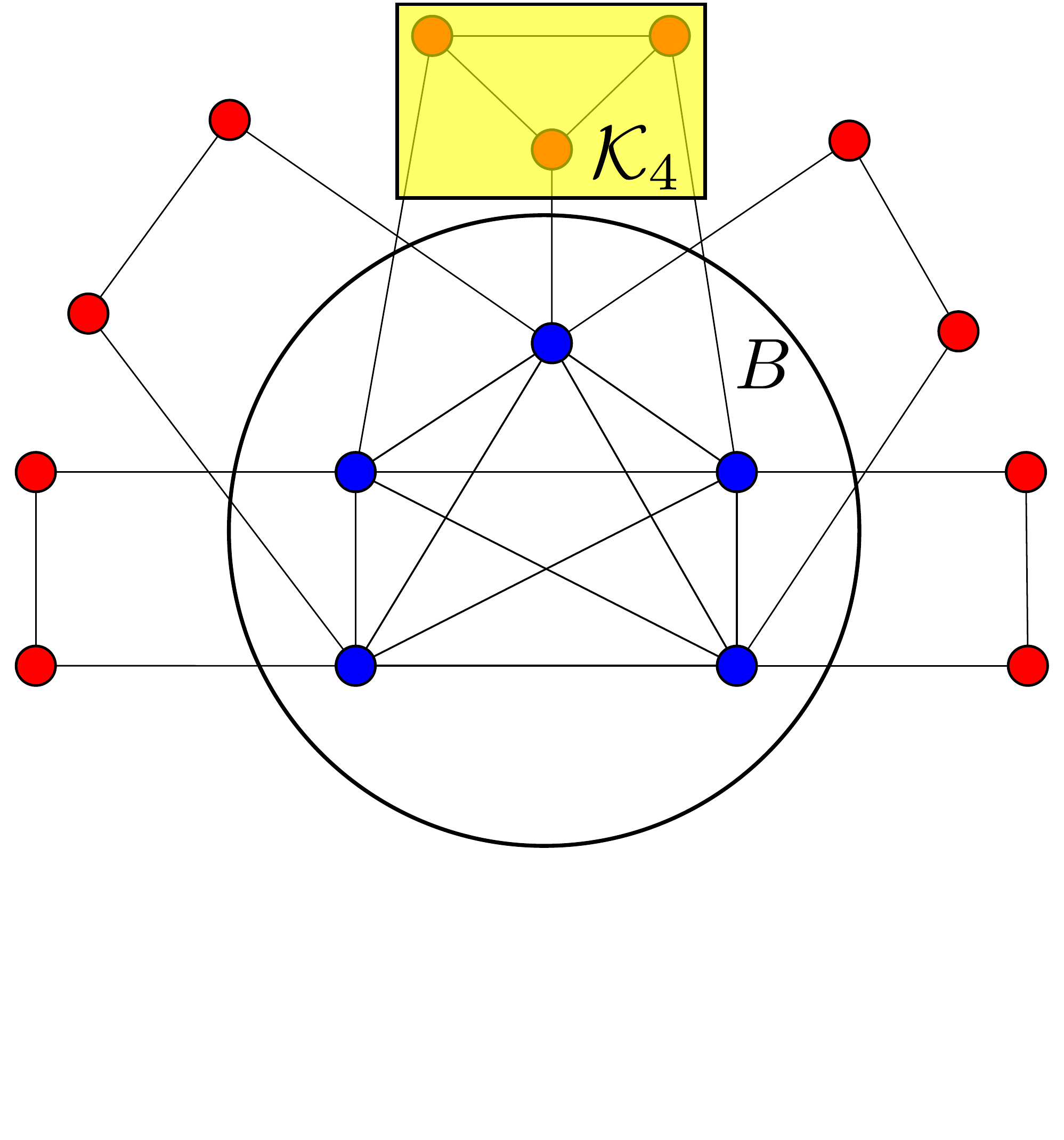}
			\hspace{1.2cm}
			\label{fig:31polar3cc-4}
		}
	\qquad 
	\subfloat
		[A (2,1)-polar graph with clique-chromatic number 3.] {
			\hspace{1.2cm}
			\includegraphics[scale=0.2]{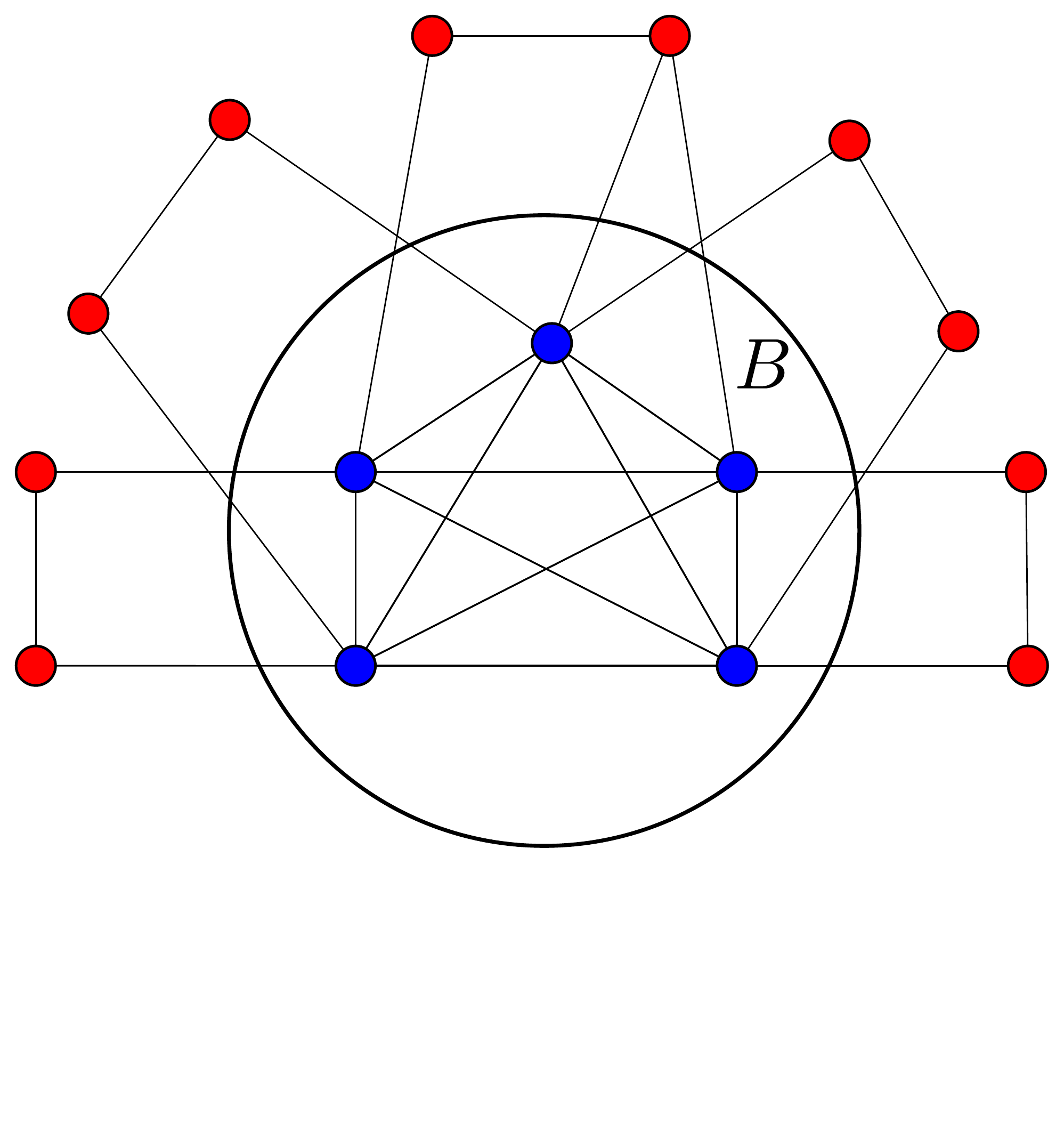}
			\hspace{1.2cm}
			\label{fig:31polar3cc-5}
		}
	\caption{Application of Algorithm~\ref{alg:reducion-np-2} given as input a
	(3,1)-polar graph, which is not (2,1)-polar, with clique-chromatic number 3.}
	\label{fig:31polar3cc}
\end{figure}

\begin{theorem}
\label{thm:2cc21polar}
The problem of 2-clique-colouring is $\mathcal{NP}$-complete for
(2, 1)-polar graphs.
\end{theorem}
\begin{proof}
The problem of 2-clique-colouring a (2, 1)-polar graph is in $\mathcal{NP}$:
Theorem~\ref{thm:nptocheck} confirms that it is in $\mathcal{P}$ to check
whether a colouring of a (2, 1)-polar graph is a 2-clique-colouring.

We claim that Algorithm~\ref{alg:reducion-np-2} is an $O(n^2)$-time
algorithm to output a (2, 1)-polar graph $G^\prime$, given a (3, 1)-polar graph
$G$, such that $G^\prime$ is 2-clique-colourable if, and only if, graph $G$ is
2-clique-colourable.

Without loss of generality, suppose that $B$ is a (maximal) clique of the graph.
We use induction on the number of triangle satellites of $G$ in order to prove
that $G$ is 2-clique-colourable if, and only if, $G^\prime$ is
2-clique-colourable. For the sake of conciseness, we prove only one way of the
basis induction. The converse and the step induction follows analogously.

Suppose that there exists only one triangle satellite $K = \{v_1, v_2, v_3\}$ of
$G$.
Let $K^\prime = \{u_1, u_2\}$ be the edge that replaced satellite $K$ to obtain
graph $G^\prime$. Suppose that there exists a 2-clique-colouring $\pi$ of $G$.
Every clique of $G$ containing a subset of a satellite $S \neq K$ of $G$ is
polychromatic. Let $\pi^\prime$ be a 2-colouring of $G^\prime$, such that
$\pi^\prime(v) = \pi(v)$, for each $v \in V(G) \cap V(G^\prime)$. Hence, every
clique of $G^\prime$ containing a subset of a satellite $S^\prime \neq K^\prime$
of $G$ is polychromatic. If $B(G)$ is a clique of $G$, then $B$ is
polychromatic. Clearly, $B(G^\prime)$ is polychromatic. Hence, we are left to
prove that every clique of $G^\prime$ containing a subset of $K^\prime$ is
polychromatic. If $K$ is either in case $\ka$, $\kb$ or $\kc$, then we are done.
If $K$ is in case $\kd$, then $\displaystyle \bigcup_{v \in K} N_B(v)$ is
polychromatic. Since $\displaystyle \bigcup_{v \in K} N_B(v) = \displaystyle
\bigcup_{v \in K^\prime} N_B(v)$ and $\pi^\prime(v) = \pi(v)$, for each $v \in
B$, $\displaystyle \bigcup_{v \in K^\prime} N_B(v)$ is polychromatic. By
Lemma~\ref{lem:31polargraphs}, $G^\prime$ is 2-clique-colourable.
\end{proof}


As a remark, we noticed strong connections between hypergraph 2-colorability and
2-clique-colouring (2, 1)-polar graphs. Indeed, we have a simpler alternative
proof showing that 2-clique-colouring (2, 1)-polar graphs is
$\mathcal{NP}$-complete by a reduction from hypergraph 2-colouring.
In constrast to graphs, deciding if a given hypergraph is 2-colourable is
an $\mathcal{NP}$-complete problem, even if all edges have cardinality at
most~3~\cite{MR0363980}. The reader may ask why we did not
exploit only the alternative proof that is quite shorter than the original
proof. The reason is related to be consistent with the next section, where we
show that even restricting the size of the cliques to be at least 3, the 2-clique-colouring of
(3, 1)-polar graphs is still $\mathcal{NP}$-complete, while 2-clique-colouring
of (2, 1)-polar graphs becomes a problem in $\mathcal{P}$.

\def\proofname{Alternative proof of Theorem~\ref{thm:2cc21polar}}
\begin{proof}
The problem of 2-clique-colouring a (2, 1)-polar graph is in $\mathcal{NP}$:
Theorem~\ref{thm:nptocheck} confirms that it is in $\mathcal{P}$ to check
whether a colouring of a (2, 1)-polar graph is a 2-clique-colouring. 

We prove that 2-clique-colouring (2, 1)-polar graphs is $\mathcal{NP}$-hard by
reducing hypergraph 2-colouring to it. The outline of the proof follows. For
every hypergraph $\mathcal{H}$, a (2, 1)-polar graph~$G$ is constructed such
that $\mathcal{H}$ is 2-colourable if, and only if, graph $G$ is
2-clique-colourable. Let~$n$ (resp. $m$) be the number of hypervertices (resp.
hyperedges) in hypergraph $\mathcal{H}$. We define graph $G$, as follows.

\begin{itemize}
	\item for each hypervertex $v_i$, $1 \leq i \leq n$, we create a vertex $v_i$
	in $G$, so that the set $\{v_1, \ldots, x_n\}$ induces a complete subgraph of
	$G$, which is the partition $B$ of graph $G$;
	\item for each hyperedge $e_j = \{v_1, \ldots, v_l\}$, $1 \leq j \leq m$,  
	we create two vertices $u_{j_{1}}$ and $u_{j_{2}}$. Moreover, we create edges
	$u_{j_{1}} v_1, \ldots, u_{j_{1}} v_{l-1}$, and $u_{j_{2}} v_l$ so that 
    $\{u_{j_{1}} u_{j_{2}}\}$ is a satellite in case~$\kd$.
\end{itemize}

Clearly, $G$ is a (2, 1)-polar graph and such construction is done in 
polynomial-time. Refer to Fig.~\ref{fig:reduction-np-hypergraph} for an example
of such construction.

\begin{figure}[t!]
\centering
	\subfloat
		[Hypergraph instance and its corresponding constructed (2, 1)-polar graph] {
			\includegraphics[scale=0.3]{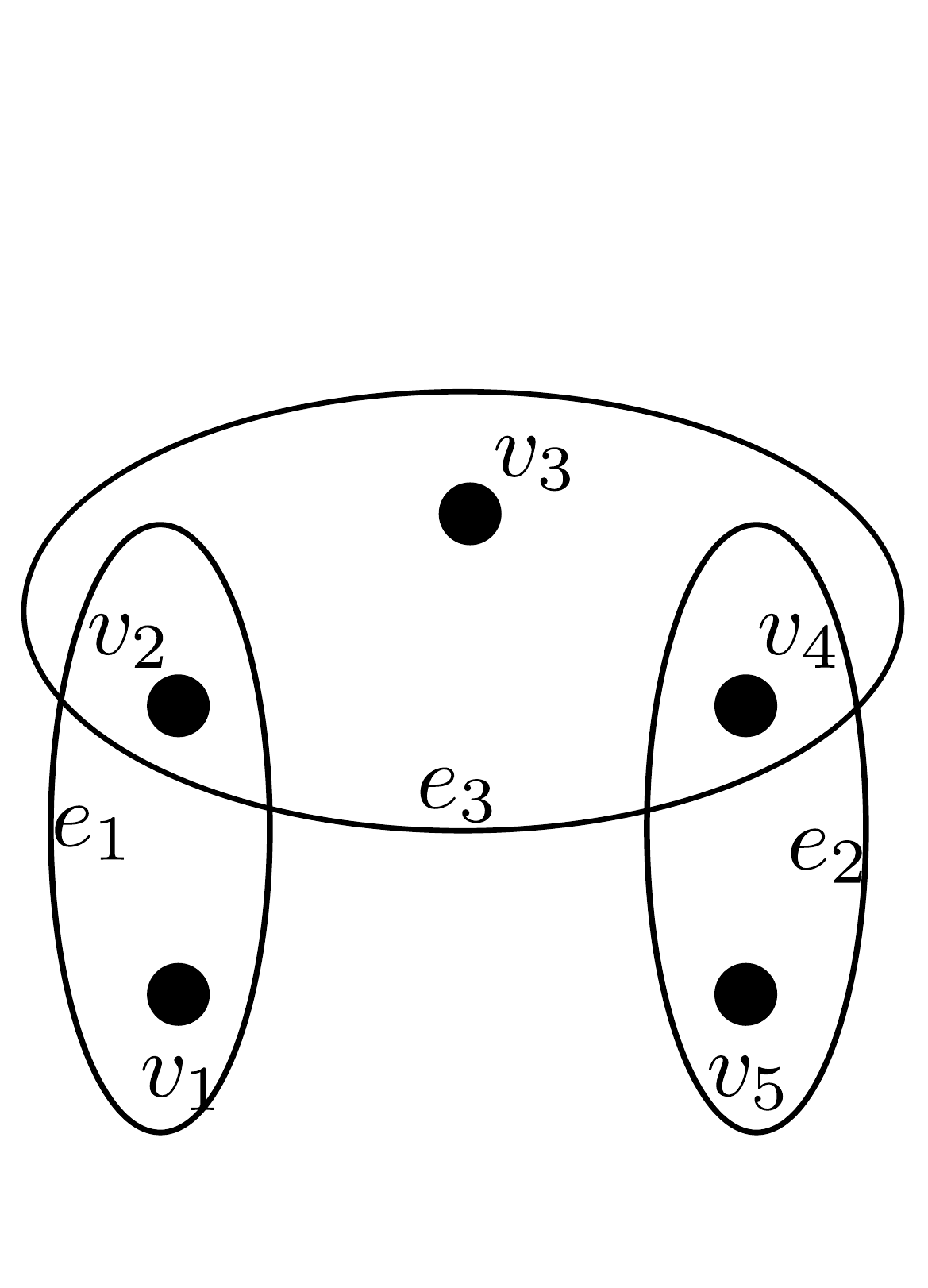}
			\hspace{2cm}
			\includegraphics[scale=0.3]{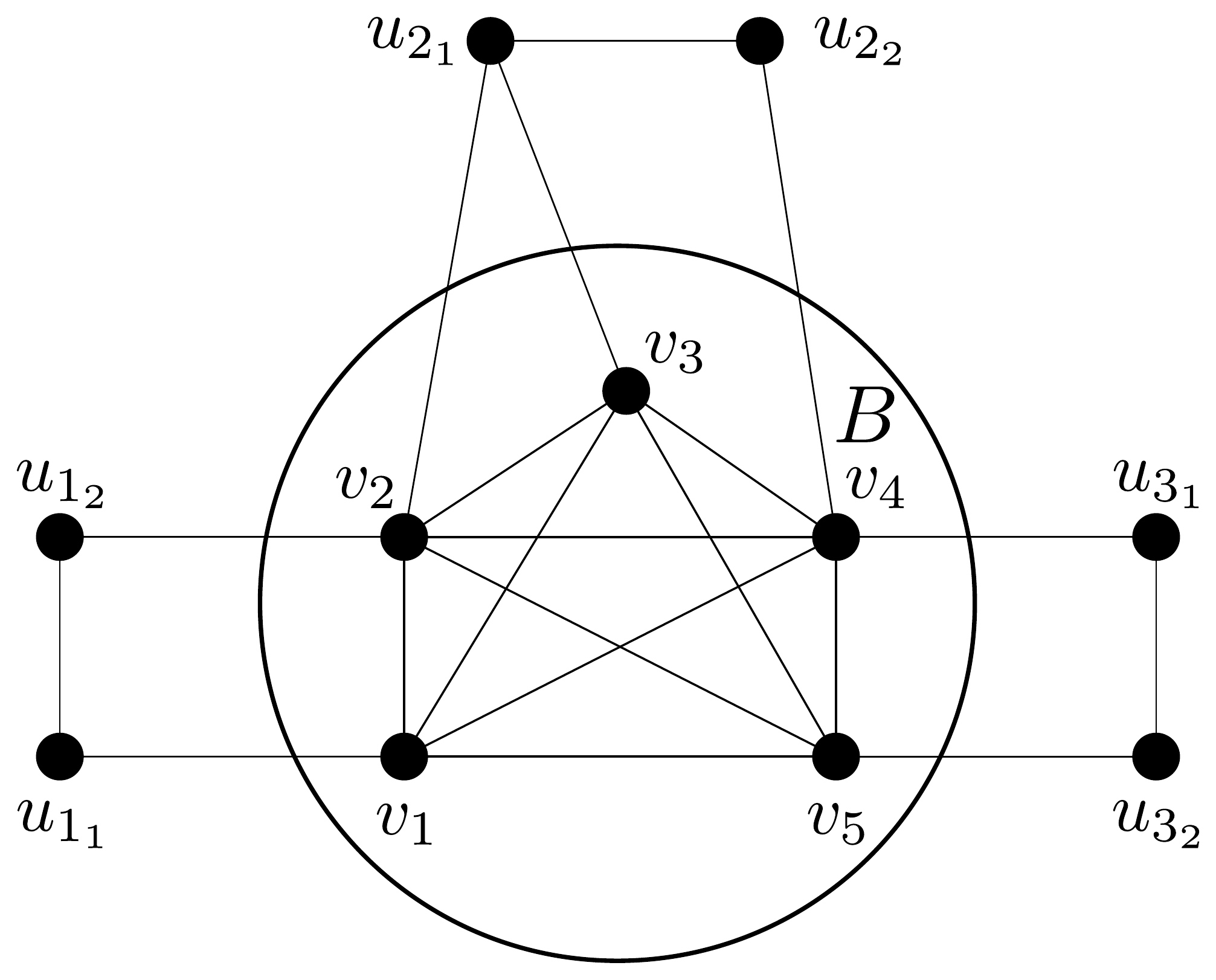}
			\label{fig:reduction-np-hypergraph-uncoloured}
		}
	\qquad 
	\subfloat
		[A proper 2-colouring assignment of the hypergraph and its corresponding
		2-colouring assignment of the constructed graph, as well as a
		2-clique-colouring assignment of the constructed graph and its corresponding
		2-colouring assignment of the hypergraph] {
			\includegraphics[scale=0.3]{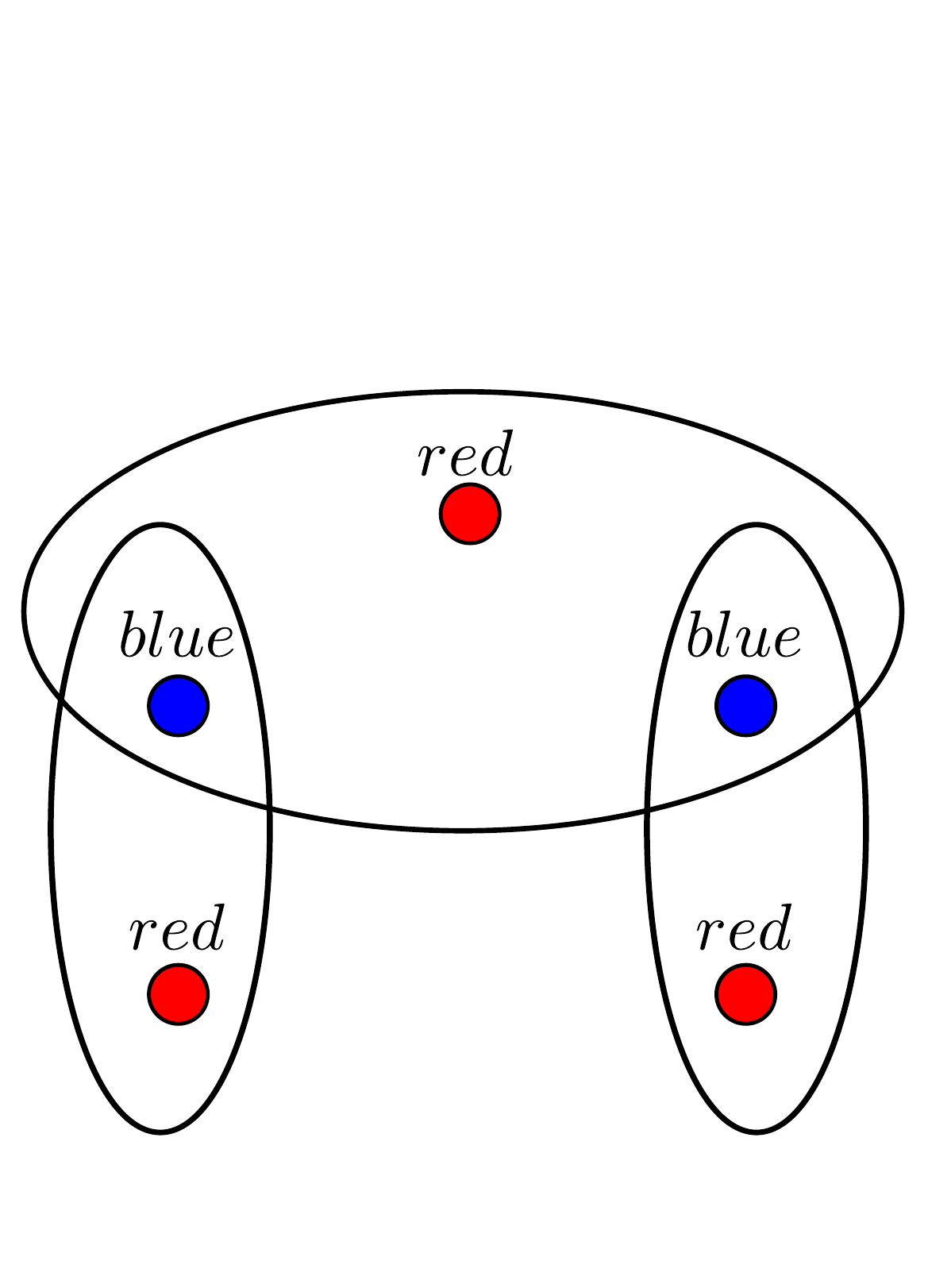}
			\hspace{2cm}
			\includegraphics[scale=0.3]{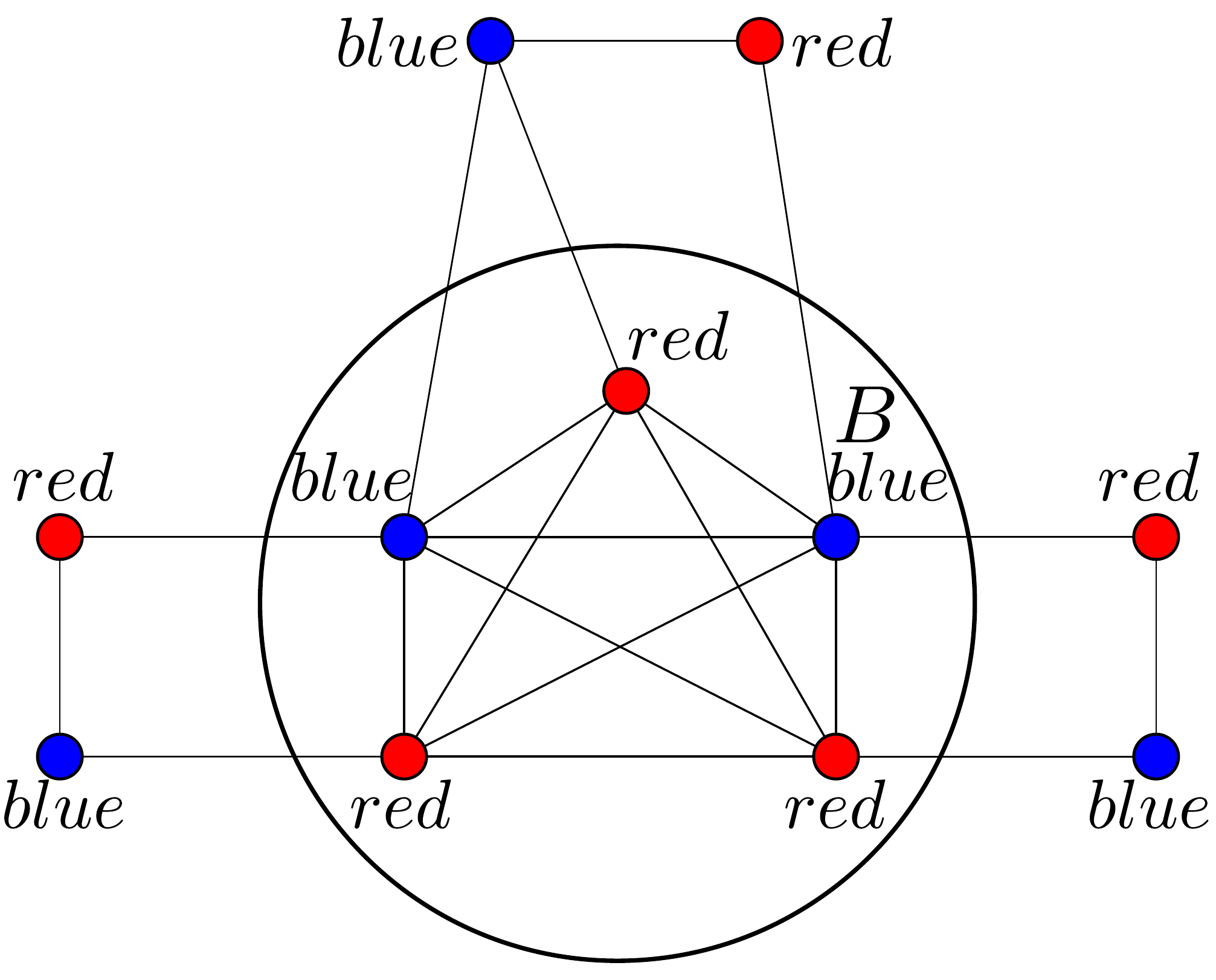}
			\label{fig:reduction-np-hypergraph-coloured}
		}
	\caption{Example of a (2, 1)-polar graph constructed for a given hypergraph
	instance}
	\label{fig:reduction-np-hypergraph}
\end{figure}

We claim that hypergraph $\mathcal{H}$ is 2-colourable if, and only if, graph
$G$ is 2-clique-colourable. Assume that there exists a proper 2-colouring $\pi$
of $\mathcal{H}$. We give a colouring to the graph $G$, as follows.
 
\begin{itemize}
	\item assign colour $\pi(v)$ for each $v$ of partition $B$,
	\item extend the 2-clique-colouring for each clique ($\{u_{j_{1}},
	u_{j_{2}}\}$) that is a satellite of $G$.
\end{itemize}

It still remains to be proved that this is indeed a 2-clique-colouring. Consider
the partition $B = \{x_1, \overline{x}_1, \ldots, x_n, \overline{x}_n\}$.
Clearly, the above colouring assigns 2 colours to this set.
Each satellite $K$ of $G$ is in case $\kd$ and $\displaystyle\bigcup_{v \in K}
N_B(v)$ is polychromatic, since $\displaystyle\bigcup_{v \in K} N_B(v) = e_j$.
By Lemma~\ref{lem:31polargraphs}, graph $G$ is 2-clique-colourable.

For the converse, we now assume that $G$ is 2-clique-colourable and we consider
any 2-clique-colouring $\pi^\prime$ of $G$. We give a colouring to hypergraph
$\mathcal{H}$, as follows. Assign colour $\pi^\prime(v)$ for each hypervertex
$v$. By Lemma~\ref{lem:31polargraphs} $\displaystyle\bigcup_{v \in K} N_B(v)$
is polychromatic for each satellite $K$ of $G$. Then, hypergraph $\mathcal{H}$
is 2-colourable, since $\displaystyle\bigcup_{v \in K} N_B(v) = e_j$ for every
hyperedge $e_j$.
\end{proof}
\def\proofname{Proof}


\section{Restricting the size of the cliques}
\label{sec:restrictingthesize}


Kratochv\'il and Tuza~\cite{Kratochvil} are interested in determining the
complexity of 2-clique-colouring of perfect graphs with all cliques having size
at least 3. We determine what happens with the complexity of 2-clique-colouring
of (2, 1)-polar graphs, of (3, 1)-polar graphs, and of weakly chordal graphs,
respectively, when all cliques are restricted to have size at least 3. The
latter result address Kratochv\'il and Tuza's question.


Given graph $G$ and $b_1, b_2, b_3 \in V(G)$, we say that we add to $G$ a copy
of an auxiliary graph $BP(b_1, b_2, b_3)$ of order $6$  -- depicted in
Fig.~\ref{fig:bpolychromatic} -- if we change the definition of~$G$ by doing the
following: we first change the definition of $V$ by adding to it copies of the
vertices $a_1$, $a_2$, $a_3$ of the auxiliary graph $BP(b_1, b_2, b_3)$; second,
we change the definition of $E$ by adding to it copies of the edges~$(u, v)$ of
$BP(b_1, b_2, b_3)$. 

Similarly, given a graph~$G$ and $b_1, b_2 \in V(G)$, we
say that we add to $G$ a copy of an auxiliary graph $BS(b_1, b_2)$ of order
$17$  -- depicted in Fig.~\ref{fig:bswitch} -- if we change the definition of
$G$ by doing the following: we first change the definition of $V$ by adding to
it copies of the vertices $b^{\prime}$, $b^{\prime\prime}$,
$b^{\prime\prime\prime}$ of the auxiliary graph $BS(b_1, b_2)$; second, we
change the definition of $E$ by adding to it edges so that $B(G)
\cup \{b_1, b_2, b^{\prime}, b^{\prime\prime}, b^{\prime\prime\prime}\}$ is a
complete set; finally, we add copies of the auxiliary graphs $BP(b_1, b_2,
b^{\prime})$, $BP(b_1, b_2, b^{\prime\prime})$, $BP(b_1, b_2,
b^{\prime\prime\prime})$, $BP(b^{\prime}, b^{\prime\prime},
b^{\prime\prime\prime})$.

\begin{figure}[t!]
\centering
	\subfloat
		[Auxiliary graph $BP(b_1, b_2, b_3)$] {
			\hspace{1cm}
			\includegraphics[scale=0.34]{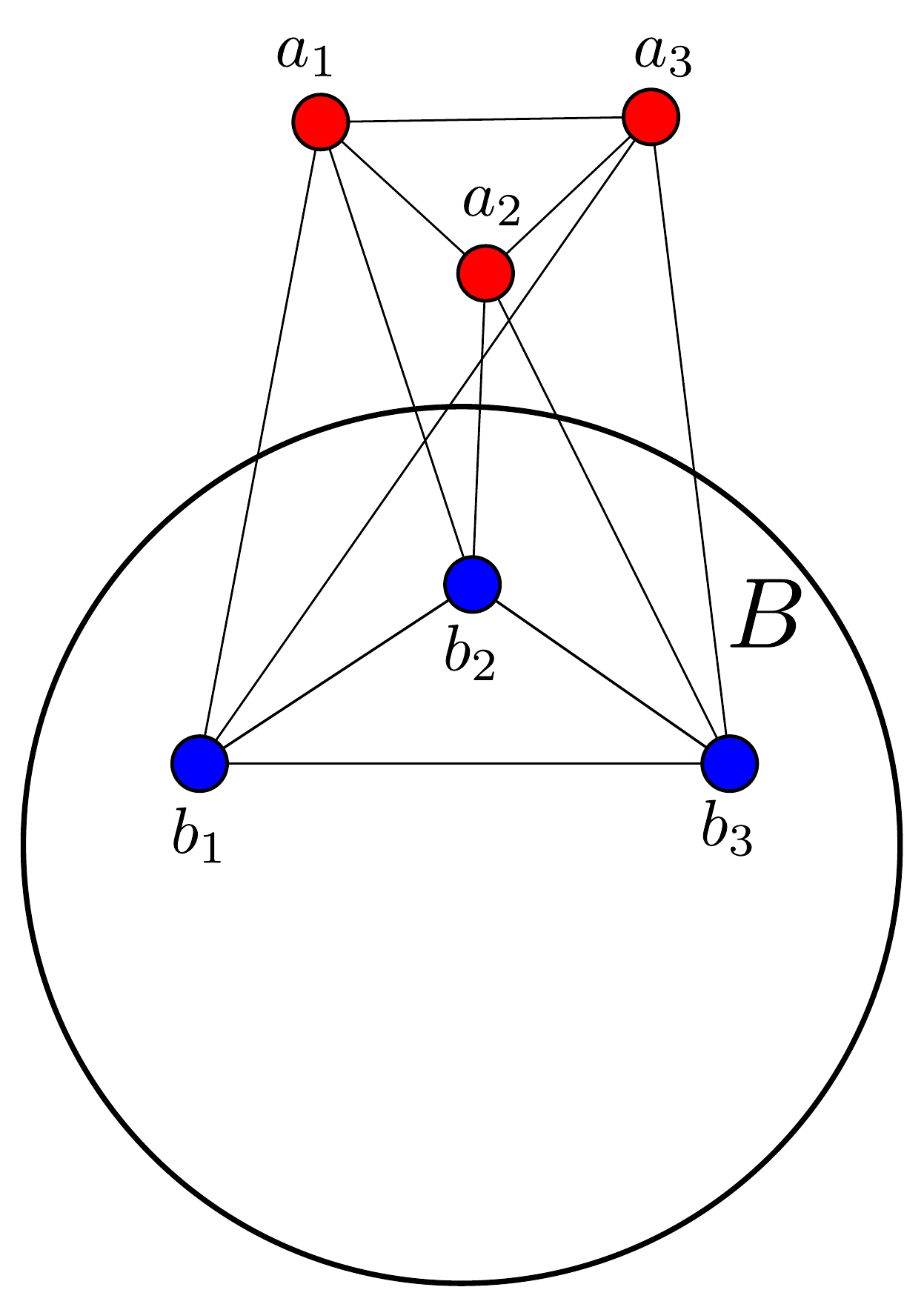}
			\hspace{1cm}
			\label{fig:bpolychromatic}
		}
	\qquad 
	\subfloat
		[Auxiliary graph $BS(b_1, b_2)$] {
			\includegraphics[scale=0.31]{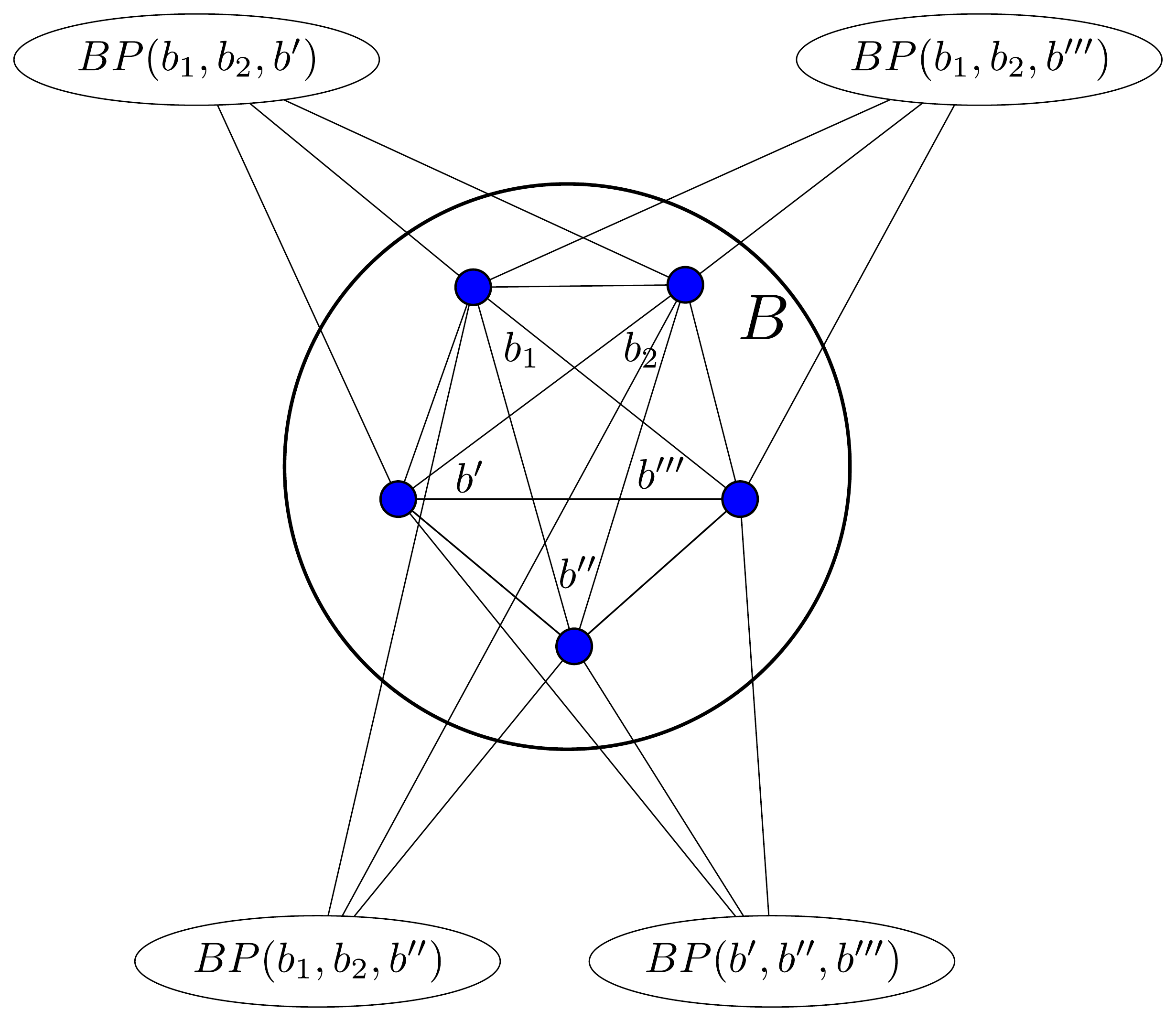}
			\label{fig:bswitch}
		}
	\caption{Auxiliary graphs $BP(b_1, b_2, b_3)$ and $BS(b_1, b_2)$}
\end{figure}


\begin{lemma}
\label{lem:auxiliarygraphbp}
Let $G$ be a weakly chordal graph (resp. (3, 1)-polar graph) and $b_1, b_2,
b_3 \in V(G)$ (resp. $b_1, b_2, b_3 \in B(G)$). If we add to $G$ a copy of an
auxiliary graph $BP(b_1, b_2, b_3)$, then the following assertions are true.

\begin{itemize}
  \item The resulting graph $G^\prime$ is weakly chordal (resp.
  (3, 1)-polar).
  \item If all cliques of $G$ have size at least 3, then all cliques of
  $G^\prime$ have size at least~3.
  \item Any 2-clique-colouring of $G^\prime$ assigns at least 2 colours to
  $b_1, b_2, b_3$.
  \item $G$ is 2-clique-colourable if $G^\prime$ is
  2-clique-colourable.
  \item $G^\prime$ is 2-clique-colourable if there exists a
  2-clique-colouring of $G$ that assigns at least 2 colours to $b_1, b_2, b_3$.
\end{itemize}
\end{lemma}
\begin{proof}
Let $G$ be a weakly chordal graph and $b_1, b_2, b_3 \in V(G)$. Add to
$G$ a copy of an auxiliary graph $BP(b_1, b_2, b_3)$ in order to obtain graph
$G^\prime$.

Suppose, by contradiction, that
$G^\prime$ has a chordless cycle $H$ with an odd number of vertices greater
than 4 or the complement $\overline{H}$ of a chordless cycle with an odd
number of vertices greater than 5. Clearly, $BP(b_1, b_2, b_3)$
is a weakly chordal graph. Since $G$ and $BP(b_1, b_2, b_3)$ are weakly chordal
graphs, $H$ and $\overline{H}$ contains a vertex of $BP(b_1, b_2, b_3)
\setminus G$ and a vertex of $G \setminus BP(b_1, b_2, b_3)$. Since $\{b_1, b_2,
b_3\}$ is a complete set that is a cutset of $G^\prime$ that disconnects $BP(b_1, b_2, b_3)
\setminus G$ from $G \setminus BP(b_1, b_2, b_3)$. Then, every
cycle with vertices of $BP(b_1, b_2, b_3) \setminus G$ and of $G \setminus
BP(b_1, b_2, b_3)$ contains a chord, i.e. there is no such $H$. Since
$a_1$, $a_2$, and $a_3$ have at most 3 neighbors, $|\overline{H}| \leq 6$.  If
$\overline{H}$ has only vertex $a_i$ of $BP(b_1, b_2, b_3) \setminus G$, then
$a_i$ has at most 2 neighbors in $\overline{H}$, which is a contradiction. If
$\overline{H}$ has only vertices $a_i$, $a_j$, $i \neq j$, of $BP(b_1, b_2,
b_3) \setminus G$, then $\overline{H}$ contains $\{b_1, b_2, b_3\}$, otherwise
$a_i$ or $a_j$ have at most 2 neighbors in $\overline{H}$. Let $u$ be a vertex
of $G \setminus BP(b_1, b_2, b_3)$ in $\overline{H}$. Since $a_i$ and $a_j$ are
not neighbors of $u$ and a vertex in $\overline{H}$ has at most two
non-neighbors in $\overline{H}$, $\{u, b_1, b_2, b_3\}$ is a complete set,
which is a contradiction. If $\overline{H}$ has all three vertices of $BP(b_1,
b_2, b_3) \setminus G$, then a vertex of $G \setminus BP(b_1, b_2, b_3)$ in $H$
has 3 non-neighbors in $\overline{H}$, which is a contradiction. Hence, there
is no such $\overline{H}$. Finally, $G^\prime$ is weakly chordal. If $G$ is a
(3, 1)-polar graph and $b_1, b_2, b_3 \in B(G)$, then~$G^\prime$ is a (3,
1)-polar graph with $A(G^\prime) = A(G) \cup \{a_1, a_2, a_3\}$ and
$B(G^\prime) = B(G)$ as the partition of $V(G^\prime)$ into two sets. Notice
that the added satellite is a triangle. Hence, $G^\prime$ is a (3, 1)-polar
graph.

Let $\mathcal{C}(G)$ be the set of cliques of graph~$G$. We have
$\mathcal{C}(G) \cap \mathcal{C}(G^\prime) = \mathcal{C}(G)$ and
$\mathcal{C}(G^\prime) \setminus \mathcal{C}(G) = \{\{a_1, b_1, b_2\}, \{a_2,
b_2, b_3\}, \{a_3, b_1, b_3\}, \{a_1, a_2, b_2\}, \{a_2, a_3, b_3\}, \{a_1, a_3,
b_1\}$, $\{a_1, a_2, a_3\}\}$. Clearly, if all cliques of $G$ have size at
least 3, then all cliques of $G^\prime$ have size at least~3.

Since $\{a_1, a_2, a_3\}$ is a clique of $G^\prime$, any 2-clique-colouring
$\pi^\prime$ of $G^\prime$ assigns at least 2 colours to $a_1, a_2, a_3$. Let
$i, j, k, \ell \in \{1, 2, 3\}$ and $\pi^\prime(a_i) \neq \pi^\prime(a_j)$.
Since $\{a_i, b_i, b_k\}$ (resp. $\{a_j, b_j, b_\ell\}$) is a clique of
$G^\prime$, $\pi^\prime$ assigns a colour which is not $\pi^\prime(a_i)$ to
$b_i$ or $b_k$ (resp. $\pi^\prime$ assigns a colour which is not
$\pi^\prime(a_j)$ to $b_j$ or $b_\ell$). Hence, $\pi^\prime$ assigns 2 distinct
colours to $b_1, b_2, b_3$.

Finally, $\pi^\prime(G)$ is a 2-clique-colouring of $G$, since $\mathcal{C}(G)
\subset \mathcal{C}(G^\prime)$. Now, consider a 2-clique-colouring of $G$ that
assigns 2 colours to $b_1, b_2, b_3$. It is easy to extend $\pi$ in order to
assign colours to the vertices of $BP(b_1, b_2, b_3) \setminus G$, such that
all cliques of $\mathcal{C}(G^\prime) \setminus \mathcal{C}(G)$ are
polychromatic.
\end{proof}

\begin{lemma}
\label{lem:auxiliarygraphbs}
Let $G$ be a weakly chordal graph (resp. (3, 1)-polar graph) and $b_1, b_2
\in V(G)$ (resp. $b_1, b_2 \in B(G)$). If we add to $G$ a copy of an auxiliary
graph $BS(b_1, b_2)$, then the following assertions are true.

\begin{itemize}
  \item The resulting graph $G^\prime$ is weakly chordal (resp. (3, 1)-polar).
  \item If all cliques of $G$ have size at least 3, then all cliques of
  $G^\prime$ have size at least~3.
  \item Any 2-clique-colouring of $G^\prime$ assigns 2 colours to
  $b_1$ and $b_2$.
  \item $G$ is 2-clique-colourable if $G^\prime$ is 2-clique-colourable.
  \item $G^\prime$ is 2-clique-colourable if there exists a
  2-clique-colouring of $G$ that assigns 2 colours to $b_1$ and $b_2$.
\end{itemize}
\end{lemma}
\begin{proof}
Let $G$ be a weakly chordal graph and $b_1, b_2 \in V(G)$. Add to
$G$ a copy of an auxiliary graph $BS(b_1, b_2)$ in order to obtain graph
$G^\prime$.

Suppose, by contradiction, that~$G^\prime$ has a chordless cycle $H$ with an
odd number of vertices greater than 4 or the complement $\overline{H}$ of a
chordless cycle with an odd number of vertices greater than 5. First, we prove
that $BS(b_1, b_2)$ is a weakly chordal (3, 1)-polar graph. A complete
graph $K_5$ with vertices $b_1, b_2, b^{\prime}, b^{\prime\prime},
b^{\prime\prime\prime}$ is a weakly chordal (3, 1)-polar graph with $A(K_5) =
\emptyset$ and $B(K_5) = \{b_1, b_2, b^{\prime}, b^{\prime\prime},
b^{\prime\prime\prime}\}$. By Lemma~\ref{lem:auxiliarygraphbp}, if we add
copies of the auxiliary graphs $BP(b_1, b_2, b^{\prime})$, $BP(b_1, b_2,
b^{\prime\prime})$, $BP(b_1, b_2, b^{\prime\prime\prime})$, and $BP(b^{\prime},
b^{\prime\prime}, b^{\prime\prime\prime})$, then we have a (weakly chordal) (3,
1)-polar graph that corresponds to $BS(b_1, b_2)$. Hence, $BS(b_1, b_2)$ is a
(weakly chordal) (3, 1)-polar graph. Since $G$ and $BS(b_1, b_2)$ are weakly
chordal graphs, $H$ and $\overline{H}$ contains a vertex of $BS(b_1, b_2)
\setminus G$ and a vertex of $G \setminus BS(b_1, b_2)$. Since $\{b_1, b_2,
b^{\prime}, b^{\prime\prime}, b^{\prime\prime\prime}\}$ is a complete set that
is a cutset of $G^\prime$ that disconnects $BS(b_1, b_2) \setminus G$ from
$G \setminus BS(b_1, b_2)$. Then, every cycle with vertices of $BS(b_1, b_2)
\setminus G$ and of $G \setminus BS(b_1, b_2)$ contains a chord, i.e
there is no such $H$. Subgraph $\overline{H}$ have vertices $b_1$ or
$b_2$, otherwise $\overline{H}$ is disconnected.
Hence, $\overline{H}$ has a 1-cutset or a 2-cutset, which is a
contradiction since the complement of a chordless cycle is triconnected.
Hence, $G^\prime$ is weakly chordal. If $G$ is a (3, 1)-polar graph and $b_1,
b_2 \in B(G)$, then $G^\prime$ is a (3, 1)-polar graph with $A(G^\prime) =
A(G) \cup (V(BS(b_1, b_2)) \setminus \{b_1, b_2, b^{\prime}, b^{\prime\prime},
b^{\prime\prime\prime}\})$ and $B(G^\prime) = B(G) \cup
\{b^{\prime}, b^{\prime\prime}, b^{\prime\prime\prime}\}$ as the
partition of $V(G^\prime)$ into two sets. Notice that all added satellites are
triangles. Hence, $G^\prime$ is a (3, 1)-polar graph.

Let $\mathcal{C}(G)$ be the set of cliques of graph $G$. If $B(G) =\emptyset$,
then we have $\mathcal{C}(G) \cap \mathcal{C}(G^\prime) = \mathcal{C}(G)
\setminus \{b_1, b_2\}$ and $\mathcal{C}(G^\prime) \setminus \mathcal{C}(G)$ is
precisely $\{b_1, b_2, b^{\prime}, b^{\prime\prime}, b^{\prime\prime\prime}\}$,
and all cliques added by the inclusion of copies of the auxiliary graphs
$BP(b_1, b_2, b^{\prime})$, $BP(b_1, b_2, b^{\prime\prime})$, $BP(b_1, b_2,
b^{\prime\prime\prime})$, and $BP(b^{\prime}, b^{\prime\prime},
b^{\prime\prime\prime})$. Otherwise, i.e.
$B(G) \neq \emptyset$, then we have $\mathcal{C}(G) \cap \mathcal{C}(G^\prime) =
\mathcal{C}(G) \setminus \{B(G)\}$ and $\mathcal{C}(G^\prime) \setminus
\mathcal{C}(G)$ is precisely $B(G^\prime)$, and all cliques added by the
inclusion of copies of the auxiliary graphs $BP(b_1, b_2, b^{\prime})$, $BP(b_1,
b_2, b^{\prime\prime})$, $BP(b_1, b_2, b^{\prime\prime\prime})$, and
$BP(b^{\prime}, b^{\prime\prime}, b^{\prime\prime\prime})$. By
Lemma~\ref{lem:auxiliarygraphbp}, all cliques added by the auxiliary graphs
have size at least 3. Then, all cliques of $G^\prime$ have size at least~3.

Consider any 2-clique-colouring $\pi^\prime$ of $G^\prime$. Since we added a
copy of the auxiliary graph $BP(b^{\prime}, b^{\prime\prime}, b^{\prime\prime\prime})$,
Lemma~\ref{lem:auxiliarygraphbp} states that $\pi^\prime$ assigns at least
2 colours to $b^{\prime}, b^{\prime\prime}, b^{\prime\prime\prime}$. Without
loss of generality, suppose that $\pi^\prime$ assigns distinct colours to $b^{\prime}$
and $b^{\prime\prime}$. Since we added copies of the auxiliary graphs
$BP(b^{\prime}, b_1, b_2)$ and $BP(b^{\prime\prime}, b_1, b_2)$,
Lemma~\ref{lem:auxiliarygraphbp} states that $\pi^\prime$ assigns at least 2 colours
to $\{b^{\prime}, b_1, b_2\}$ and at least 2 colours to $\{b^{\prime\prime},
b_1, b_2\}$, i.e. $\pi^\prime$ assigns a colour which is not $\pi^\prime(b^{\prime})$ to
$b_1$ or $b_2$ and a colour which is not $\pi^\prime(b^{\prime\prime})$ to $b_1$ or
$b_2$. Hence, $\pi^\prime$ assigns 2 distinct colours to $b_1, b_2$.

If $B(G) = \emptyset$, then $\pi^\prime(G)$ is a 2-clique-colouring of $G$, since
$\mathcal{C}(G) \setminus \mathcal{C}(G^\prime) = \mathcal{C}(G) \setminus
\{b_1, b_2\}$ and $\pi^\prime$ assigns distinct colours to $b_1, b_2$.
Otherwise, i.e. $B(G) \neq \emptyset$, $\pi^\prime(G)$ is a 2-clique-colouring
of $G$ since $\pi^\prime$ assigns at least 2 colours to $\{b_1, b_2\} \subset
B(G)$ and to every clique of $\mathcal{C}(G) \cap \mathcal{C}(G^\prime) =
\mathcal{C}(G) \setminus B(G)$.

Now, consider a 2-clique-colouring of $G$ that assigns 2 colours to $b_1,
b_2$. Assign the same colour of $b_1$ to $b^{\prime}$ and $b^{\prime\prime}$.
Assign the same colour of $b_2$ to $b^{\prime\prime\prime}$. 
The sets $\{b_1, b_2, b^{\prime}\}$, $\{b_1, b_2, b^{\prime\prime}\}$, $\{b_1,
b_2, b^{\prime\prime\prime}\}$, and $\{b^{\prime}, b^{\prime\prime},
b^{\prime\prime\prime}\}$ have 2 colours each. By
Lemma~\ref{lem:auxiliarygraphbp}, all cliques added by the
inclusion of copies of the auxiliary graphs $BP(b_1, b_2, b^{\prime})$, $BP(b_1,
b_2, b^{\prime\prime})$, $BP(b_1, b_2, b^{\prime\prime\prime})$, and
$BP(b^{\prime}, b^{\prime\prime}, b^{\prime\prime\prime})$ are polychromatic. 
Hence, we have a 2-clique-colouring of $G^\prime$.
\end{proof}

We strengthen the result that 2-clique-colouring of (3, 1)-polar graphs is
$\mathcal{NP}$-complete, now even restricting all cliques to have size at
least~3, which gives a subclass of weakly chordal graphs.

\begin{theorem}
\label{thm:2cc31polarcliques}
The problem of 2-clique-colouring is $\mathcal{NP}$-complete for (weakly
chordal) (3, 1)-polar graphs with all cliques having size at least~3.
\end{theorem} 
\begin{proof}
The problem of 2-clique-colouring a (3, 1)-polar graph with all cliques having
size at least 3 is in $\mathcal{NP}$: Theorem~\ref{thm:nptocheck} confirms that
to check whether a colouring of a (3, 1)-polar graph is a 2-clique-colouring is
in $\mathcal{P}$.

We prove that 2-clique-colouring (3, 1)-polar graphs with all cliques having
size at least 3 is $\mathcal{NP}$-hard by reducing {\sc NAE-SAT} to it.
The outline of the proof follows. For every formula~$\phi$, a (3, 1)-polar graph
$G$ with all cliques having size at least 3 is constructed such that $\phi$ is
satisfiable if, and only if, graph $G$ is 2-clique-colourable. Let $n$ (resp.
$m$) be the number of variables (resp. clauses) in formula $\phi$. We define
graph $G$ as follows.

\begin{itemize}
	\item for each variable $x_i$, $1 \leq i \leq n$, we create two vertices
	$x_i$ and $\overline{x}_i$. Moreover, we create edges so that the set $\{x_1,
	\overline{x}_1, \ldots, x_n, \overline{x}_n\}$ induces a complete subgraph
	of~$G$.
	\item for each variable $x_i$, $1 \leq i \leq n$, add a copy of the auxiliary
	graph $BS(x_i, \overline{x}_i)$. Vertices $x_i$ and $\overline{x}_i$
	correspond to the literals of variable~$x_i$.
	\item for each clause $c_j = (l_a, l_b, l_c)$, $1 \leq j \leq m$, we add a copy
	of the auxiliary graph $BP(l_a, l_b, l_c)$.	 
\end{itemize}

Refer to Fig.~\ref{fig:reduction-np-cliques} for an example of such construction, given
a formula $\phi = (x_{1} \vee \overline{x_{2}} \vee x_{4}) \wedge ( x_{2} \vee
\overline{x_{3}} \vee \overline{x_{5}} )\wedge ( x_{1} \vee x_{3} \vee x_{5})$.

\begin{figure}[t!]
\centering
	\includegraphics[scale=0.4]{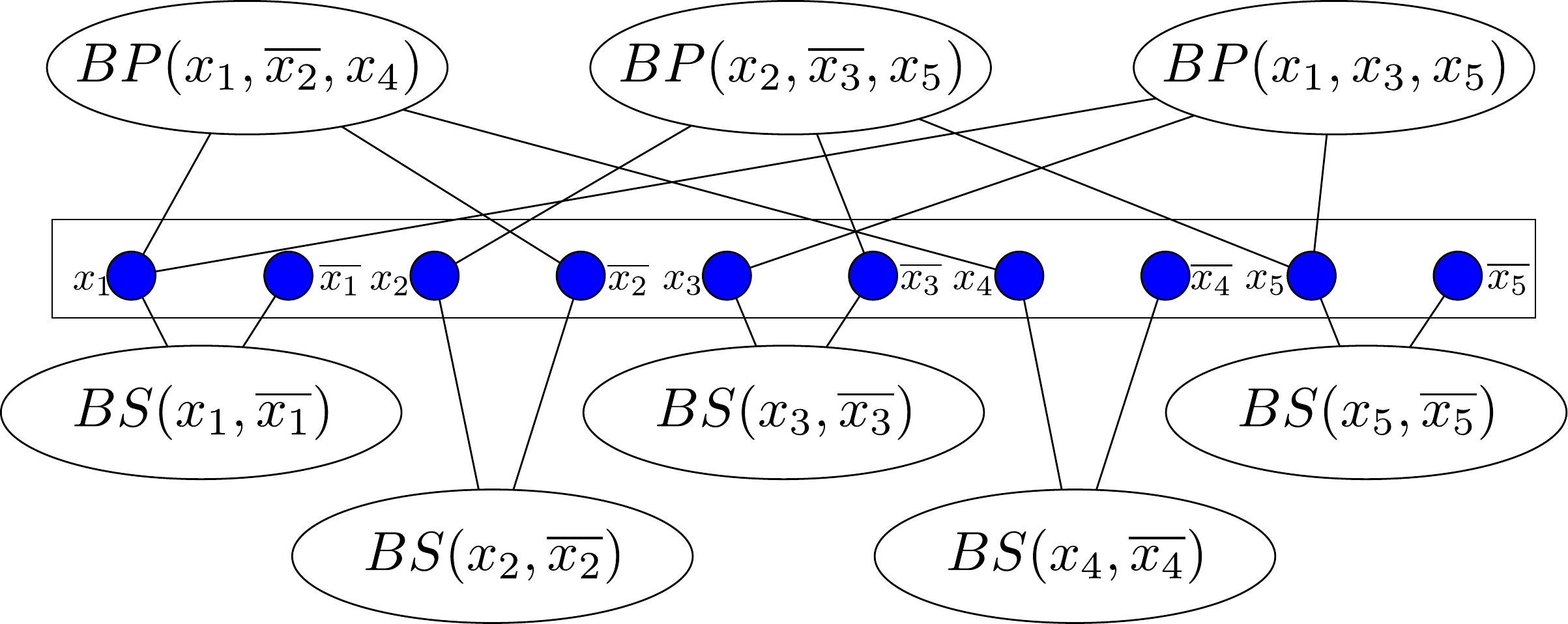}
	\caption{Example of a (3, 1)-polar graph with all cliques having size at least
	3 constructed for a NAE-SAT instance $\phi = (x_{1} \vee \overline{x_{2}} \vee
	x_{4}) \wedge ( x_{2} \vee \overline{x_{3}} \vee \overline{x_{5}} )\wedge (
	x_{1} \vee x_{3} \vee x_{5})$}
	\label{fig:reduction-np-cliques}
\end{figure}

First, we prove that the graph $G$ is a (3, 1)-polar graph with all cliques
having size at least~3.

Consider the set $\{x_1, \overline{x}_1, \ldots, x_n, \overline{x}_n\}$.
Clearly, this set is a clique with size at least 3 and also a (3, 1)-polar
graph. Lemma~\ref{lem:auxiliarygraphbs} states that, for each added auxiliary
graph $BS(x_i, \overline{x}_i)$ to a (3, 1)-polar graph with
all cliques having size at least 3, every obtained graph remains in the class.
Lemma~\ref{lem:auxiliarygraphbp} states that, for each added auxiliary graph
$BP(l_{a_{c_j}}, l_{b_{c_j}}, l_{c_{c_j}})$ to a (3, 1)-polar
graph with all cliques having size at least 3, every obtained graph remains in
the class. Hence, $G$ is a (3, 1)-polar graph with all cliques
having size at least 3.

Such construction is done in polynomial-time. One can check with
Lemmas~\ref{lem:auxiliarygraphbp} and~\ref{lem:auxiliarygraphbs} that~$G$ has
$3m + 17n$ vertices.

We claim that formula $\phi$ is satisfiable if, and only if, there exists a 
2-clique-colouring of~$G$. Assume there exists a valuation $v_\phi$ such that
$\phi$ is satisfied. We give a colouring to graph $G$, as follows.
 
\begin{itemize}
	\item assign colour 1 to $l \in \{x_1, \overline{x}_1, \ldots, x_n,
	\overline{x}_n\}$ if it corresponds to the literal which receives the $true$
	value in $v_\phi$, otherwise we assign colour 2 to it.
	
	\item extend the 2-clique-colouring to the copy of the auxiliary graph
	$BS(x_i, \overline{x}_i)$, for each variable $x_i$, $1 \leq i \leq n$,
	according to Lemma~\ref{lem:auxiliarygraphbs}. Notice that the necessary
	condition to extend the 2-clique-colouring is satisfied.
	
	\item extend the 2-clique-colouring to the copy of the auxiliary graph
	$BP(l_{a}, l_{b}, l_{c})$, for each triangle $c = \{l_{a}, l_{b}, l_{c}\}$, $1
	\leq j \leq m$, according to Lemma~\ref{lem:auxiliarygraphbp}. Notice that the
	necessary condition to extend the 2-clique-colouring is satisfied.
\end{itemize}

It still remains to be proved that this is indeed a 2-clique-colouring.

Consider the set $\{x_1, \overline{x}_1, \ldots, x_n, \overline{x}_n\}$.
Clearly, the above colouring assigns 2 colours to this set.
Lemma~\ref{lem:auxiliarygraphbs} states that, for each added auxiliary graph
$BS(x_i, \overline{x}_i)$ to a 2-clique-colourable weakly chordal (3, 1)-polar
graph, we obtain a 2-clique-colourable graph. Lemma~\ref{lem:auxiliarygraphbp}
states that, for each added auxiliary graph $BP(l_{a}, l_{b}, l_{c})$  to a
2-clique-colourable weakly chordal (3, 1)-polar graph, we obtain a
2-clique-colourable graph. Hence, graph $G$ is 2-clique-colourable.

For the converse, we now assume that $G$ is 2-clique-colourable and we consider
any 2-clique-colouring. Recall that the vertices $x_i$ and $\overline{x}_i$ have
distinct colours, since we added the auxiliary graph $BS(x_i, \overline{x}_i)$,
for each variable $x_i$. Hence, we define $v_\phi$ as
follows. The literal $x_i$ is assigned $true$ in $v_\phi$ if the corresponding
vertex has colour 1 in the clique-colouring, otherwise it is assigned $false$.
Since we are considering a 2-clique-colouring, every triangle (clique) $c_j$, $1
\leq j \leq m$, is polychromatic. As a consequence, there exists at least one literal
with $true$ value in $c_j$ and at least one literal with $false$ value in every
clause $c_j$. This proves that $\phi$ is satisfied for valuation~$v_\phi$.
\end{proof}

As an remark, a shorter alternative proof that 2-clique-colouring is
$\mathcal{NP}$-complete for weakly chordal (3, 1)-polar graphs with all cliques
having size at least 3 can be obtained by a reduction from {\sc Positive
NAE-SAT}. The alternative proof follows analogously to alternative proof of
Theorem~\ref{thm:2cc31polar}.

On the other hand, we prove that 2-clique-colouring (2, 1)-polar graphs becomes
polynomial when all cliques have size at least 3. 

\begin{theorem}
\label{thm:2cc21polarcliques}
The problem of 2-clique-colouring is polynomial for
(2, 1)-polar graphs with all cliques having size at least~3.
\end{theorem} 
\begin{proof}
Let $S$ be a satellite of $G$. First, $S$ is an edge, since $G$ is a (2,
1)-polar graph. Second, every clique of $G$ has size at least~3. Then, there is
a vertex of partition $B$ that is a neighbor of both vertices of $S$. This
implies that every satellite of $G$ is in case $\kc$. Notice that
Algorithm~\ref{alg:reducion-np-2} outputs a complete graph, when $G$ is given
as input. Hence, $G$ is 2-clique-colourable (see Theorem~\ref{thm:2cc21polar}).

The polynomial-time algorithm to give a 2-clique-colouring for (2, 1)-polar
graphs with all cliques having size at least~3 follows. Give any 2-colouring to
the vertices of $B$. For each satellite $S$, assign colour 1 to a vertex of $S$,
if it has a neighbor in $B$ with colour 2, otherwise assign colour 2. It is easy
to check that it is a 2-clique-colouring of $G$, since every satellite of $G$
is in case $\kc$.
\end{proof}

In the proof that 2-clique-colouring weakly chordal graphs is a
$\Sigma_2^P$-complete problem (Theorem~\ref{thm:weaklychordal}), we constructed
a weakly chordal graph with $K_2$ cliques to force distinct colours in their
extremities (in a 2-clique-colouring).
We can obtain a weakly chordal graph with no cliques of size~2 by adding
copies of the auxiliary graph $BS(u, v)$, for every $K_2$ clique $\{u, v\}$.
Auxiliary graphs $AK$ and $NAS$ become $AK^\prime$ and $NAS^\prime$,
both depicted in Fig.~\ref{fig:auxiliarygraphscliques3}. 

\begin{figure}[t!]
\centering
	\subfloat
		[$AK^\prime(a, g)$]
		{
			\includegraphics[scale=0.4]{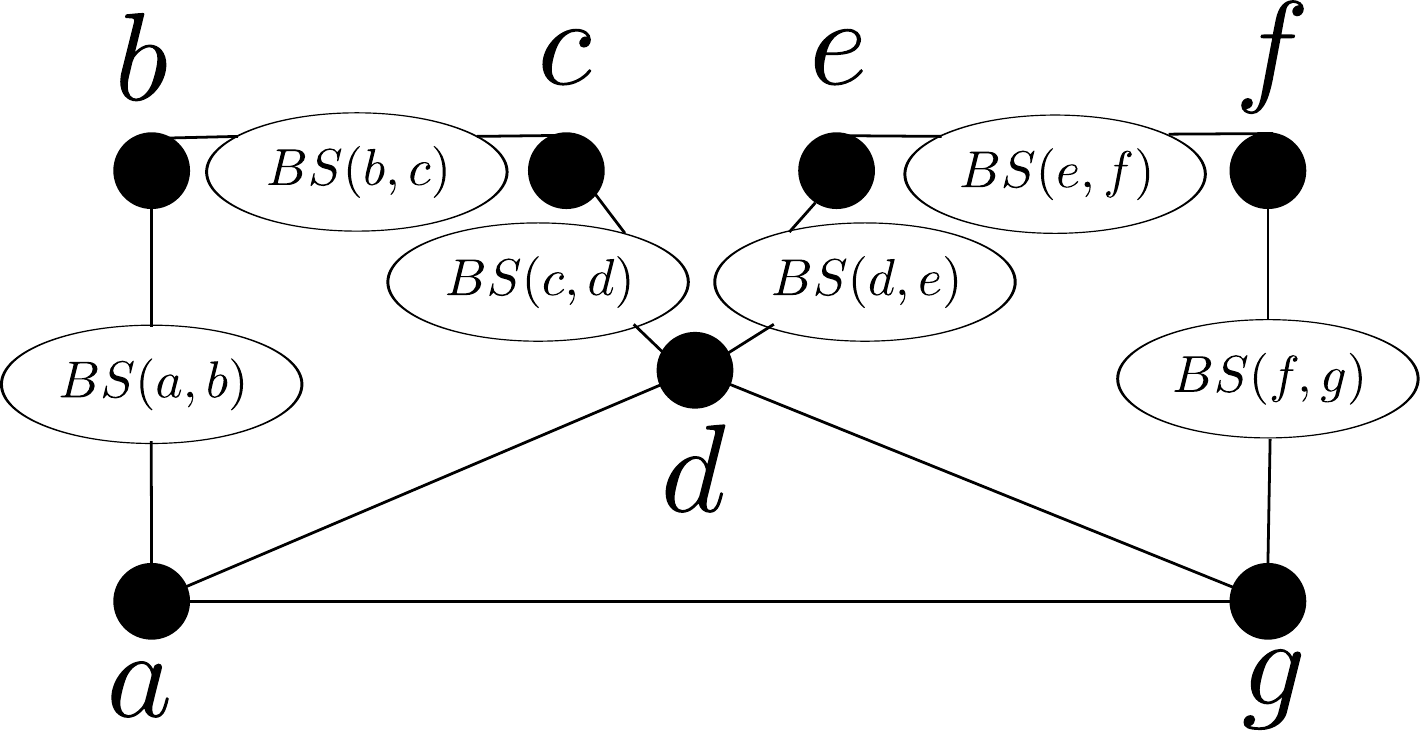}
			\label{fig:adjacentkeepercliques3}
		}
	\qquad
	\subfloat
		[$NAS^\prime(a, j)$] 
		{
			\includegraphics[scale=0.4]{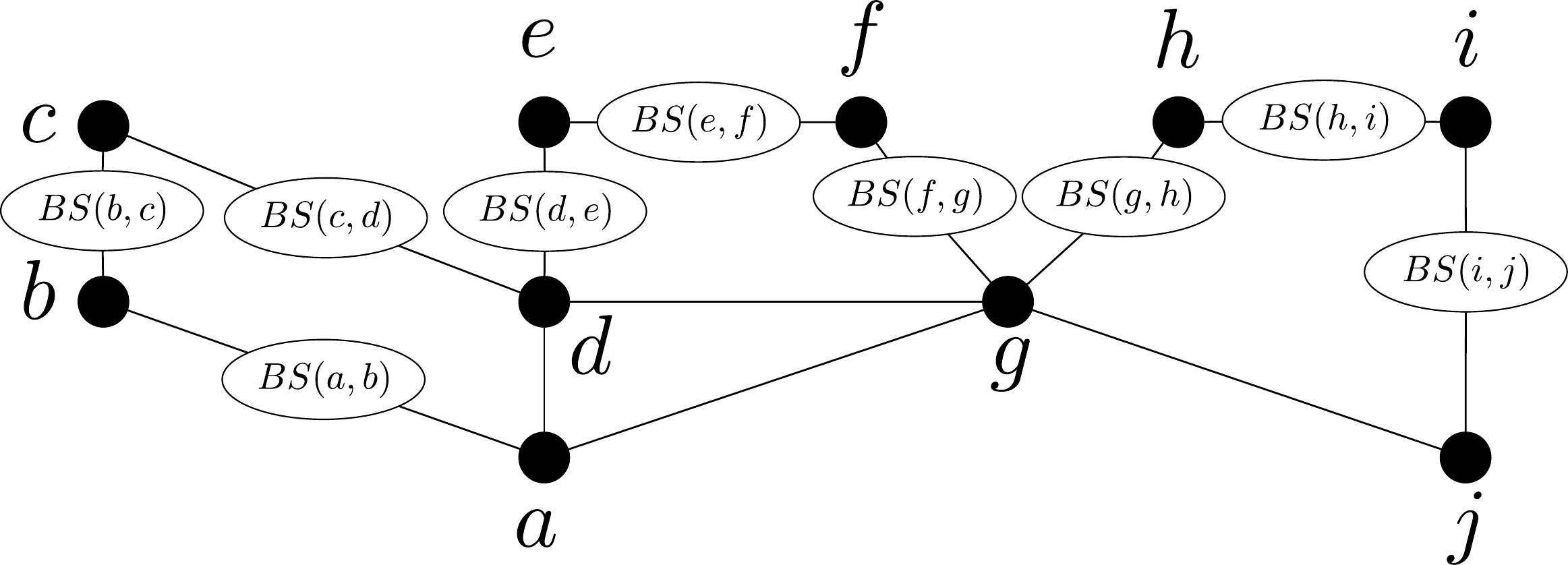}
			\label{fig:nonadjacentswitchercliques3}
		}
	\caption{Auxiliary graphs $AK^\prime(a, g)$ and $NAS^\prime(a, j)$}
	\label{fig:auxiliarygraphscliques3}
\end{figure}

Finally, the
weakly chordal graph constructed in Theorem~\ref{thm:weaklychordal} becomes a
weakly chordal graph with no $K_2$ clique, depicted in
Fig.~\ref{fig:reduction-kratochvil-uncoloured}. 

\begin{figure}[t!]
\centering
	\includegraphics[scale=0.42]{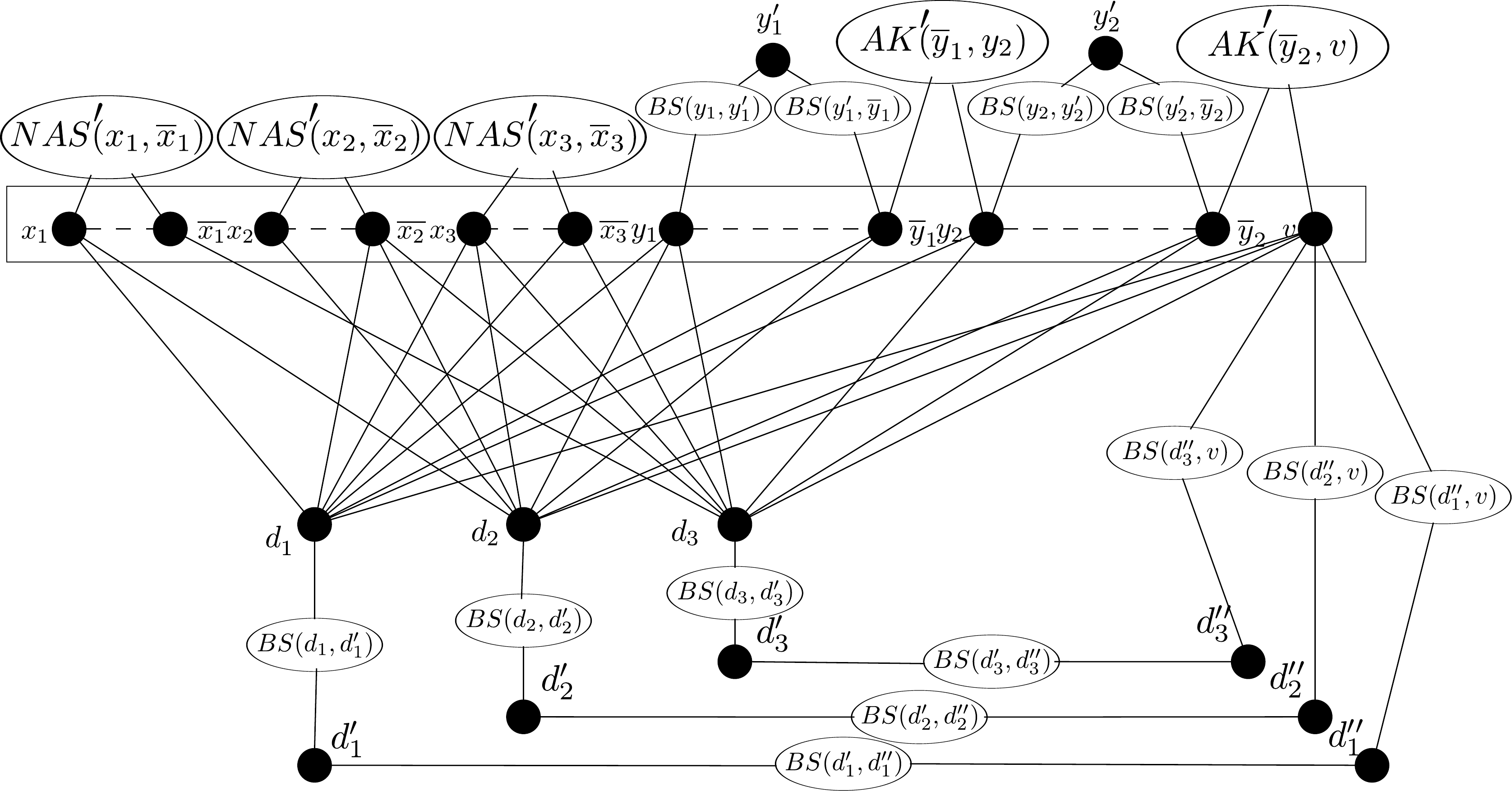}
	\caption{Graph constructed for a QSAT2 instance $\Psi = (x_1 \wedge
	\overline{x}_2 \wedge y_2) \vee (x_1 \wedge x_3 \wedge \overline{y}_2) \vee
	( \overline{x}_1 \wedge \overline{x}_2 \wedge y_1)$}
	\label{fig:reduction-kratochvil-uncoloured}
\end{figure}

Such construction is done in
polynomial-time. Notice that, in the constructed graph of
Theorem~\ref{thm:weaklychordal}, every $K_2$ clique $\{u, v\}$ has 2 distinct
colours in a clique-colouring. Hence, one can check with
Lemmas~\ref{lem:auxiliarygraphbp} and~\ref{lem:auxiliarygraphbs} that the
obtained graph is weakly chordal and it is 2-clique-colourable if, and only if,
the constructed graph of Theorem~\ref{thm:weaklychordal} is
2-clique-colourable. This implies the following theorem.

\begin{theorem}
\label{thm:weaklychordalkratochvil}
The problem of 2-clique-colouring is $\Sigma_2^P$-complete for weakly chordal
graphs with all cliques having size at least 3.
\end{theorem}

As a direct consequence of Theorem~\ref{thm:weaklychordalkratochvil}, we have
that 2-clique-colouring is $\Sigma_2^P$-complete for perfect graphs with all
cliques having size at least~3.

\begin{corollary}
\label{cor:perfectkratochvil}
The problem of 2-clique-colouring is $\Sigma_2^P$-complete for perfect
graphs with all cliques having size at least 3.
\end{corollary}


\section{Final considerations}
\label{sec:final}

Marx~\cite{Marx2011} proved complexity results for $k$-clique-colouring, for
fixed $k \geq 2$, and related problems that lie in between two distinct complexity
classes, namely $\Sigma_2^P$-complete and $\Pi_3^P$-complete. Marx approaches
the complexity of clique-colouring by fixing the graph class and diversifying
the problem. In the present work, our point of view is the opposite: we rather
fix the (2-clique-colouring) problem and we classify the problem complexity according
to the inputted graph class, which belongs to nested subclasses of weakly
chordal graphs. We achieved complexities lying in between three distinct
complexity classes, namely $\Sigma_2^P$-complete, $\mathcal{NP}$-complete and
$\mathcal{P}$. Fig.~\ref{fig:hierarquiacomplexidade} shows the relation of
inclusion among the classes of graphs of Table~\ref{t:tabela}.
The 2-clique-colouring complexity for each class is highlighted.

\begin{figure}[t!]
\centering
		\includegraphics[scale=0.22]{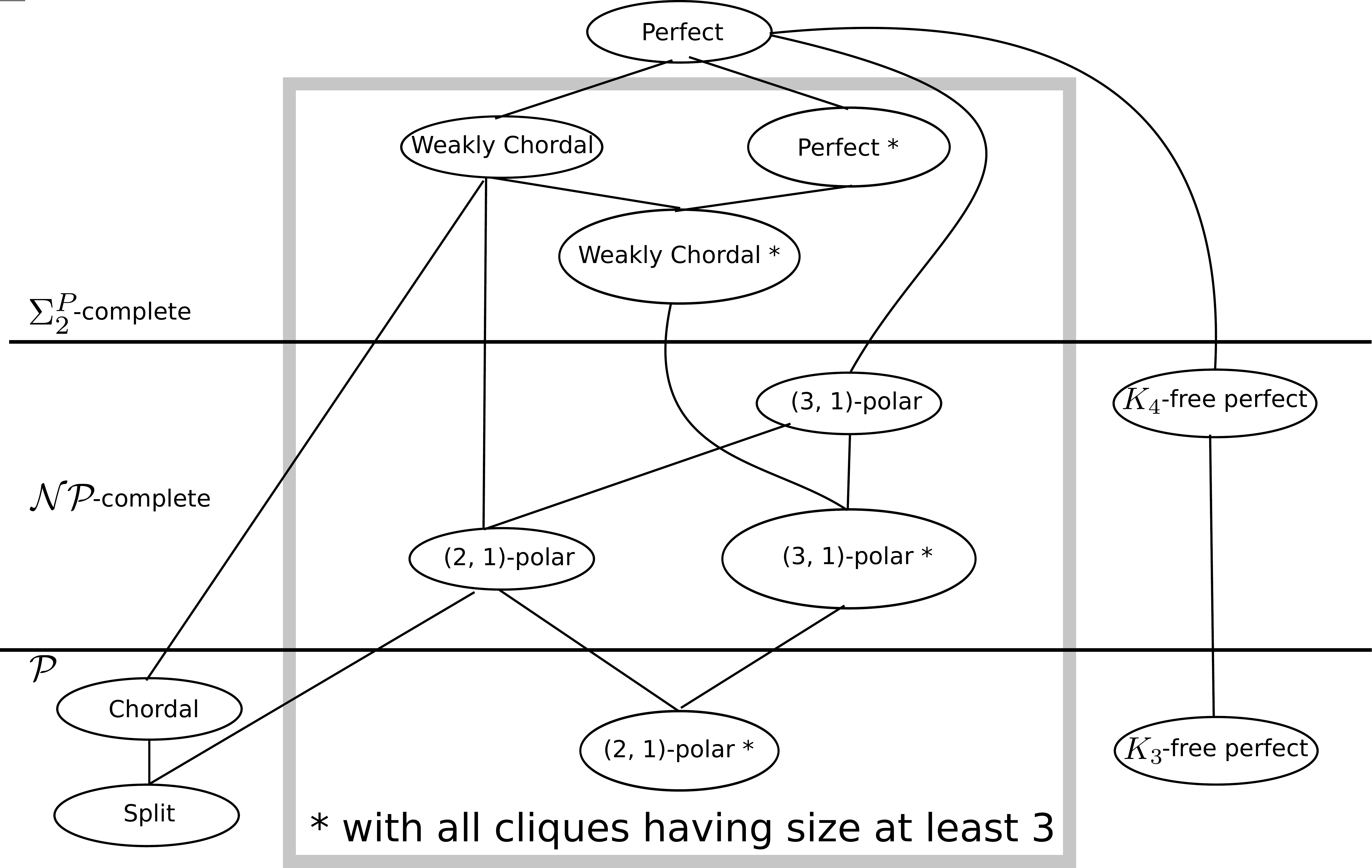}
	\caption{2-clique-colouring complexity of perfect graphs and subclasses.}
	\label{fig:hierarquiacomplexidade}
\end{figure}

Notice that the perfect graph subclasses for which the
2-clique-colouring problem is in $\mathcal{NP}$ mentioned so far in the present
work satisfy that the number of cliques is polynomial. We remark that the
complement of a matching has an exponential number of cliques and yet the
2-clique-colouring problem is in $\mathcal{NP}$, since no such graph is
2-clique-colourable.
Now, notice that the perfect graph subclasses for which the 2-clique-colouring problem is in $\mathcal{P}$
mentioned so far in the present work satisfy that all graphs in the class are
2-clique-colourable. Mac\^edo Filho et al.~\cite{latin2012} have proved that
unichord-free graphs are 3-clique-colourable, but a unichord-free graph is
2-clique-colourable if and only if it is perfect. As a future work, we aim to
find subclasses of perfect graphs where not all graphs are 2-clique-colourable
and yet the 2-clique-colouring problem is in $\mathcal{P}$ when restricted to
the class.

\section*{Acknowledgments}

We are grateful to Jayme Szwarcfiter for introducing us the class of $(\alpha, \beta)$-polar graphs.

\bibliographystyle{plainnat}
\bibliography{manuscript-hierarchy-full}

\end{document}